\documentclass[11pt]{article}
\usepackage[letterpaper,margin=1.00in]{geometry}
\usepackage{amsmath, amssymb, amsthm, amsfonts}
\usepackage{bbm}
\usepackage{amscd}
\usepackage{mathrsfs}

\usepackage{comment} 
\usepackage{ifthen}
\usepackage{tikz}
\usetikzlibrary{positioning,decorations.pathreplacing}
\usepackage{ bbold }
\usepackage{graphicx}
\usepackage{color}
\usepackage{algorithm}
\usepackage[noend]{algpseudocode}
\usepackage{epstopdf}
\usepackage{wrapfig}
\usepackage{paralist}
\usepackage{wasysym}
\usepackage[textsize=tiny]{todonotes}

\usepackage{listings} 

\usepackage{pdflscape} 

\usepackage[most]{tcolorbox}
\tcbset{colback=gray!10, colframe=gray, boxrule=0.5pt, arc=2pt, left=4pt, right=4pt, top=4pt, bottom=4pt}

\usepackage{framed}
\usepackage[framemethod=tikz]{mdframed}
\usepackage[bottom]{footmisc}
\usepackage{enumitem}
\setitemize{noitemsep,topsep=3pt,parsep=3pt,partopsep=3pt}
\usepackage[font=small]{caption}
\usepackage{xspace}

\usepackage{thmtools} 
\usepackage{thm-restate} 

\usepackage{hyperref}
\hypersetup{
    unicode=false,          
    colorlinks=true,        
    linkcolor=red,          
    citecolor=darkgreen,        
    filecolor=magenta,      
    urlcolor=cyan           
}

\newtheorem{theorem}{Theorem}[section]
\newtheorem{lemma}[theorem]{Lemma}
\newtheorem{meta-theorem}[theorem]{Meta-Theorem}
\newtheorem{claim}[theorem]{Claim}

\usepackage[capitalize, nameinlink,noabbrev]{cleveref}

\crefname{theorem}{Theorem}{Theorems}
\crefname{proposition}{Proposition}{Propositions}
\crefname{observation}{Observation}{Observations}
\crefname{lemma}{Lemma}{Lemmas}
\crefname{claim}{Claim}{Claims}
\crefname{problem}{Problem}{Problems}
\crefname{conjecture}{Conjecture}{Conjectures}
\crefname{question}{Question}{Questions}
\crefname{example}{Example}{Examples}
\crefname{fact}{Fact}{Facts}

\definecolor{darkgreen}{rgb}{0,0.5,0}

\usepackage{algcompatible}
\algnewcommand\algorithmicswitch{\textbf{switch}}
\algnewcommand\algorithmiccase{\textbf{case}}

\algdef{SE}[SWITCH]{Switch}{EndSwitch}[1]{\algorithmicswitch\ #1\ \algorithmicdo}{\algorithmicend\ \algorithmicswitch}%
\algdef{SE}[CASE]{Case}{EndCase}[1]{\algorithmiccase\ #1}{\algorithmicend\ \algorithmiccase}%
\algtext*{EndSwitch}%
\algtext*{EndCase}%

\newcommand{\eps}{\varepsilon}

\newcommand{\local}{$\mathsf{LOCAL}\xspace$ }

\renewcommand{\P}{\textrm{P}}

\newcommand{\poly}{\operatorname{poly}}

\renewcommand{\phi}{\varphi}

\newcommand{\E}{\mathbb{E}}
\renewcommand{\Pr}{\P}

\renewcommand{\paragraph}[1]{\vspace{0.15cm}\noindent {\bf #1}:}

\newcommand{\FullOrShort}{full}
\ifthenelse{\equal{\FullOrShort}{full}}{
  
  \newcommand{\fullOnly}[1]{#1}
  \newcommand{\shortOnly}[1]{}

  }{

    \newcommand{\fullOnly}[1]{}
    \newcommand{\IncludePictures}[1]{}
   
  }

\renewcommand{\Pr}{\P}
\usepackage{appendix}
\title{Towards True Work-Efficiency in Parallel Derandomization: \\ MIS, Maximal Matching, and Hitting Set}

\begin{document}
\date{}
\author{Mohsen Ghaffari \\ \small MIT \\ \small ghaffari@mit.edu \and Christoph Grunau \\ \small ETH Zurich \\ \small cgrunau@ethz.ch }
\maketitle

\begin{abstract} Derandomization is one of the classic topics studied in the theory of parallel computations, dating back to the early 1980s. Despite much work, all known techniques lead to deterministic algorithms that are not work-efficient. For instance, for the well-studied problem of maximal independent set---e.g., [Karp, Wigderson STOC'84; Luby STOC' 85; Luby FOCS'88]---state-of-the-art deterministic algorithms require at least $m \cdot \poly(\log n)$ work, where $m$ and $n$ denote the number of edges and vertices. Hence, these deterministic algorithms will remain slower than their trivial sequential counterparts unless we have at least $\poly(\log n)$ processors. 

In this paper, we present a generic parallel derandomization technique that moves exponentially closer to work-efficiency. 
The method iteratively rounds fractional solutions representing the randomized assignments to integral solutions that provide deterministic assignments, while maintaining certain linear or quadratic objective functions, and in an \textit{essentially work-efficient} manner. As example end-results, we use this technique to obtain deterministic algorithms with $m \cdot \poly(\log \log n)$ work and $\poly(\log n)$ depth for problems such as maximal independent set, maximal matching, and hitting set. 
\end{abstract}

   \thispagestyle{empty}

{   \newpage
    \hypersetup{linkcolor=blue}
    \tableofcontents
    \setcounter{page}{0}
    \thispagestyle{empty}
}

\newpage
\setcounter{page}{1}
\section{Introduction and Related Work}
All known techniques for derandomizing parallel algorithms come with the drawback of increasing the work by a $\poly(\log n)$ factor. For instance, for the well-studied problem of maximal independent set, which was the main target of classic works by Karp and Wigderson~\cite{karp1984fast}, Luby~\cite{luby1985simple,luby1988removing}, and Alon, Babai, and Itai~\cite{alon86}, the currently best known deterministic parallel algorithms require at least $m \poly(\log n)$ work. Hence, unless we have at least $\poly(\log n)$ processors, these parallel algorithms are embarrassingly slower than their trivial sequential counterparts. This paper is centered on addressing this limitation. We present a novel, generic parallel derandomization technique via gradually rounding fractional solutions. We show its power by using it to obtain deterministic parallel algorithms---for maximal independent set, maximal matching, and hitting set---with work $m \poly(\log\log n)$, thus moving exponentially closer to true work efficiency. We next review the relevant context and state of the art, and then state our results.

\subsection{Context}
\paragraph{Parallel Model---Work, Depth, and Work-Efficiency} We follow the standard \textit{work-depth} model~\cite{jaja1992introduction, blelloch1996programming}, where the algorithm runs on $p$ processors with write and read access to a shared memory. In the case of multiple concurrent writes to a memory location, we assume that an arbitrary one takes place. In any algorithm $\mathcal{A}$, its depth $D(\mathcal{A})$ is the longest chain of computational steps in $\mathcal{A}$, each of which depends on the previous ones. In other words, this is the time that it would take the algorithm to run even if we were given an infinite number of processors. The work $W(\mathcal{A})$ is the total number of computational steps in $\mathcal{A}$. The time $T_{p}(\mathcal{A})$ to run the algorithm on $p$ processors clearly satisfies $T_{p}(\mathcal{A})\geq \max\{D(\mathcal{A}), W(\mathcal{A})/p\}$. By Brent's principle~\cite{brent1974parallel}, we also know that $T_{p}(\mathcal{A}) \leq D(\mathcal{A})+ W(\mathcal{A})/p$. 

The objective in parallel computations is to devise algorithms that run faster than their sequential counterparts, and ideally, we want a speed-up proportional to the number of processors $p$. In particular, this requires the parallel algorithm to have a work bound equal/close to the best known sequential algorithm. Algorithms whose work bound matches the sequential counterpart are called \textit{work-efficient}. A relaxation is \textit{nearly work-efficient} algorithms, whose work bound can be a $\poly(\log n)$ factor larger. These algorithms need at least $\poly(\log n)$ processors to match (or beat) the speed of single-processor computation.

We comment that much of the early work in the theory of parallel computation, in the early 1980s and before, prioritized the depth and allowed large polynomial work bounds. These results assumed a large polynomial number of processors. This was in line with the general state of algorithmic research at the time, focusing on polynomial time computation with less emphasis on the exact polynomial. However, starting in the late 1980s and certainly by the 1990s~\cite{jaja1992introduction}, the focus moved more and more on the work bound. The aim is to attain parallel algorithms that yield speed up---compared to sequential algorithms---for a relatively small number of processors (certainly much below linear, but ideally even close to constants). Indeed, this is the primary focus in recent research, and over the past decade, there have been numerous exciting results on achieving work-efficient or nearly work-efficient parallel algorithms for various problems. See, e.g.,  \cite{fineman2018nearly, jambulapati2019parallel, blelloch2020parallelism, li2020faster, andoni2020parallel, cao2020efficient,dhulipala2021theoretically,anderson2021parallel,rozhovn2022undirected,rozhovn2022deterministic, ghaffari2023work,ghaffari2024work}.

\medskip
\paragraph{State of the Art in Parallel Derandomization}
We focus here on the maximal independent set (MIS) problem, which was the target of early work on parallel derandomization, and remains a central problem in the area. Karp and Wigderson~\cite{karp1984fast} gave the first (deterministic) parallel algorithm for MIS, and along the way commented that a simpler part of it provides a randomized algorithm. Their deterministic algorithm has $\tilde{O}(n^3)$ work and $\poly(\log n)$ depth. At the core of their algorithm is a parallel derandomization of the following nature: Their argument shows that there exists an independent set with a certain high score, by doing an averaging over all $t$-subsets of a (large) set of vertices. Hence, a randomly selected $t$-subset would also achieve a similarly high score (in expectation, or after a few repetitions). To make this deterministic, they show that the space of all $t
$-subsets, which can have $n^{\Theta(t)}$ options, can be replaced with a much more compact space of certain combinatorial designs, with only $\tilde{O}(n)$ options, while retaining the averaging argument. Hence, one can check all of these options in parallel and obtain a deterministic algorithm.

Luby~\cite{luby1985simple} made the above derandomization much more generic, casting it as a derandomization for randomized algorithms with pairwise independence. He gave an elegant randomized parallel MIS algorithm with expected $O(m)$ work and $\poly(\log n)$ depth, and he showed that the analysis (of each round) can be done using merely pairwise independence. Hence, instead of the naive exponential-size space of $n$ independent random variables, the space of randomness can be limited to $O(n^2)$ points which are pairwise independent. One can check all of these in parallel and deterministically, and that results in a deterministic parallel MIS algorithm with $O(m \cdot n^2)$ work. Alon, Luby, and Babai~\cite{alon86} independently presented similar algorithms for MIS, along with the same derandomization technique for pairwise-independent randomized algorithms with work overhead $\Omega(n^2)$, and its generalization to $d$-wise for constant $d$ which has work overhead $n^{\Theta({d})}$.

In a subsequent paper, Luby~\cite{luby1988removing} noted that the main drawback of the above techniques is the blow up in the work bound. He explicitly phrased it in terms of the number of processors needed. To mitigate this issue, Luby gave an elegant parallel derandomization technique by doing a high-dimensional binary search: instead of simply running the randomized parallel algorithm for all the $\Theta(n^2)$ possible seeds of pairwise independent randomness, he proposed doing a binary search in a well-structured pairwise space, in $\Theta(\log n)$ stages. Said differently, he showed how one can view the seed space as a $\Theta(\log n)$-bit string, such that one can fix these bits one by one. For each bit, one has to compute the expectation of the desired score functions, conditioned on the bits fixed so far, and then choose the next bit so as to maintain the score. This is the most generic technique known to date for derandomizing parallel algorithms, which can usually retain work-efficiency up to $\poly(\log n)$ factors. It has played a pivotal role in other fundamental deterministic parallel algorithms. See, e.g., the work of Berger, Rompel, Shor~\cite{berger1989efficient} on parallel deterministic algorithms for set cover, and those of Motwani, Naor, Naor~\cite{motwani1989probabilistic} and Berger and Rompel~\cite{berger1989simulating} on parallel deterministic algorithms for set discrepancy and lattice approximation problems. However, the method inherently needs a $\poly(\log n)$ increase in the work bound. In particular, it involves $\Omega(\log n)$ stages of a (non-trivial, if not complex) binary search, and each of these stages uses at least linear work. This state-of-the-art gives the following disappointing impression: if we want a deterministic algorithm, we have to have at least $\poly(\log n)$ processors, or otherwise we are better off using the naive sequential algorithms (with just one processor). Our work is an attempt at understanding and bypassing this limitation.

\subsection{Our Results}
We give a novel parallel derandomization technique whose \textit{work overhead} is exponentially smaller, being $\poly(\log\log n)$ instead of $\poly(\log n)$. Our approach is quite different than the above derandomizations, and can be viewed as a simpler gradual \textit{rounding} algorithm in the following sense: The approach starts by viewing the randomized algorithms as a fractional solution, with a certain quadratic score function, and rounds the solution step by step, each time by a $2$ factor, while approximately retaining the score. As a part of the rounding, it also ensures that the problem remaining for the next rounding step is a constant factor smaller. Hence, over all iterations, this yields a deterministic solution that is essentially work-efficient. The following statement presents our concrete result for the benchmark problem of maximal independent set:

\begin{restatable}{theorem}{mis}\textnormal{(\textbf{Maximal Independent Set Algorithm})}
\label{thm:MIS}
    There is a deterministic parallel algorithm that, given any $n$-node $m$-edge graph $G=(V, E)$, computes a maximal independent set of it using $O((m+n) \poly(\log \log n))$ work and $\poly(\log n)$ depth.  
\end{restatable}
We get a similar algorithm for the maximal matching problem:

\begin{restatable}{theorem}{mm}\textnormal{(\textbf{Maximal Matching Algorithm})}
\label{thm:MM}
    There is a deterministic parallel algorithm that, given any $n$-node $m$-edge graph $G=(V, E)$, computes a maximal matching of it using $O((m+n) \poly(\log \log n))$ work and $\poly(\log n)$ depth. 
\end{restatable}

As a core part of the above results, and as a subroutine that we think can find applications in a much wider range of problems, we give an essentially work-efficient deterministic parallel algorithm for the \textit{hitting set} problem, abstracted in the statement below. The general setup is a bipartite graph $G=(U\sqcup V, E)$, where selecting a subset $S$ that includes each node $v\in V$ with probability $p_v$ would hit a constant fraction of nodes $u\in U$ (if desired, weighted by a nonnegative importance $imp_u$ for each node), in the sense $|N_G(u)\cap S|\in [1, \Theta(1)]$. The statement shows that we can deterministically compute a set $S$ with similar guarantees.

\begin{theorem}\label{thm:hittingSet} \textnormal{(\textbf{Hitting Set Algorithm})} Consider an $n$-node $m$-edge bipartite graph $G=(U\sqcup V, E)$ and suppose that each $v$ has a probability (or fractional solution) $p_v\in [0,1]$ such that $\forall u\in U$, we have $\sum_{v\in N_G(u)} p_v \in \Theta(1).$ Suppose also that each node $u\in U$ has a real-valued importance $imp_{u}\in \mathbb{R}_{+}$. There is a parallel algorithm that, using $O((m+n) \poly(\log \log n))$ work and $\poly(\log n)$ depth, computes two subsets $S\subseteq V$ and $U_{good}\subseteq U$ such that 
\begin{itemize} 
\item $\sum_{u\in U_{good}} imp_u\geq (3/4)\sum_{u\in U} imp_u$, 
\item $\forall u\in U$, we have $|N_{G}(u)\cap S|\in \Theta(1)$.
\end{itemize}
\end{theorem}

\paragraph{A note on the work bounds} In the above theorem statements, we expressed the work bound as $O((m+n) \poly(\log \log n))$. Our write-up focuses on conveying the main ideas and thus we have prioritized algorithmic simplicity over the exact constant in the exponent of $\poly(\log\log n)$. We believe that an optimized variant of our algorithms, which repeats certain steps with more nuanced parameterizations, achieves a work bound of at most $O((m+n) (\log \log n)^{1.1})$. However, it remains a tantalizing question whether one can make the overhead substantially smaller than $\Theta(\log\log n)$, and even reach true work efficiency of $O(m+n)$.

\subsection{Our Technique in a Nutshell}
Instead of discussing our method for the hitting set, maximal independent set and matching problems, let us first use a much simpler warm up. We can showcase the basic rounding idea of our technique in the context of the max cut problem. 

\paragraph{Warm up---Max Cut} Given any $n$-node $m$-edge edge-weighted graph $G=(V, E)$ where $w(e)$ denotes the weight of edge $e$, the max cut problem asks for finding a set $S\subset V$ that maximizes $\sum_{e\in (S, V\setminus S)} w(e)$. Including each node $v\in V$ in $S$ with probability $1/2$ gives a cut $(S, V\setminus S)$ with expected weight $\sum_{e\in E} w(e)/2$. The max cut can be seen as maximizing $\sum_{\{v, v'\}\in E} (x_v + x_{v'} - 2x_v x_{v'})\cdot w(e)$ where $x_v\in \{0,1\}$, for which the \textit{fractional solution} $x_v=1/2$ for all $v\in V$ nicely represents the simple randomized algorithm and gives (fractional) cut value $\sum_{\{v, v'\}\in E} (x_v + x_{v'} - 2x_v x_{v'})\cdot w(e) = \sum_{e\in E} w(e)/2.$

There is a simple and classic way to turn the randomized/fractional solution into a deterministic sequential algorithm, using conditional expectations\footnote{Indeed, this max cut problem is used often in algorithmic courses and textbooks as the context for introducing the method of conditional expectations.}: We process vertices $V=\{1,2, \dots, n\}$ one by one, deciding for each $i\in V$ whether to put it in $S$ or $V\setminus S$ in a greedy way, such that we maximize the weight of the incident edges cut. That is, we put $i$ on the side that at most half of the total weight of edges from $i$ to $\{1, 2, \dots, i-1\}$ are in that side. However, this method does not yield a satisfactory parallel deterministic algorithm, as the process decides about the vertices one by one and thus has depth $\Omega(n)$. 

One can obtain a deterministic parallel algorithm by leveraging pairwise independence: notice that $\sum_{e\in E} w(e)/2$ is the expected weight of the cut in the randomized process, even if the choices of different vertices is pairwise independent (instead of mutually independent). Hence, there is a space of $O(n^2)$ randomness seeds that is good enough for this randomized algorithm~\cite{luby1985simple, alon86}. Checking all of these in parallel gives a deterministic parallel algorithm with $O(mn^2)$ work and $\poly(\log n)$ depth. Once can use Luby's method of binary search in a well-structured pairwise space~\cite{luby1988removing} to cut the work to $O(m \poly(\log n))$. We are not aware of a prior deterministic result with a better work bound. We next discuss a simple method that yields a deterministic algorithm with $O(m\log\log n)$ work and $\poly(\log n)$ depth.

\begin{center}
\begin{minipage}{0.95\textwidth}
\begin{mdframed}[backgroundcolor=black!10]
\vspace{3pt}
\begin{lemma}
There is deterministic parallel algorithm that, given any $n$-node $m$-edge edge-weighted graph $G=(V, E)$ where $w(e)$ denotes the weight of edge $e$ and any $\eps>0$, computes a set $S \subset V$ such that $cut_{G}(S, V\setminus S)\geq (1/2-\eps) \sum_{e\in E} w(e)$. The algorithm uses $O((m+n)\log\log n)$ work and $\poly(\log n/\eps)$ depth.
\end{lemma}
\smallskip
\begin{proof}[Proof Sketch]
We compute a $O(1/\eps)$-coloring of vertices such that the weight of monochromatic edges is at most $\eps\cdot \sum_{e\in E} w(e)$. This can be done with $O((m+n)\log\log n)$ work and $\poly(\log n/\eps)$ depth as we later describe in \Cref{lem:defective_coloring}. Then, we process the vertices in $O(1/\eps)$ stages. In each stage $i$, we make a greedy choice for each node of color $i$ based on its bichromatic edges connecting to vertices of colors $\{1, 2, \dots, i-1\}$. Note that this is done in parallel for all nodes of color $i$, but given the coloring, the edges impacted by different nodes are disjoint.     
\end{proof}
\vspace{1pt}
\end{mdframed}
\end{minipage}
\end{center}
The above gives a simple illustration of the rounding idea, (\textit{one step of}) rounding fractional $1/2$ solutions to $0$ or $1$ while (\textit{approximately}) retaining \textit{one} quadratic objective function. We next discuss a few high-level points about the rounding in our main results, which are far more involved.

\paragraph{Hitting Set, Maximal Matching, and Maximal Independent Set} Let us focus on the hitting set problem, which gives a simpler context for discussing the technique, and even just a special \textit{regular} case of it, in the following sense: Consider a bipartite graph $G=(U\cup V, E)$ and suppose that it is $U$-regular meaning that each node in $U$ has degree $d$. Thus, selecting each node $v\in V$ with probability $p_v=\Theta(1/d)$ gives, with probability at least $2/3$, a set $S\subset V$ such that, for nearly all nodes $u\in U$, we have $|N_{G}(u)\cap S| \in [1, \Theta(1)]$ and $|S|=\Theta(|V|/d).$ How can we compute such a set $S$ deterministically, and in parallel?

This can be seen as a multi-step rounding problem from the fractional solution $\mathbf{x}=(x_1, x_2, \dots, x _n)$ where $x_i=\Theta(1/d)$ for all $i\in V$, to an integral solution $\mathbf{y}=\{0, 1\}^{V}$, which gives the characteristic vector of $S$. There are a few notable differences here, compared to the max-cut problem: 
\begin{enumerate} 
\item[(1)] Here, we do not have a single and clear quadratic objective function as we had in the max cut case. Instead, we are trying to maintain a constraint for (nearly) all $u\in U$, namely that $\sum_{v\in N_{H}(u)} x_v$ is maintained (approximately) throughout the rounding of the vector $\mathbf{x}=(x_1, x_2, \dots, x _n)$ and we will eventually have $|N_{G}(u)\cap S| \in [1, \Theta(1)]$. 
\item[(2)] We have $\log d$ steps of $2$-factor rounding, instead of a single one, and this can be in general up to $\log n$ steps. To retain work-efficiency without a $\log n$ factor in the work bound, we need to ensure that after every rounding step, the size of the remaining rounding problem, and in particular its number edges, goes down by a constant factor.
\end{enumerate}
To address the above two points, and control many constraints at the same time, we use the following approach: we bundle the neighbors $N_{H}(u)$ of each node $u$ into bundles $B_1$, $B_2$, \dots, $B_{|N_H(u)|/b}$ of size $b$---ignoring divisability issues and losses for now. We then write one quadratic objective function to (approximately) enforce that, for the set $V'\subseteq V$ that remain nonzero after one rounding, for each bucket $B_i$, we have a small $(|B_i\cap V'|-b/2)^2$. Notice that when $V'$ is chosen randomly where each $V$ is placed in $V'$ with probability $1/2$, we have $\E[(|B_i\cap V'|-b/2)^2]=\Theta(b)$. What we do in derandomization is like putting a Chebyshev-like inequality/constraint in each bucket. If we manage to maintain that $(|B_i\cap V'|-b/2)^2 \leq \eps \cdot \Theta(b^2)$ for some $\eps < 1/b$, that means $|V'\cap B_i|\in 1/2 (1\pm \eps)|V \cap B_i|$. That is, the single step of rounding perturbed the number of neighbors inside the bucket only by a $1\pm \eps$, compared to the expectation. We add up all these objectives $(|B_i\cap V'|-b/2)^2$ over all nodes $u$ and all the relevant buckets, and instead of controlling each of them individually, we control their summation (with appropriate normalizations). Then with an averaging argument we can show that $1-\poly(\eps)$ fraction of nodes $u\in U$ (or more generally, such a fraction of their total importance weights) have such a controlled movement in nearly all of their buckets and $|N_{H}(u)\cap V'| \in 1/2 (1\pm \poly(\eps))|N_{H}(u)\cap V|$. With a similar bucketing idea (and additional quadratic objectives for each of the additional buckets, which are all added together with the previous ones), we enforce that the number of remaining edges goes down by a constant factor.

However, the above reveals two issues: 

\begin{enumerate} 
\item[(3)] Given that we have up to $\log n$ steps of rounding, to keep the overall losses at bay, we could not afford to have a constant $\eps$ relative loss in each step, as they add up. We would need to set $\eps$ to at most $1/\log n$, if we followed the vanilla approach. However, that would not be work efficient, because of the next item.
\item[(4)] In the above bucketing and quadratic formulas, we need to have $\eps\leq 1/b$. So, if we go with the seemingly necessary $\eps<1/\log n$ as mentioned above, we would have $b\geq \log n$. But the size of the quadratic formulas, one per bucket, and the virtual graph that would represent them for rounding purposes, is a $b$ factor larger than the original graph. Hence, even the graph representing one iteration's rounding would be of size $\Omega(m \log n)$, thus nullifying our work efficiency pursuit. 
\end{enumerate}
To combat these, we start with a constant $\eps$ for the first iteration, thus a constant $b=\Theta(1/\eps)$, and in the course of different rounding iterations, we geometrically decrease $\eps$ by a constant factor per iteration, leveraging the fact that the size of the graph remaining for the iteration also had a constant decay. Setting the factors appropriately, we can ensure that the overall approach remains essentially work-efficienct. 

The discussion above omitted a few additional points. Let us mention three here: (I) The fact that the general hitting set instance is not regular, and different nodes $v\in V$ have different $p_v$. (II) Especially because of the previous point, the error coming from the nodes left outside buckets (in different probability classes) needs to be also controlled. (III) For the MIS problem, the objective functions are more complex, and we will not be able to control them as tightly as the above process. The remedies to these issues, and thus also our overall approach, are more involved and we leave them for the technical sections.

\paragraph{Related Work in Distributed Computing} The general \textit{rounding idea} in our work, of viewing randomized algorithms as fractional solutions and then gradually rounding them, is inspired by the similar local rounding approach that underlies much recent progress in deterministic distributed graph algorithms~\cite{GhaffariK21, ghaffari2023netdecomp, faour2022local, ghaffari2024near}. However, the computational aspects are quite different between the two models and those distributed algorithms imply, at best, parallel algorithms with work bound $(m+n) \poly(\log n)$, which was known from 1980s results~\cite{luby1988removing}. 

In particular, even for rudimentary tasks such as coloring (\Cref{lem:coloring,lem:defective_coloring}), which are basic ingredients of our algorithm, the algorithms on the distributed side do not yield work-efficient deterministic parallel algorithms and would need $\Omega((m+n)(\log n))$ work. The distributed MIS/MM algorithms involve $\poly(\log n)$ iterations without enforcing a shrinkage in problem sizes, and thus their overall work bound is $\Omega((m+n)\poly(\log n))$. For work-efficiency in the parallel algorithm, we need to ensure a constant-factor shrinkage in the problem over different iterations, which requires much tighter control on the objective functions, as discussed above by quadratic objective functions over different buckets. In a sense, one can view this paper as nicely showing how this general gradual rounding idea can be imported from distributed algorithms to parallel algorithms, though with quite different instantiations and internal ingredients, to yield parallel derandomizations that essentially retain work efficiency.

\section{Preliminaries}
\subsection{Basic notations and graph representation format}\label{subsec:basics}
\paragraph{Notations} For any graph $G=(V, E)$, we use the notation $E(G)$ to denote the edges of the graph $G$, and for any subset $S\subseteq V$, the notation $G[S]$ denotes the subgraph induced by set $S$.

\paragraph{(Compact) Graph Representation} We assume that the input is an $n$-node $m$-edge graph $G=(V, E)$, provided as adjacency lists: an array $[1..n]$ for vertices and for each for $i\in [1..n]$, we are given $Adj[i]$ which is an array of length $deg(v)$ that describes the neighbors of $i$. Alternatively, we can assume that the graph is provided as a list of edges $E$. In that case, we can obtain the adjacency lists representation, using a standard deterministic polynomial-range integer sorting subroutine~\cite{bhatt1991improved} (later stated in \Cref{lem:SortBhatt}), using $O((m+n)\log\log n)$ work and $O(\log n)$ depth: We first duplicate each edge $\{u,v\}$ as $(u, v)$ and $(v,u)$ and then sort this list lexicographically, thus giving the adjacency list of each node as a contiguous block. By including for each node $v$ two additional items $(v, 0)$ and $(v, n+1)$ in the sort, we also have direct access to the beginning and end of this contiguous block of the adjacency list of $v$.

More generally, when we invoke a subroutine (abstracted by a lemma) on a subgraph $G'=(V', E')$, to maintain work-efficiency, we assume that the subgraph is represented in a compact form to the subroutine/lemma. In particular, we give a numbering $[1..n']$ for $n'=|V'|$ as the array of vertices. For that, we can perform a simple prefix sum on the array $A[1..n]$ with a value $A[i]=1$ for every vertex $i\in V'$ and $A[i]=0$, which gives to the node $i$ its new identifier $i'=\sum_{j=1}^{i} A[j]\in \{1, \dots, n'\}.$ This takes $O(n)$ work and $O(\log n)$ depth. 
We can perform a similar compaction for edges, assuming edges of $E'$ are marked in the adjacency lists $Adj[i]$, for each $i$: By scanning all entries of $Adj[i]$, we can read the new identifiers of the neighbors of $i$. Using an additional prefix sum for each retained vertex $i\in V'$ on $Adj[i]$ among edges marked to be in $E'$, we can compute the compact $Adj'[i]$ that only has neighbors of $i$ in $E'$, and this uses additional $O((m+n))$ work and $O(\log n)$ depth. If the edges in $E'$ are provided only as a list of edges, then we can apply the sorting subroutine mentioned before (\Cref{lem:SortBhatt}) to obtain compact adjacency lists for $G'$, using $O((m+n)\log\log n)$ work and $O(\log n)$ depth.

\paragraph{Precomputations for Prime Selection and Root Computations} In our coloring algorithms (\Cref{lem:coloring,lem:defective_coloring}), we use subroutines that find a prime number $p$ in a given range $[x, 2x] \subseteq [1, \Theta(n^{1/3})]$, where $n$ is the number of nodes in the original input graph, and we also compute $\sqrt{y}$ in $\mathbb{F}_p$ for $y\in \mathbb{F}_p$. We prepare tables for these with a one-time parallel computation, which provides all the primes in $[1, \Theta(n^{1/3})]$ and all the relevant roots. This can be done deterministically using no more than $O(n^{2/3})$ work and $O(\log n)$ depth (in these bounds, we can even afford to check if each smaller number divides each particular number in the range). In addition, for each prime $p$ in this range, we can find all the roots using $O(p)$ work and $O(\log n)$ depth, simply by computing $z^2$ for all $z\in \mathbb{F}_{p}$. Thus, in total, this is also at most $O(n^{2/3})$ work. Since our main theorem statements (e.g., for the MIS problem) have complexity $\Omega(n)$, these one-time $O(n^{2/3})$ work bounds are subsumed by the complexities of the other parts and we do not mention them explicitly again.  

\subsection{Sorting Subroutines}
In our algorithms, we make use of sorting as a basic subroutine. For many of our applications, the following simple lemma suffices.
\begin{lemma}[\textbf{Sorting Lemma}]
    \label{lem:sorting}
    Suppose that we have an input array $A[1..k]$ of $k$ integers, each in the range $\{1, 2, \dots, \lceil\log N \rceil\}$ (here $N$ is the number of nodes in the original graph.) There is a deterministic algorithm that outputs the sorted array $A[1..k]$, using $O(k\log\log N)$ work and $O(\log k \cdot \log\log N)$ depth. 
\end{lemma}
\begin{proof}
We perform a variant of the classic radix sort algorithm, specified for our setup. In particular, we view each integer as a $\log\log N$ bit number and use $\log\log N$ iterations in the algorithm. In iteration $i$, we re-sort all entries $A[1], A[2], \dots, A[k]$ according to their $i^{th}$ least significant bit, and importantly using a stable sorting subroutine. Let $A[j]_i$ denote the $i^{th}$ least significant bit of $A[j]$. As standard, stability here means that for any two entries $A[j]$ and $A[j]'$ for $j<j'$, if $A[j]_i = A[j']_i$, the sorting will have $A[j]$ before $A[j']$ in the sorted order according to the $i^{th}$ bit. 

Let us discuss this stable sort for the $i^{th}$ bit. We first generate an array $C[1..k]$ where for each index $j$ such that $A[j]_i=0$, we enter $C[j]=1$ and otherwise we enter $C[j]=0$. Then we perform a prefix sum so that each index $j$ learns $i_j=\sum_{j'=1}^{j} C[j]$. For each $j$ such that $A[j]_0$, we write $A[j]$ in index $i_j$ of a temportary sorted array $B$, i.e., we set $B[i_j]=A[j]$. Let $I=\sum_{j'=1}^{k} C[k]$ denote the total number of ones in $C[1..k]$. Then, we do essentially the same process but to provide output indices for $A[j]$ such that $A[j]_i=1$. First, set $C[1..k]=1*[k]- C[1..k]$, i.e., replace zeros with ones and vice versa. Then, perform a new prefix sum on this updated $C[1..k]$, so that each index $j$ learns $i_j=\sum_{j'=1}^{j} C[j]$. Then, for each $j$ such that $A[j]_i=1$, write $A[j]$ in $B[I+i_j]$. At the end, we move the items from the temporary array $B[]$ to our main array $A[]$, by setting $A[1..k]=B[1..k]$, to prepare for the next iteration (or for the final output).

Per iteration, we spend $O(k)$ work and $O(\log(k))$ depth. Given that we have $O(\log\log N)$ iterations, the work and depth bounds follow.
\end{proof}

For moving between representations of a graph (e.g., from the list of all graph edges to adjacency lists, one per node), we use a stronger deterministic sorting stated below. We note that this lemma assumes that, in the case of concurrent writes, an arbitrary one takes place. 
\begin{lemma} [Bhatt et al.~\cite{bhatt1991improved}]
\label{lem:SortBhatt}
    Suppose that we have an input array $A[1..k]$ of $k$ integers, each in the range $\{0, 1, 2, \dots, N \}$. There is a deterministic parallel algorithm that sorts this array using $O(k\log\log N)$ work and $O(\log k + \log\log N)$ depth.
\end{lemma}

\subsection{Coloring, Defective Coloring, and Gradual Rounding}

\begin{lemma}[\textbf{Coloring Lemma}]
\label{lem:coloring}
    Given a graph $G=(V, E)$ with $n$ nodes and $m=|E|$ where each node has a unique identifier in $\{1, \dots n\}$, there is a deterministic parallel algorithm that computes an $O(\Delta^2)$ coloring of $G$, using $O((m+n)\log\log n)$ work and $\poly(\log n)$ depth. Here, $\Delta$ is an upper bound on the maximum outdegree, in a given orientation of edges of $G$ (and on the maximum degree, in the absence of such an orientation). 
\end{lemma}
\begin{proof} This proof can be viewed as an adaptation of a classic scheme of Linial~\cite{linial1987LOCAL} in distributed computing. However, Linial's algorithm ignores computational aspects (as is done in the \local model of distributed computing) and a direct application of it would result in an algorithm with work bound $O((m+n)\log n \log^* n)$. Here, we provide an algorithm with the same outline but different details that achieves an $O((m+n)\log\log n)$ work bound.

Initially, we view the identifiers of the nodes as a trivial coloring with $n$ colors. The algorithm has $O(\log \log n)$ iterations, where per iteration we reduce the number of colors from its current bound of $k$ to $\max\{\Theta(k^{2/3}), 5\Delta^2\}$. We describe the process for one such iteration.

Let $k'=\max\{3 k^{1/3}, 3\Delta\}$ and choose a prime $p\in [k', 2k']$, leveraging that we already pre-computed the list of all primes $[1, \Theta(n^{1/3})]$ as mentioned in the preliminaries. We view each of the old colors $q\in \{1, 2, \dots, k\}$ as a degree-2 polynomial in $\mathbb{F}_p$, represented as $f_{q}(x)=a_q x^2+b_qx+c_q$, where $c_q= (q \; mod \; p)$, $b_q=((q-c_q)/p \; mod \; p)$, and $a_q=((q-c_q - b_q p)/p^2 \; mod \; p)$. This space of the new colors will be points $(x, y)$ for $x, y\in \mathbb{F}_p$, which is $p^2 = O((k')^2)$ colors. Consider node $v$ with old color $q$. We consider the evaluations of $f_{q}(x)=a_q x^2+b_qx+c_q$ for $x\in \{0, 1, 2, 3deg(v)-1\}$, where $deg(v)\leq \Delta$ denotes the outdegree of node $v$ (or its degree, in the absence of an orientation). We identify a point $(x, f_q(x))$ such that there is no out-neighbor $u$ of $v$ (or neighbor, in the absence of an orientation) with old color $q'$ for which $f_q'(x)=f_q(x)$.

For every outneighbor $u$ (or neighbor, in the absence of an orientation), for which the old color we denote as $q'$, the polynomial $f_q'(x)-f_q(x) = ax^2 +bx+c$ has at most two roots, since it is of degree two. We can identify these two roots in constant work using the standard quadratic formula $(-b\pm \sqrt{b^2-4ac})/2a$ (where calculations are done in $\mathbb{F}_p$, and leveraging that we have already precomputed all the roots in $\mathbb{F}_p$ as mentioned in the preliminaries). We then mark the related entries of $x$ that are in $\{0, 1, 2, 3deg(v)-1\}$ as \textit{lost}. Since for each outneighbor we mark at most two points as lost, among $\{0, 1, 2, 3deg(v)-1\}$, there is at least one (and indeed many) points that are not lost. We identify the smallest such point $x^*$ (e.g., using a prefix sum) and set the new color of node $v$ equal to $(x^*, f_{q}(x^*))$. 

Notice that the work for the above recoloring is $O(1+deg_{v})$ for node $v$, and the process is done in $O(\log n)$ depth. Hence, the overall work for this recoloring is $O(m+n)$, and the depth is $O(\log n)$. Since we have $O(\log\log n)$ iterations of recoloring to go from $n$ colors to $O(\Delta^2)$, the total work is $O((m+n) \log\log n)$ and the depth becomes $O(\log n\cdot \log\log n)$. 
\end{proof}

\begin{lemma}[\textbf{Defective Coloring Lemma}]
\label{lem:defective_coloring}
    Consider any input graph $G=(V, E)$ with $n$ nodes and $m=|E|$, where each node has a unique identifier in $\{1, \dots n\}$ and each edge $e\in E$ has a single-word nonnegative weight $w(e)$. Let $\eps>0$ be an input parameter. There is a deterministic parallel algorithm that computes an $O((1/\eps))$ defective coloring of $G$ such that the total weight of monochromatic edges is at most $\eps \sum_{e\in E} w(e)$. The algorithm uses $O((m+n)\log\log n)$ work and $\poly(\log n/\eps)$ depth.
\end{lemma}

\begin{proof} The algorithm has two phases, which we describe separately.  

\paragraph{The first phase} Let $I=\lceil\log_{3/2}\log n\rceil$ and $\eps'=\eps/2I$. The first phase consists of  $I$ iterations, where we start with the trivial $n$-coloring implied by the identifiers, and gradually reduce the number of colors to $O(1/(\eps')^2)$. Per iteration, we reduce the number of colors from its current bound $k$ to a bound $\Theta(\max\{k^{2/3}, (1/\eps')^2\})$. Furthermore, in each iteration, we declare an additional set of edges with total weight at most $\eps' \sum_{e\in E}w(e)$ as lost, and we remove them from the graph, allowing them to be monochromatic. Hence, by the end of the first phase, the total weight of these lost edges is at most $\lceil\log_{3/2}\log n\rceil\cdot \eps' \sum_{e\in E}w(e) \leq (\eps/2) \sum_{e\in E}w(e)$. Besides these lost edges which are removed from the graph, we ensure that each other edge is bichromatic in every iteration. The overall approach per iteration is similar to the one in the proof of \Cref{lem:coloring}, except for the much smaller number of colors, made possible by allowing $\eps'$ fraction of edge weight to become monochromatic. 

Consider one iteration where the current number of colors is $k$. Let $k'=\max\{3 k^{1/3}, 3(1/(\eps'))\}$. We will reduce the number of colors in this iteration to $\Theta((k')^2)$. First, choose a prime $p\in [k', 2k']$, leveraging that we already pre-computed the list of all primes $[1, \Theta(n^{1/3})]$. We view each of the old colors $q\in \{1, 2, \dots, k\}$ as a degree-2 polynomial in $\mathbb{F}_p$, represented as $f_{q}(x)=a_q x^2+b_qx+c_q$, where $c_q= (q \; mod \; p)$, $b_q=((q-c_q)/p \; mod \; p)$, and $a_q=((q-c_q - b_q p)/p^2 \; mod \; p)$. The space of the new colors $\{1, 2, \dots, \Theta(k')^2\}$ will be points $(x, y)$ for $x, y\in \mathbb{F}_p$. Consider node $v$ with old color $q$. We look at the evaluations of $f_{q}(x)=a_q x^2+b_q+c_q$ for $x\in \{0, 1, 2, 3\lceil(1/\eps')\rceil-1\}$, and we identify a point $(x, f_q(x))$ such that at most $\eps'$ of neighbors $u$ of $v$ have an old color $q'$ for which $f_q'(x)=f_q(x)$.

For every neighbor $u$, for which the old color we denote as $q'$, the polynomial $f_q'(x)-f_q(x) = ax^2 +bx+c$ has at most two roots, since it is of degree two. We can identify these two roots in constant work using the standard quadratic formula (where calculations are done in $\mathbb{F}_p$, and leveraging that we have already precomputed all the roots in $\mathbb{F}_p$ as mentioned in the preliminaries). If a root has $x\in \{0, 1, 2, 3\lceil(1/\eps')\rceil-1\}$, we add the weight of the edge $w(\{u,v\})$ to the hit-score of that point $x$. Since for each edge $\{v, u\}$ we add at most $2w(\{u,v\})$ to the hit scores of all the points, the total hit score on the points is at most $2\sum_{u\in N_{G}(v)} w(\{v, u\})$. Hence, among $x\in \{0, 1, 2, 3\lceil(1/\eps')\rceil-1\}$, there is a point (and indeed many) with hit score less than $\eps' \sum_{u\in N_{G}(v)} w(\{v, u\})$. We choose one such point $(x, f_q(x))$ as the new color of node $v$. We then check the neighbors one more time and for any node $u$ that chose the same color, mark edge $\{u, v\}$ as lost, and we remove it from the graph. Notice that since for each node $v$ the weight of its lost edges in this iteration is at most $\eps'\sum_{u\in N_{G}(v)} w(\{v, u\})$, the total weight of all edges declared lost in this iteration is at most $\eps' \sum_{e\in E}w(e)$. 

We an add one small optimization in the above, for computational purposes: if a node $v$ has degree $d(v)\leq (1/\eps')$, it suffices to limit its range of $x\in \{0, 1, 2, 3deg(v)\} \subseteq \{0, 1, 2, 3\lceil(1/\eps')\rceil-1\}$. Notice that there must be an $x\in \{0, 1, 2, 3deg(v)\}$ that is not hit by any of the neighbors (as each of them hits two points). Thus, in this corner case, node $v$ will have no monochromatic edge. But more crucially, now the work per iteration for every node $v$ is upper bounded by $O(deg(v))$. Hence, the overall work for the iteration is $O(m+n)$. And the depth is $O(\log n)$. Hence, the total work and depth of the first phase are $O((m+n)\log\log n)$ and $O(\log n \log\log n)$. At the end of this first phase, we have a coloring with $O(((\log\log n)/\eps)^2)$ colors, with the guarantee that the total weight of edges lost so far is at most $(\eps/2)\sum_{e\in E}w(e)$. 

\paragraph{The second phase} We now explain how to update the colors from the current bound of $k=O(((\log\log n)/\eps)^2)$ to the output bound of $k'=3\lceil 1/\eps\rceil $ by sacrificing an additional set of edges of weight at most $(\eps/2)\sum_{e\in E}w(e)$. 
Orient each edge (that is not lost so far) as going from its endpoint with the higher color to the endpoint with the lower color. We process the colors one by one, in $k$ iterations, where per iteration $i$ we compute the new color of the vertices with old color $i$. Consider iteration $i$ and a vertex $v$ with old color $i$, and its outgoing edges (i.e., those going to vertices with old color at most $i-1$, which are vertices for which we have already computed the new color). Choose a new color $q\in \{1, 2, \dots, 3\lceil /\eps\rceil \}$ such that the weight of the outgoing edges connecting to nodes which have received new color $q$ is less than $\eps/2$ fraction of the total weight of these outgoing edges. This is the new color of node $v$. Declare all outgoing edges connecting to nodes with new color $q$ as lost. Again, we can add the small optimization that if node $v$ has less than $1/\eps$ outneighbors, we can limit its color range to $\{1, 2, \dots, outdeg(v)\}$, and ensure it has no monochromatic outgoing edge. Overall, the weight of the edges declared lost (and thus monochromatic) in the second phase is $(\eps/2)\sum_{e\in E}w(e)$. This phase of the algorithm has work $O(m+n)$, since each edge is processed in only one iteration, and depth $O(((\log\log n)/\eps)^2 \cdot O(\log n) = \poly(\log n/\eps)$. 
\end{proof}

\begin{lemma}[\textbf{Rounding Lemma}]
\label{lem:local_rounding}
Let $H=(V, E)$ be a graph with $n$ nodes and $m$ edges. Also, for each node $v\in V$, consider a (potentially negative) utility $util(v)\in \mathbb{R}$ and for every edge $e\in E$, consider a cost $cost(e)\in \mathbb{R}_{\geq 0}$. Suppose each node $v\in V$ has an identifier in $[n]$. Also, let $\eps>0$ be an arbitrary input. There is a deterministic parallel algorithm with work $O((m+n)\log \log n)$ and depth $\poly(\log n/\eps)$ that computes a set $V'\subseteq V$ such that 
$$\sum_{v\in V'} util(v) - \sum_{e\in \binom{V'}{2} \cap E} cost(e) \geq \big(\sum_{v\in V} (1/2) \cdot util(v) - \sum_{e\in E} (1/4) \cdot cost(e)\big) -\eps \sum_{e\in E} cost(e).$$
\end{lemma}
\begin{proof} We first invoke \Cref{lem:defective_coloring} to compute a defective coloring with $O(1/\eps)$ colors and where the total cost of the monochromatic edges is at most $(\eps) \sum_{e\in E} cost(e)$. This uses $O((m+n)\log \log n)$ work and $\poly(\log n/\eps)$ depth. We (temporarily, for our algorithm's computations) remove these monochromatic edges from the graph. Even if we allow all these edges to appear in $\binom{V'}{2}$, their total cost is $\eps \sum_{e\in E} cost(e)$. Then, we perform a standard derandomization by conditional expectation (for the random process where each node $v\in V$ is put in $V'$ with probability $1/2$) and for the objective $\sum_{v\in V'} util(v) - \sum_{e\in \binom{V'}{2} \cap E} cost(e)$. In particular, we use $O(1/\eps)$ iterations where per iteration $i$ we decide about nodes of color $i$, in a greedy fashion so as to maximize the utilities minus costs, whether to put each of them in $V'$ or in $V\setminus V'$. 
\end{proof} 

\section{Hitting Sets} 
In this section, we present our derandomization for the hitting set problem, formally stated below as \Cref{lem:hitting_set}. 

\begin{restatable}{lemma}{hittingSet}\textnormal{(\textbf{Hitting Set Lemma})}
\label{lem:hitting_set}
Let $H$ be an $n$-vertex $m$-edge bipartite graph with bipartition $V(H) = U \sqcup V$, where each node has a unique identifier in $\{1,2,\ldots, 2n\}$. Also, let $N\geq n$ be a given upper bound. For every $u \in U$, let $imp_u \in \mathbb{R}_{\geq 0}$ and for every $v \in V$, let $k_v \in \{0,1,\ldots, \lceil \log N \rceil\}$. 
There exists a constant $C\geq 1$ and a deterministic parallel algorithm with work $(m+n)\poly(\log \log N) + \poly(\log N)$ and depth $\poly(\log N)$ that computes two subsets $S \subseteq V$ and $U_{good} \subseteq U$ such that $\sum_{u \in U_{good}} imp_u \geq 0.9\sum_{u s\in U} imp_u$ and for every $u \in U_{good}$, $|N_H(u) \cap S| \in [0.5\sum_{v\in N_{H}(u)} 2^{-k_v} -0.5, C\sum_{v\in N_{H}(u)} 2^{-k_v} +C]$.
\end{restatable}

In the following two subsections, we present the key ingredients of this lemma, for two different regimes of probability $p_v=2^{-k_v}$, the low probability regime where $p_v = o(1/\poly(\log n))$ for all $v$ that we target, and then the high probability regime where $p_v = \Omega(1/\poly(\log n))$ for all the remaining $v$. Afterward, we present the proof of \Cref{lem:hitting_set} as a wrap-up subsection \Cref{subsec:hitting_set_wrap}. 

\subsection{Low Probability Regime}
\begin{lemma}[Low Probability 1/2 Sampling]
\label{lem:low_prob_half}
Let $H$ be an $n$-vertex $m$-edge bipartite graph with bipartition $V(H) = U \sqcup V$, where each node has a unique identifier in $\{1,2,\ldots, 2n\}$, let $N\geq n$ be a given upper bound, and let $K=\lceil 100\log\log N \rceil$, and $\gamma \in [\Theta(1/\log N), 0.01)$. For every $u \in U$, let $imp_u \in \mathbb{R}_{\geq 0}$ and for every $v \in V$, let $k_v \in \{K+1, \ldots, \lceil \log N \rceil\}$.

There exists a deterministic parallel algorithm with work $(m+n)\poly(\log \log N)/\gamma^{20}$ and depth $\poly(\log N)$ that computes a bipartite graph $H'\subseteq H$, a subset $U_{good}\subseteq U$, and a subset $S\subseteq V$, such that for $H''=H'[U_{good}\cup S]$, we have:
\begin{itemize}
\item $|E(H'')|+|S|\leq (2/3) \left(|E(H)|+|V| \right) + \poly(\log N)$, 
\item $\sum_{u \in U_{good}} imp_u \geq (1-\gamma)\sum_{u \in U} imp_u$
\item for every node $u\in U$, we have $|N_{H'}(u)| \geq |N_{H}(u)| -\gamma2^{K}$, and
\item for every node $u \in U_{good}$, we have $\sum_{v \in N_{H''}(u)} 2^{-k_v} \in ((1/2)(1\pm \gamma) \sum_{v \in N_H(u)}2^{-k_v} \pm \gamma/2)$.
\end{itemize}

\end{lemma}
\begin{proof}
We partition $V$ into parts $V^{(j)}$ for $j\in \{K+1, \ldots, \lceil \log N \rceil\}$ where we put each node $v\in V$ in $V^{(k_v)}$. Also, let $H^{(j)}$ be the induced subgraph by $U\sqcup V^{(j)}$.
We can compute this partition using the sorting lemma (\cref{lem:sorting}) with $O(n \log \log N)$ work and $\poly(\log N)$ depth.

Let $b := \lfloor \min\{(1/\gamma)^{6}, \gamma 2^{K-1}/\lceil\log N \rceil\}\rfloor$. We define a graph $H'^{(j)}$ by dropping a few edges from $H^{(j)}$: For each node $u \in U$ and each $j\in \{K+1, \ldots, \lceil \log N \rceil\}$, we set $N_{H'^{(j)}}(u)$ by dropping up to $b-1$ neighbors of $u$ in $H^{(j)}$, such that $|N_{H'^{(j)}}(u)|$ is divisible by $b$. We let $H'$ be the graph by the union of the edges of $H'^{(j)}$ for all $j$ and $V(H') = U \sqcup V$. 

Then, for each $j$ and $u \in U$, we partition $N_{H'^{(j)}}(u)$ arbitrarily into buckets $N_{H'^{(j)}}(u) = B^{(j)}_1(u) \sqcup B^{(j)}_2(u) \sqcup \ldots \sqcup B^{(j)}_{|N_{H'^{(j)}}(u)|/b}(u)$ of size $b$. Notice that $|N_{H^{(j)}}(u) \setminus N_{H'^{(j)}}(u)| < b$ and therefore $|N_{H'}(u)| - |N_H(u)| \geq - \lceil \log(N)\rceil \cdot b \geq - \gamma  \cdot 2^K/2 $. This is permitted in the third item of the lemma. Furthermore, the total probability contribution of these dropped neighbors in $\sum_{v \in N_{H}(u)} 2^{-k_v}$ is at most $\gamma/4$, and thus this makes $\pm \gamma/2$ additive error in the fourth item of the lemma, which is permitted. 

Our derandomization will use three potential functions for choosing the set $S\subseteq V$: (I) one to ensure that the number of edges goes down by roughly a $1/2$ factor, and (II) the other to ensure that for nearly all vertices in $U$, weighted by importances, the probability summation $2^{-k_v}$ in their neighborhood drops by roughly $1/2$, and (3) to ensure that $|S|\approx |V|/2$.

\begin{align*}
 \Phi_1(S) &:= &&\frac{4}{\sum_{u\in U} \sum_{j=K+1}^{\lceil\log N\rceil} |N_{H'^{(j)}}(u)|} \cdot \sum_{u \in U} \sum_{j=K+1}^{\lceil\log N\rceil} \left( \sum_{i=1}^{|N_{H'^{(j)}}(u)|/b}\left( |S \cap B^{(j)}_i(u)| -  b/2\right)^2 \right) 
\end{align*}

\begin{align*}
 \Phi_2(S) &:= &&\big(\frac{1}{\sum_{u\in U} imp_u}\big) \cdot  \\
 & && \sum_{u \in U} \left(4 \cdot \frac{imp_{u}}{\sum_{j=K+1}^{\lceil\log N\rceil} |N_{H'^{(j)}}(u)| 2^{-j}} \cdot \sum_{j=K+1}^{\lceil\log N\rceil} 2^{-j} \left( \sum_{i=1}^{|N_{H'^{(j)}}(u)|/b}\left( |S \cap B^{(j)}_i(u)| -  b/2\right)^2 \right) \right) 
\end{align*}

Additionally, let $ B'_1\sqcup B'_2\ldots \sqcup  B'_{\lfloor|V|/{b}\rfloor} \subseteq V$ be a partition of (all except up to $b-1$) $V$ into buckets of size $b$ and

\begin{align*}
 \Phi_3(S) := \frac{4}{b\lfloor|V|/{b}\rfloor} \sum_{i=1}^{\lfloor|V|/{b}\rfloor} \left( |S \cap B^{'}_i| -  b/2\right)^2.
\end{align*}

Notice that for $S$ chosen randomly by including each $V$ node in it with probability $1/2$, we have $\E[\Phi_1(S)]=\E[\Phi_2(S)]=\E[\Phi_3(S)]=1$. We next discuss how we can view these potential functions in the utility and cost language of \Cref{lem:local_rounding} (ignoring the constant terms). In particular, we will have $util_{i}()$ and $cost_{i}()$ to represent $\Phi_{i}(S)$.

For each vertex $v\in V^{(j)}$, we define a utility 
$$util_1(v) = \frac{4}{\sum_{u\in U} \sum_{j=K+1}^{\lceil\log N\rceil} |N_{H'^{(j)}}(u)|} \cdot \sum_{u \in N_{H'^{(j)}}(v)} (b-1)$$
and 
$$util_2(v) = \big(\frac{4}{\sum_{u\in U} imp_u}\big) \sum_{u \in N_{H'^{(j)}}(v)} \frac{imp_{u}}{\sum_{j=K+1}^{\lceil\log N\rceil} |N_{H'^{(j)}}(u)| 2^{-j}} \cdot 2^{-j} (b-1)$$
and for $v \in B'_1 \sqcup B'_2 \ldots \sqcup B'_{\lfloor |V|/b \rfloor}$, let 
$$ util_3(v) = \frac{4}{b\lfloor|V|/{b}\rfloor} (b-1). $$

Furthermore, we make an auxiliary multi-graph $\bar{H}=(V, \bar{E})$ as follows: for every $u\in U$ and for every bucket $B^{j}_{i}(u)$ and each two $v, v' \in B^{j}_{i}(u)$, we add a new edge $\{v, v'\}$ to $\bar{E}$ (notice that there might be several such parallel edges, one for each $u$). We define two costs for this particular edge $\{v, v'\}$ as follows:

$$cost_1(\{v, v'\}) = \frac{8}{\sum_{u\in U} \sum_{j=K+1}^{\lceil\log N\rceil} |N_{H'^{(j)}}(u)|}$$
and 
$$cost_2(\{v, v'\}) = \big(\frac{1}{\sum_{u\in U} imp_u}\big) \frac{imp_{u}}{\sum_{j=K+1}^{\lceil\log N\rceil} |N_{H'^{(j)}}(u)| 2^{-j}} \cdot 2^{-j} \cdot 8.$$ Also, 
for every $i$ and every pair $v, v' \in B'_i$, let us add an edge $\{v, v'\}$ with cost 
$$cost_3(\{v, v'\}) = \frac{8}{b\lfloor|V|/{b}\rfloor}.$$

Notice that $\sum_{e\in \bar{E}} (cost_1(e) + cost_2(e)+cost_3(e)) \leq \Theta((b-1))$. Using \cref{lem:local_rounding}, by setting $\eps = \Theta(1/(b-1))$ with a small leading constant, we can compute with work $(mb + n)\poly(\log \log N) = (m+n)\poly(\log \log N)/\gamma^{30}$ and depth $\poly(\log N)$ a subset $S \subseteq V$ such that $\Phi_1(S) + \Phi_2(S) + \Phi_3(S) \leq 3.1$. From this, we can make several conclusions.

We first use the above and in particular $\Phi_1\leq 3.1$ to conclude that $\sum_{u\in U} |N_{H'}(u) \cap S| \leq (0.51) \sum_{u\in U} |N_{H}(u)|$. Let us call a bucket $B^{(j)}_i(u)$ bad if $||B^{(j)}_i(u)\cap S|-b/2|\geq b^{0.8}$. Notice that for each bad bucket, we have $\left( |S \cap B^{(j)}_i(u)| -  b/2\right)^2 \geq b^{1.6}$. On the other hand, we know that $\Phi_1(S)\leq 3.1$ where \begin{align*}
 \Phi_1(S) &:= &&\frac{4}{\sum_{u\in U} \sum_{j=K+1}^{\lceil\log N\rceil} |N_{H'^{(j)}}(u)|} \cdot \sum_{u \in U} \sum_{j=K+1}^{\lceil\log N\rceil} \left( \sum_{i=1}^{|N_{H'^{(j)}}(u)|/b}\left( |S \cap B^{(j)}_i(u)| -  b/2\right)^2 \right). 
 \end{align*}
Hence, overall, at most $\left(\sum_{u\in U} \sum_{j=K+1}^{\lceil\log N\rceil} |N_{H'^{(j)}}(u)|/b\right)/\Theta(b^{0.6}) = \sum_{u\in U} |N_{H'}(u)|/\Theta(b^{1.6})$ buckets are bad. Therefore, we can conclude that

\begin{align*}
    \sum_{u\in U} |N_{H'}(u) \cap S| \leq  ((1/2+1/b^{0.2}+\Theta(1/b^{0.6}))\sum_{u\in U} |N_{H}(u)| \leq (0.51) \sum_{u\in U} |N_{H}(u)|.
\end{align*}

Next, we use $\Phi_2(S) \leq 3.1$ to conclude that for nearly all nodes in $U$, weighted by importance, the probability summation in their neighborhood goes down by a $1/2$ factor. Let $\mathcal{B}^{(j)}(u)$ be the collection of all $i$ such that bucket $B^{(j)}_{i}(u)$ is bad, with the definition of bad bucket being the same as before, i.e., $||B^{(j)}_i(u)\cap S|-b/2|\geq b^{0.8}$. Let us call a node $u$ bad if $\sum_{j=K+1}^{\lceil\log N\rceil} \sum_{i\in \mathcal{B}^{(j)}(u)} b 2^{-j}$ is more than $\left(\sum_{j=K+1}^{\lceil\log N\rceil} \sum_{i=1}^{|N_{H'^{(j)}}(u)|/b} b 2^{-j}\right)/b^{0.3}$. Let $U_{good}$ be the set of all $u\in U$ that are not bad.

For any node $u\in U_{good}$, we have 
\begin{align*} 
& && \sum_{v \in N_{H'}(u)\cap S}2^{-k_v} = \left(\sum_{j=K+1}^{\lceil\log N\rceil} \sum_{i=1}^{|N_{H'^{(j)}}(u)|/b} |B^{(j)}_i(u)\cap S|  2^{-j}\right) \\
&\geq&& (1-1/b^{0.3})(1/2-1/b^{0.2}) \left(\sum_{j=K+1}^{\lceil\log N\rceil} \sum_{i=1}^{|N_{H'^{(j)}}(u)|/b} b 2^{-j}\right) \\
&\geq&& (1/2-2/b^{0.2}) \left(\sum_{j=K+1}^{\lceil\log N\rceil} \sum_{i=1}^{|N_{H'^{(j)}}(u)|/b} b 2^{-j}\right)\\
& = && (1/2-2/b^{0.2}) \sum_{v \in N_{H'}(u)}2^{-k_v} \\
& \geq &&  (1/2)(1-\gamma) \sum_{v \in N_{H'}(u)}2^{-k_v} \\
& \geq &&(1/2)(1-\gamma) \sum_{v \in N_{H}(u)}2^{-k_v} - \gamma/2
\end{align*}
and also that
\begin{align*} 
& && \sum_{v \in N_{H'}(u)\cap S}2^{-k_v} = \left(\sum_{j=K+1}^{\lceil\log N\rceil} \sum_{i=1}^{|N_{H'^{(j)}}(u)|/b} |B^{(j)}_i(u)\cap S|  2^{-j}\right) \\
&\leq&& ((1-1/b^{0.3})(1/2+1/b^{0.2})+(1/b^{0.3})) \left(\sum_{j=K+1}^{\lceil\log N\rceil} \sum_{i=1}^{|N_{H'^{(j)}}(u)|/b} b 2^{-j}\right) \\
&\leq&& (1/2+2/b^{0.2}) \left(\sum_{j=K+1}^{\lceil\log N\rceil} \sum_{i=1}^{|N_{H'^{(j)}}(u)|/b} b 2^{-j}\right)\\
& \leq && (1/2+2/b^{0.2}) \sum_{v \in N_{H}(u)}2^{-k_v} \\
& \leq && (1/2)(1+\gamma) \sum_{v \in N_{H}(u)}2^{-k_v}.
\end{align*}

Now, we argue that $\sum_{u\in U_{good}} imp_{u} \geq (1-\gamma)\sum_{u\in U} imp_{u}$. For each bad bucket $B^{(j)}_i(u)$, we have $\left( |S \cap B^{(j)}_i(u)| -  b/2\right)^2 \geq b^{1.6}$. And a node $u$ is called bad if we have $$\sum_{j=K+1}^{\lceil\log N\rceil} \sum_{i\in \mathcal{B}^{(j)}(u)} b 2^{-j} \geq \left(\sum_{j=K+1}^{\lceil\log N\rceil} \sum_{i=1}^{|N_{H'^{(j)}}(u)|/b} b 2^{-j}\right)/b^{0.3}.$$ Hence, for any bad node $u$, we have 
\begin{align*}
   \left(\frac{imp_{u}}{\sum_{j=K+1}^{\lceil\log N\rceil} |N_{H'^{(j)}}(u)| 2^{-j}} \cdot \sum_{j=K+1}^{\lceil\log N\rceil} 2^{-j} \left( \sum_{i=1}^{|N_{H'^{(j)}}(u)|/b}\left( |S \cap B^{(j)}_i(u)| -  b/2\right)^2 \right) \right) \geq imp_{u} \cdot \Theta(b^{0.3})
\end{align*}
Given that $\Phi_3(S)\leq 3.1$, we conclude $\sum_{u\in U_{good}} imp_{u} \geq (1-1/\Theta(b^{0.3}))\sum_{u\in U} imp_{u} \geq (1-\gamma)\sum_{u\in U} imp_{u}$.

Finally, we use $\Phi_3(S)\leq 3.1$ to conclude that $|S|\leq (2/3) |V| + b \leq (2/3)|V| + \poly(\log N)$. Consider the buckets $B'_1\sqcup B'_2\ldots \sqcup  B'_{\lfloor|V|/{b}\rfloor} \subseteq V$ that we had. Let us call bucket $B'_{i}$ bad if $||B'_{i}\cap S|-b/2|\geq 0.1b$. Notice that for each bad bucket, we have $(|B'_{i}\cap S|-b/2)^2\geq 0.01b^2$. Since $\Phi_3(S)= \frac{4}{b\lfloor|V|/{b}\rfloor} \sum_{i=1}^{\lfloor|V|/{b}\rfloor} \left( |S \cap B^{'}_i| -  b/2\right)^2\leq 3.1$, we know that at most $\Theta(1/b)<0.01$ fraction of these buckets can be bad. Hence, we conclude that $|S|\leq (1 + 0.1 + 0.01)|V|/2 + \poly(\log N)$.

\end{proof}

\begin{lemma}[Low Probability Regime]
\label{lem:hitting_set_low_probability}
Let $H$ be an $n$-vertex $m$-edge bipartite graph with bipartition $V(H) = U \sqcup V$. Consider a given upper bound $N\geq n$ and let $K=\lceil 100\log\log N \rceil$. For every $u \in U$, let $imp_u \in \mathbb{R}_{\geq 0}$ and for every $v \in V$, let $k_v \in \{K+1, \ldots, \lceil \log N \rceil\}$ and assume $|N_H(u)| \geq \lceil 10\log^{25}(N) \rceil$. There exists a deterministic parallel algorithm with work $(m+n)\poly(\log \log N) + \poly(\log N)$ and depth $\poly(\log N)$ that computes two subsets $S \subseteq V$ and $U_{good} \subseteq U$ such that $\sum_{u \in U_{good}} imp_u \geq 0.99\sum_{u \in U} imp_u$ and for every $u \in U_{good}$, $|N_H(u) \cap S|\cdot 2^{-K} \in (1\pm 0.01)(\sum_{v \in N_H(u)}2^{-k_v}) \pm 0.01$.
\end{lemma}

\begin{proof}


We define $H^{(0)} = H$; $U^{(0)} = U$; $V^{(0)} = V$ and $V^{(0)}_{fixed} = \emptyset$

For $i = 0,1,\ldots, I$ for $I=\lceil\log(N)\rceil$.

In round $i$, we invoke \cref{lem:low_prob_half} with input:
\begin{itemize}
    \item $H^{(i)}$ with bipartition $U^{(i)} \sqcup V^{(i)}$
    \item $\gamma^{(i)} = \max(0.0000001 \cdot 0.99^i, 0.0000001/\log(N))$ 
    \item $imp_u$ for every $u \in U^{(i)}$
    \item $k_v^{(i)} = k_v - i$ for every $v \in V^{(i)}$
\end{itemize}

As a result, we obtain a bipartite graph $H'^{(i)} \subseteq H^{(i)}$, a subset $U^{(i)}_{good} \subseteq U^{(i)}$,  and a subset $S^{(i)} \subseteq V^{(i)}$, such that for $H''^{(i)} := H'^{(i)}[U^{(i)}_{good} \sqcup S^{(i)}]$, it holds that

\begin{itemize}
\item[(I)] $|E(H''^{(i)})|+|S^{(i)}|\leq (2/3) \left(|E(H^{(i)})|+|V^{(i)}| \right) + \poly(\log N)$, 
\item[(II)] $\sum_{u \in U^{(i)}_{good}} imp_u \geq (1-\gamma^{(i)})\sum_{u \in U^{(i)}} imp_u$,
\item[(III)] for every node $u\in U^{(i)}$, we have $|N_{H'^{(i)}}(u)| \geq |N_{H^{(i)}}(u)| -\gamma^{(i)} 2^{K}$, and
\item[(IV)] for every node $u \in U^{(i)}_{good}$, we have $$\sum_{v \in N_{H''^{(i)}}(u)}2^{-k^{(i)}_v} \in ((1/2)(1 \pm \gamma^{(i)})\sum_{v \in N_{H^{(i)}}(u)}2^{-k^{(i)}_v} \pm \gamma^{(i)}).$$ 
\end{itemize}
Then, we set
$V^{(i+1)}_{fixed} = V^{(i)}_{fixed} \cup \{v \in S^{(i)} \colon k_v - (i+1) = K \}$

We define $H^{(i+1)} = H''^{(i)}[U^{(i)}_{good} \sqcup (S^{(i)} \setminus V^{(i+1)}_{fixed})]$, $U^{(i+1)} = U^{(i)}_{good}$ and $V^{(i+1)} = S^{(i)} \setminus V^{(i+1)}_{fixed}$. 

At the end of the for loop, we set the final output as $S = V^{(I+1)}_{fixed}$, and $U_{good} = U_{good}^{(I)}$.

\begin{claim}
 For every $u \in U_{good}$, $|N_H(u) \cap S| \cdot 2^{-K} \in (1\pm 0.01)(\sum_{v \in N_H(u)} 2^{-k_v} \pm 0.01)$.
\end{claim}
\begin{proof}

For each iteration $i$, let us define the set of \textit{discarded neighbors} of node $u$ in this iteration as $\Gamma^{(i)}(u)= (N_{H^{(i)}}(u) \cap S^{(i)}) \setminus (N_{H^{''(i)}}(u) \cap S^{(i)})$. 

Let us also define $\bar{V}^{(i)} = V^{(i)}\cup V^{(i)}_{fixed}$ for each $i$. From item (IV) above, for every node $u\in U^{(i)}_{good}$,  we have that
$$\sum_{v \in (N_{H}(u)\cap \bar{V}^{(i+1)}) \setminus (\bigcup_{j=0}^{i}\Gamma^{(j)}(u))} 2^{-\max\{K, k(v)-i-1\}} \in (1\pm \gamma^{(i)})\sum_{v \in (N_{H}(u)\cap \bar{V}^{(i)}) \setminus (\bigcup_{j=0}^{i-1}\Gamma^{j}(u))} 2^{-\max\{K, k(v)-i\}} \pm \gamma^{(i)}.$$
From this, applying it iteratively and using \Cref{lem:iterativeLoss}, which gives a cleaner bound for iterative appearance of multiplicative and additive losses, we can conclude that 
\begin{align*} 
\sum_{v \in (N_{H}(u)\cap S) \setminus (\bigcup_{j=0}^{I}\Gamma^{(j)}(u))} 2^{-K} = 
\sum_{v \in (N_{H}(u)\cap \bar{V}^{(I+1)}) \setminus (\bigcup_{j=0}^{I}\Gamma^{(j)}(u))} 2^{-K}
\in 
\prod_{i=0}^{I}(1\pm 2\gamma^{(i)}) \cdot \sum_{v \in N_{H}(u)} 2^{-k(v)} \pm \sum_{i=0}^{I} 2\gamma^{(i)}
\end{align*}

from item (III) above, for every node $u\in U^{(i)}_{good}$, we have that $|\Gamma^{(i)}(u)| \leq \gamma^{(i) } 2^{K}.$ Hence, for every node $u\in U_{good}$, we can conclude that $|\bigcup_{i} \Gamma^{(i)}(u)| \leq \sum_{i=0}^{I} \gamma^{(i) } 2^{K} $.

From the above, we see that 

\begin{align*} \sum_{v \in (N_{H}(u)\cap S)}  2^{-K} &\geq && \sum_{v \in (N_{H}(u)\cap S) \setminus (\bigcup_{j=0}^{I}\Gamma^{(j)}(u))} 2^{-K} \\
&\geq &&
\prod_{i=0}^{I}(1 -  2\gamma^{(i)}) \sum_{v \in N_{H}(u)} 2^{-k(v)} - (\sum_{i=0}^{I} 2\gamma^{(i)}) \end{align*}

Given that $\prod_{i=0}^{I}(1 -  2\gamma^{(i)})\geq 1-10^{-2}$ and $\sum_{i=0}^{I} 2\gamma^{(i)} \leq 10^{-2}$, we can conclude that $\sum_{v \in (N_{H}(u)\cap S)}  2^{-K} \geq (1-0.01)\sum_{v \in N_{H}(u)} 2^{-k(v)}-0.01$.

In addition, from item (III) above, for every node $u\in U^{(i)}_{good}$, we have that $|\Gamma^{(i)}(u)| \leq \gamma^{(i) } 2^{K}.$ Hence, for every node $u\in U_{good}$, we can conclude that $|\bigcup_{i} \Gamma^{(i)}(u)| \leq \sum_{i=0}^{I} \gamma^{(i) } 2^{K} $. Using this, and  the upper bound portion of \Cref{lem:iterativeLoss}, we conclude that

\begin{align*} \sum_{v \in (N_{H}(u)\cap S)}  2^{-K} &\leq && \sum_{v \in (N_{H}(u)\cap S) \setminus (\bigcup_{j=0}^{I}\Gamma^{(j)}(u))} 2^{-K} + |\bigcup_{i} \Gamma^{(i)}(u)| \cdot 2^{-K} \\ 
& \leq &&
\prod_{i=0}^{I}(1 +  2\gamma^{(i)}) \sum_{v \in N_{H}(u)} 2^{-k(v)} + \sum_{i=0}^I 2\gamma^{(i) } + (\sum_{i=0}^{I} \gamma^{(i)}) \end{align*}

Again, since $\prod_{i=0}^{I}(1 +  2\gamma^{(i)})\leq 1+10^{-2}$ and $\sum_{i=0}^{I} 3\gamma^{(i)} \leq 10^{-2}$, we know that $\sum_{v \in (N_{H}(u)\cap S)}  2^{-K} \leq (1+0.01)\sum_{v \in N_{H}(u)} 2^{-k(v)}+0.01$. 
Hence, having both sides, we conclude that $|N_{H}(u)\cap S|\cdot 2^{-K} \in (1\pm 0.01)\sum_{v \in N_{H}(u)} 2^{-k(v)}\pm 0.01$
\end{proof}
It remains to discuss the work and depth of the algorithm.
The overall work is  

\begin{align*}
& &&\sum_{i=0}^I(|E(H^{(i)})| + |U^{(i)}| + |V^{(i)}|)\poly(\log \log N)/(\gamma^{(i)})^{20} \\
&\leq &&\sum_{i=0}^I|U|\poly(\log \log N)/(\gamma^{(i)})^{20} + 
\sum_{i=0}^I(|E(H^{(i)})| + |V^{(i)}|)\poly(\log \log N)/(\gamma^{(i)})^{20} \\
&\leq &&\sum_{i=0}^I\frac{|E(H)|}{\lceil 10\log^{25}(N) \rceil}\poly(\log \log N)/(\gamma^{(i)})^{20} \\ 
& &&+ \sum_{i=0}^I \left(\max(0.9^i(|E(H)| + |V(H)|), \poly(\log N) \right) \poly(\log \log N)/(\gamma^{(i)})^{20} \\
&=&& (n+m)\poly(\log \log N) + \poly(\log N)
\end{align*}

and the overall depth is $I \cdot \poly(\log N) = \poly(\log N)$.
\end{proof}

\subsection{High Probability Regime}
\begin{lemma}[High Probability 1/2 Sampling]
\label{lem:high_probability_half}
Let $H$ be an $n$-vertex $m$-edge bipartite graph with bipartition $V(H) = U \sqcup V$, where each node has a unique identifier in $\{1,2, \ldots, 2n\}$, let $N\geq n$ be a given upper bound, and let $\gamma \in [1/\log(N), 0.01)$. For every $u \in U$, let $imp_u \in \mathbb{R}_{\geq 0}$. There exists a deterministic parallel algorithm with work $(m+n)\poly(\log \log N)/\gamma^{30}$ and depth $\poly(\log N)$ that computes two subsets $S \subseteq V$ and $U_{good} \subseteq U$ such that 

\begin{itemize}
\item $\sum_{u \in U_{good}} imp_u \geq (1-\gamma)\sum_{u \in U} imp_u$ 
\item for every $u \in U_{good}$, $|N_H(u) \cap S| \in (1 \pm \gamma)\frac{|N_H(u)|}{2} \pm (1/\gamma)^{7}$.
\end{itemize}
\end{lemma}
\begin{proof}

Let $b := \lceil (1/\gamma)^{6}\rceil$. For each node $u \in U$, we slightly adjust its neighborhood $N_H(u)$ by dropping up to $b-1$ of the neighbors, so that the remaining $|N_H(u)|$ is a multiple of $b$---notice that this causes at most an additive $\pm b \in \pm (1/\gamma)^{7}$ difference in the output size of $|N_H(u) \cap S|$, which is permitted in the second condition. For the remaining part of $N_{H}(u)$, we partition it arbitrarily into buckets $N_H(u) = B_1(u) \sqcup B_2(u) \sqcup \ldots \sqcup B_{|N_H(u)|/b}(u)$ of size $b$.

We want a derandomization such that, for nearly all buckets $B_i(u)$ over different $u$ (normalized and weighted by importance $imp(u)$, to be made precise), it holds that $|B_{i}(u) \cap S|$ is very close to $b/2$. In particular, we will use $(|B_{i}(u) \cap S| - b/2)^2$ as a measure of the drift from the $b/2$ target in this bucket, and perform the derandomization such that an aggregate of these drifts over all buckets and all $u$ is minimized. Concretely, our derandomization will try to approximately maintain the potential 
\begin{align*}
 \Phi(S) &:= &&\sum_{u \in U} \frac{imp_u}{|N_H(u)|} \cdot \left( \sum_{i=1}^{|N_H(u)|/b}\left( |S \cap B_i(u)| -  b/2\right)^2 \right) \\
 &=  &&\sum_{u \in U} \frac{imp_u}{|N_H(u)|} \cdot \sum_{i=1}^{|N_H(u)|/b}\left( 2\binom{|S \cap B_i(u)|}{2} -  (b-1) |S \cap B_i(u)| + b^2/4 \right) 
\end{align*}
as if set $S$ is chosen by including each node $v\in V$ in it with probability $1/2$. Notice that if $S$ is chosen randomly by including each $v\in V$ in $S$ with probability $1/2$, we would get 
\begin{align*}
\E[\Phi(S)] &=  \sum_{u \in U} \frac{imp_u}{|N_H(u)|} \cdot \sum_{i=1}^{|N_H(u)|/b}\left( 2 \E[\binom{|S \cap B_i(u)|}{2}] -  (b-1) \E[|S \cap B_i(u)|] + b^2/4 \right)  \\
&=  \sum_{u \in U} \frac{imp_u}{|N_H(u)|} \sum_{i=1}^{|N_H(u)|/b}\left((b)(b - 1)/4  - (b-1) b/2 + b^2/4\right) \\
&= \sum_{u \in U} \frac{imp_u}{|N_H(u)|} \sum_{i=1}^{|N_H(u)|/b}\frac{b}{4} \\
&= \sum_{u \in U} imp_u \cdot \frac{1}{4}.
\end{align*}

We next translate this potential to the utility/cost language of \Cref{lem:local_rounding} using an appropriate auxiliary mutli-graph $\bar{H}=(V, \bar{E})$. 
In particular, for each node $v \in V$, let us define
\[util(v) := (b-1)\sum_{u \in N_H(v)} \frac{imp_u}{|N_H(u)|}.\]
Notice that we have
\begin{align*}
\sum_{v \in V} util(v)
=\sum_{v \in V} (b-1)\sum_{u \in N_H(v)} \frac{imp_u}{|N_H(u)|}
= \sum_{u \in U} (b-1) \cdot imp_u.
\end{align*}
Furthermore, for every bucket $B_i(u)$ and two distinct vertices $v, v' \in B_i(u)$, let us add a new edge between $v$ and $v'$ in the edge multi-set $\bar{E}$ of the auxiliary graph $\bar{H}$ (notice that there might be several parallel such edges, one for each $u$). We set the cost of this newly added edge as
\[cost(\{v,v'\}) := 2 \cdot \frac{imp_u}{|N_H(u)|} \]
Notice that we have

\begin{align*}
\sum_{e\in \bar{E}} cost(e) &= &&\sum_{u \in U} \sum_{i = 1}^{|N_H(u)|/b}\sum_{v,v' \in \binom{B_i(u)}{2}}  2 \cdot \frac{imp_u}{|N_H(u)|} \\
&= &&\sum_{u\in U} \frac{|N_H(u)|}{b} \cdot \binom{b}{2} \cdot 2 \cdot \frac{imp_u}{|N_H(u)|} \\
& = &&\sum_{u \in U} (b-1) \cdot imp_u.
\end{align*}
We can write $\Phi(S) = \big(\sum_{u\in U}\frac{imp_u}{|N_H(u)|} (\sum_{i=1}^{|N_{H}(u)|/b} b^2/4)\big) - \big(\sum_{v \in S} util(v)\big) + \big(\sum_{e \in \bar{E}\cap \binom{S}{2}} cost(e)\big).$

Using \cref{lem:local_rounding} with $\eps = 1/(4(b-1))$, we compute a subset $S \subseteq V$ satisfying 

\begin{align*}
    \Phi(S) &\leq&& \big(\sum_{u\in U}\frac{imp_u}{|N_H(u)|} (\sum_{i=1}^{|N_{H}(u)|/b} b^2/4)\big) - \big(\frac{1}{2}\sum_{v \in V} util(v)\big) + \frac{1}{4}\big(\sum_{e \in \bar{E}} cost(e)\big)  + \eps \sum_{e \in \bar{E}} cost(e) \\
    &=&& \sum_{u \in U} imp_u \cdot \frac{1}{4}. + \eps \sum_{e \in \bar{E}} cost(e) \\
    &=&& \sum_{u \in U} imp_u \frac{1}{4} + \frac{1}{4(b-1)} \cdot \sum_{u \in U} imp_u (b-1)  \\    
    &\leq&& (\sum_{u \in U} imp_u \frac{1}{4}) (2).
\end{align*}

Let us call a bucket $B_i(u)$ bad if $||B_i(u)\cap S|-b/2|\geq b^{0.8}$. Let us call a node $u$ bad if more than $1/b^{0.3}$ fraction of its buckets are bad, and let $U_{good}$ be the set of nodes that are not bad in this sense. For any node $u\in U_{good}$, we can conclude that $|N_{H}(u)\cap S| \geq (1-1/b^{0.3})(1/2-1/b^{0.2}) |N_{H}(u)|\geq (1/2-2/b^{0.2}) |N_{H}(u)| \geq (1/2)(1-\gamma)|N_{H}(u)|$, and also that $|N_{H}(u)\cap S| \leq ((1-1/b^{0.3})(1/2+1/b^{0.2})+(1/b^{0.3}))|N_{H}(u)| \leq (1/2+2/b^{0.2}) |N_{H}(u)| \leq (1/2)(1+\gamma)|N_{H}(u)|$.  

Recall that
\begin{align*}
\Phi(S) &:= &&\sum_{u \in U} \frac{imp_u}{|N_H(u)|} \cdot \left( \sum_{i=1}^{|N_H(u)|/b}\left( |S \cap B_i(u)| -  b/2\right)^2 \right)
\end{align*}
and that for our chosen set $S$, we have $\Phi(S)\leq \sum_{u \in U} imp_u/2$. Notice that for each bad bucket, we have $\left( |S \cap B_i(u)| -  b/2\right)^2 \geq b^{1.6}$. Thus, for each bad node, its contribution to the potential is at least $imp_{u} b^{0.3}.$ Hence, we know that $\sum_{u\in U\setminus U_{good}} imp_{u} \leq (1/(b^{0.3})) \sum_{u\in U} imp_{u}$. That is, $\sum_{u\in U_{good}} imp_{u} \geq (1-1/(b^{0.3})) \sum_{u\in U} imp_{u} \geq (1-\gamma) \sum_{u\in U} imp_{u}$. 

It remains to discuss the work and depth of the algorithm. Note that the auxiliary multi-graph $\bar{H}$ has at most $n$ vertices and $m \cdot b$ edges, and we can construct $\bar{H}$ in $O(mb)$ work and $\poly(\log N)$ depth. With the same work and depth, we can also calculate the node utilities and edge costs.
The invocation \cref{lem:local_rounding} takes $(mb + n)\poly(\log \log N) $ work and $\poly(\log N/\eps)$ depth.
As $b := \lceil (1/\gamma)^6\rceil$ and $\eps = 1/(4(b-1))$, we can conclude that the overall work is 
 at most $(m+n)\poly(\log \log N)/\gamma^{30}$ and the depth is $\poly(\log N)$.

\end{proof}

\begin{lemma}[High Probability Regime]
\label{lem:hitting_set_high_probability}
Let $H$ be an $n$-vertex $m$-edge bipartite graph with bipartition $V(H) = U \sqcup V$, where each node has a unique identifier in $\{1,2,\ldots, 2n\}$, let $N\geq n$ be a given upper bound, and let $K=\lceil 100\log\log N \rceil$. For every $u \in U$, let $imp_u \in \mathbb{R}_{\geq 0}$ and for every $v \in V$, let $k_v \in \{1,2, \ldots, K\}$. There exists a constant $C\geq 1$ and a deterministic parallel algorithm with work $(m+n)\poly(\log \log N)$ and depth $\poly(\log N)$ that computes two subsets $S \subseteq V$ and $U_{good} \subseteq U$ such that $\sum_{u \in U_{good}} imp_u \geq 0.99\sum_{u \in U} imp_u$ and for every $u \in U_{good}$, $|N_H(u) \cap S| 
 \in [ (0.9\sum_{v \in N_H(u)}2^{-k_v}) - 0.1, C(\sum_{v \in N_H(u)}2^{-k_v}) + C]$.
\end{lemma}

\begin{proof}
During the process, in each iteration $i$, we will have a set $V^{(i)}_{alive}$ of $V$-vertices that remain alive, with the interpretation that $V^{(0)}_{alive} = V$, and per $i$, we have $V^{(i+1)}_{alive} \subseteq V^{(i)}_{alive}$. For each $v\in V^{(i)}_{alive}$, we define $k^{(i)}_{v} = \min\{k_{v}, K-i\}$. We would like to determine $V^{(i)}_{alive}$ such that, for each $i$, and for nearly all $u\in U$, we have that $\sum_{v\in N_{H}(u)\cap V^{(i+1)}_{alive}} 2^{-k^{(i+1)}_{v}}$ remains close to $\sum_{v\in N_{H}(u)\cap V^{(i)}_{alive}} 2^{-k^{(i)}_{v}}$ up to a small additive error such that these additive errors add up to a small constant over all iterations.

\medskip

As initialization, we define $U^{(0)} = U$; $V^{(0)} = \{v \in V \colon k_v = K\}$  and $H^{(0)} = H[U^{(0)} \sqcup V^{(0)}]$.

For $i = 0,1,\ldots, I$ for $I=K-50$, in round $i$ invoke \cref{lem:high_probability_half} with input:
\begin{itemize}
    \item $H^{(i)}$ with bipartition $U^{(i)} \sqcup V^{(i)}$
    \item $\gamma^{(i)} = 1/(100(K-i)^{2})$
    \item $imp_u$ for every $u \in U^{(i)}$
\end{itemize}
As a result, we obtain two sets $S^{(i)}$ and $U_{good}^{(i)}$ satisfying

\begin{itemize}
\item[(I)] $\sum_{u \in U^{(i)}_{good}} imp_u \geq (1-\gamma^{(i)})\sum_{u \in U^{(i)}} imp_u$, and
\item[(II)] for every $u \in U^{(i)}_{good}$, $|N_{H^{(i)}}(u) \cap S^{(i)}| \in (1 \pm \gamma^{(i)})\frac{|N_{H^{(i)}}(u)|}{2} \pm (1/\gamma^{(i)})^7$.
\end{itemize}
Then, we set
$V^{(i+1)}_{alive} = V^{(i)}_{alive} \setminus \left(V^{(i)} \setminus S^{(i)} \right)$, and we define $H^{(i+1)} = H[U^{(i)}_{good} \sqcup \{v \in V^{(i+1)}_{alive} \colon k^{(i+1)}_v = K-(i+1)\}]$, and go to iteration $i+1$ of the foor loop.

At the end of the for loop, we set the final output as $S = V^{(I+1)}_{alive}$, and $U_{good} = U_{good}^{(I)}$.

From (II), we have that $|N_{H^{(i)}}(u) \cap S^{(i)}| \in (1 \pm \gamma^{(i)})\frac{|N_{H^{(i)}}(u)|}{2} \pm (1/\gamma^{(i)})^7$. Hence, given that $i\leq K-50$, we have 
\begin{align*}
    \sum_{v\in N_{H}(u)\cap V^{(i+1)}_{alive}} 2^{-k^{(i+1)}_{v}} &\in &&(1\pm \gamma^{(i)})\sum_{v\in N_{H}(u)\cap V^{(i)}_{alive}} 2^{-k^{(i)}_{v}} \pm \frac{(1/\gamma^{(i)})^7}{2^{(K-(i+1))}} \\
    &\in &&(1\pm \gamma^{(i)})\sum_{v\in N_{H}(u)\cap V^{(i)}_{alive}} 2^{-k^{(i)}_{v}} \pm \gamma^{(i)}.
\end{align*}


Using this iteratively, and by applying \Cref{lem:iterativeLoss} which is our simple lemma for iterations of multiplicative and additive error, we conclude that 

\begin{align*}
   \sum_{v\in N_{H}(u)\cap V^{(I+1)}_{alive}} 2^{-k^{(I+1)}_{v}} &\in && (\prod_{j=0}^{i} (1+2\gamma^{(j)})) \sum_{v\in N_{H}(u)\cap V} 2^{-k_{v}} \pm (2 \sum_{j=0}^{I} \gamma^{(j)}) \\
   &\in && (1\pm0.1) \sum_{v\in N_{H}(u)\cap V} 2^{-k_{v}} \pm (0.1)
\end{align*}

In the end, we get

\begin{align*}
  |N_H(u) \cap S| = |N_H(u) \cap V^{(I+1)}_{alive}| \leq 2^{50}  \left((1+0.1) \sum_{v\in N_{H}(u)\cap V} 2^{-k_{v}} +  (0.1)  \right) \leq C \sum_{v\in N_{H}(u)\cap V} 2^{-k_{v}}+C.
  \end{align*}

  We also get

  \begin{align*}
  |N_H(u) \cap S| &= |N_H(u) \cap V^{(I+1)}_{alive}| \\
  &\geq \sum_{v \in N_H(u) \cap V^{(I+1)}_{alive}} 2^{-k_v^{(I+1)}} \\
  &\geq 0.9\sum_{v \in N_H(u)} 2^{-k_v} - 0.1.
  \end{align*}
These two conclude the argument for $|N_H(u) \cap S|  \in [ (\sum_{v \in N_H(u)}2^{-k_v}) - 0.1, C(\sum_{v \in N_H(u)}2^{-k_v}) + C]$, for every $u \in U_{good}$.

Next, we show that $\sum_{u \in U_{good}} \geq 0.99\sum_{u \in U} imp_u$. Indeed, we have

\begin{align*}
 \sum_{u \in U_{good}} imp_u = \sum_{u \in U^{(I)}_{good}} imp_u \geq \left( 1 - \sum_{j=0}^I \gamma^{(j)}\right) \sum_{u \in U} imp_u \geq 0.99 \sum_{u \in U} imp_u.
\end{align*}

The work can be upper bounded by

\[(m+n)\poly(\log \log N)/\left(\sum_{j=0}^I\left(\gamma^{(j)}\right)^{30}\right) = (m+n)\poly(\log \log N)\]

and the depth can be upper bounded by 

\[I \cdot \poly(\log N) = \poly(\log N).\]

\end{proof}

\subsection{Wrap Up}
\label{subsec:hitting_set_wrap}
Here, using the lemmas developed in the previous two subsections, we present the proof of \Cref{lem:hitting_set}. For convenience, we first restate the lemma.
\hittingSet*

\begin{proof}[Proof of \Cref{lem:hitting_set}]
Let $K :=\lceil 100\log \log N\rceil$. 
Let $H^{(low)} = H[U^{(low)} \sqcup V^{(low)}]$ with

\[V^{(low)} = \{v \in V \colon k_v \in \{K+ 1, K+2,\ldots, \lceil \log(N)\rceil\}\}\] 

and

\[U^{(low)} = \{u \in U \colon |N_H(u) \cap V^{(low)}| \geq \lceil 10 \log^{25} (N) \rceil\}.\]

We invoke \cref{lem:hitting_set_low_probability} with input $H^{(low)}$ and as a result we obtain two subsets $S^{(low)} \subseteq V^{(low)}$ and $U^{(low)}_{good} \subseteq U^{(low)}$ such that $\sum_{U^{(low)}_{good}} imp_u \geq 0.99\sum_{u \in U^{(low)}} imp_u$ and  for every $u \in U^{(low)}_{good}$, $|N_H(u) \cap S^{(low)}| \cdot 2^{-K} \in (1\pm 0.01)(\sum_{v \in N_H(u) \cap V^{(low)}} 2^{-k_v}) \pm 0.01$.

Let $H^{(high)} = H[U^{(high)} \sqcup V^{(high)}]$ with

\[V^{(high)} = (V \setminus V^{(low}) \cup S^{(low)} \] 

and

\[U^{(high)} = U^{(low)}_{good} \cup (U \setminus U^{(low)}).\]

We also define for every $v \in V^{(high)}$,

\[k^{(high)}_v = \min(k_v, K).\]







Next, We invoke \cref{lem:hitting_set_high_probability} with input $H^{(high)}$ and as a result we obtain two subsets $S \subseteq V^{(high)}$ and $U_{good} \subseteq U^{(high)}$ such that $\sum_{u \in U_{good}} imp_u \geq 0.99\sum_{u \in U^{(high)}} imp_u$ and  for every $u \in U_{good}$,  $|N_H(u) \cap S| 
 \in [ (\sum_{v \in N_H(u) \cap V^{(high)}}2^{-\min(k_v, K)}) - 0.1, C_{\ref{lem:hitting_set_high_probability}}(\sum_{v \in N_H(u) \cap V^{(high)}}2^{-\min(k_v, K)}) + C_{\ref{lem:hitting_set_high_probability}}]$.

 First, note that

 \begin{align*}
 \sum_{u \in U_{good}} imp_u 
 &\geq 0.99 \sum_{u \in U^{(high)}} imp_u  \\
 &= 0.99 \left(\sum_{u \in U^{(low)}_{good}} imp_u + \sum_{u \in U \setminus U^{(low)}} imp_u \right) \\
 &\geq 0.99 \left(0.99 \cdot\sum_{u \in U^{(low)}} imp_u + \sum_{u \in U \setminus U^{(low)}} imp_u \ \right) \\
 &\geq 0.9 \sum_{u \in U} imp_u.
 \end{align*}

Thus, it remains to verify that for every $u \in U_{good}$, we have $|N_H(u) \cap S| \in [(0.99)\sum_{v\in N_{H}(u)} 2^{-k_v} -0.5, C\sum_{v\in N_{H}(u)} 2^{-k_v} +C]$. 

First, let's consider a $u \in U_{good} \cap U^{(low)}_{good}$. We have

\begin{align*}
   |N_H(u) \cap S| 
   &\leq C_{\ref{lem:hitting_set_high_probability}}\sum_{v \in N_H(u) \cap V^{(high)}}2^{-\min(k_v, K)} + C_{\ref{lem:hitting_set_high_probability}} \\
   &= C_{\ref{lem:hitting_set_high_probability}}(\sum_{v \in N_H(u) \setminus V^{(low)}} 2^{-k_v} + |N_H(u) \cap S^{(low)}| 2^{-K})  + C_{\ref{lem:hitting_set_high_probability}}\\
 &\leq C_{\ref{lem:hitting_set_high_probability}}\left(\sum_{v \in N_H(u) \setminus V^{(low)}} 2^{-k_v} + \left((1+0.01)\sum_{v \in N_H(u) \cap V^{(low)}} 2^{-k_v} + 0.01\right)\right)  + C_{\ref{lem:hitting_set_high_probability}} \\
 &\leq C\sum_{v\in N_{H}(u)} 2^{-k_v} +C.
\end{align*}

and

\begin{align*}
   |N_H(u) \cap S| 
   &\geq 0.9\sum_{v \in N_H(u) \cap V^{(high)}}2^{-\min(k_v, K)} -0.1 \\
   &= 0.9 \sum_{v \in N_H(u) \setminus V^{(low)}} 2^{-k_v} + 0.9 |N_H(u) \cap S^{(low)}| 2^{-K} -0.1\\
 &\geq 0.9\sum_{v \in N_H(u) \setminus V^{(low)}} 2^{-k_v} + 0.9((1-0.01)\sum_{v \in N_H(u) \cap V^{(low)}} 2^{-k_v} - 0.01)  - 0.1 \\
 &\geq (0.5)\sum_{v\in N_{H}(u)} 2^{-k_v} -0.5.
\end{align*}

Next, consider a $u \in U_{good} \setminus U^{(low)}$.
We have

\begin{align*}
 |N_H(u) \cap S| 
   &\leq C_{\ref{lem:hitting_set_high_probability}}\sum_{v \in N_H(u) \cap V^{(high)}}2^{-\min(k_v, K)} + C_{\ref{lem:hitting_set_high_probability}}\\ 
   &\leq C_{\ref{lem:hitting_set_high_probability}}(\sum_{v \in N_H(u)} 2^{- k_v} + |N_H(u) \cap V^{(low)}| \cdot 2^{-K})  + C_{\ref{lem:hitting_set_high_probability}}\\
   &\leq C_{\ref{lem:hitting_set_high_probability}}(\sum_{v \in N_H(u)} 2^{- k_v} + 10)\leq C\sum_{v\in N_{H}(u)} 2^{-k_v} +C
\end{align*}

and

\begin{align*}
 |N_H(u) \cap S| 
   &\geq 0.9\sum_{v \in N_H(u) \cap V^{(high)}}2^{-\min(k_v, K)} -0.1 \\
   &\geq 0.9\sum_{v \in N_H(u) \setminus V^{(low)}} 2^{-k_v} - 0.1 \\
   &\geq 0.9(\sum_{v \in N_H(u)} 2^{-k_v} - 2^{-K}|N_H(u) \cap V^{(low)}|) - 0.1 \\
   &\geq 0.9(\sum_{v\in N_{H}(u)} 2^{-k_v} - 2^{-K} \cdot \lceil 10\log^{25}(N)\rceil) -0.1 \\
   &\geq 0.5\sum_{v\in N_{H}(u)} 2^{-k_v} - 0.5.
\end{align*}

Both the invocation of \cref{lem:hitting_set_low_probability} and \cref{lem:hitting_set_high_probability} take $(m+n) \poly(\log \log N) + \poly(\log N)$ work and $\poly(\log N)$ depth, so the overall work is $(m+n) \poly(\log \log N) + \poly(\log N)$  and the overall depth is $\poly(\log N)$.

\end{proof}

\section{Maximal Matching}
Here, we present our maximal matching algorithm, abstracted by \Cref{thm:MM}, using the derandomized hitting set result developed in the previous section. For convenience, let us restate the maximal matching result before discussing its proof.
\mm*
\begin{proof}
We first discuss the general algorithm outline and how the core ingredient in this outline can be performed using a very simple randomized parallel algorithm. Then, we explain how we use our deterministic parallel hitting set algorithm, presented in \Cref{lem:hitting_set}, to replace this randomized ingredient, thus building a deterministic parallel maximal matching algorithm.

\paragraph{Algorithm Outline} For each node $v$, let $d(v)=2^{\lceil deg(v)\rceil}$ where $deg(v)$ denotes its degree. We call this the rounded degree of node $v$. We sort vertices using \Cref{lem:sorting}, according to $\log(d(v))$, using $O(n\log\log n)$ work and $\poly(\log n)$ depth. Let $\Delta = \max_{v} d(v)$. Thus, each vertex is placed in one of $\log \Delta$ rounded degree categories.
Let $V_d$ be the set of vertices in rounded degree category $d$. The algorithm has $\log \Delta$ stages. In stage $i$, we target nodes of the maximum (remaining) rounded-degree category $V_d$, and ensure that by the end of the stage, no such node remains.

We next describe one stage. Consider stage $i$ and let $d=\Delta/2^{i}$. Here we target nodes $V_d$ of rounded degree $d$ and ensure that by the end of the stage, no such node remains. The stage consists of $O(\log n)$ identical iterations. Next, we describe one iteration. Each iteration involves some early clean up work, to remove some edges incident on previously matched vertices, and then it performs the core part of the iteration that computes an additional matching to be added to the output (eventually maximal) matching.

\paragraph{Each iteration in a stage---initial clean up} From the previous iterations, some nodes have been matched, and their edges have been marked for removal. We first remove some of those edges marked for removal. Concretely, let $V_d$ be the set of all vertices $v$ that were placed in the rounded degree $d(v)=d$ category, which we know from the past (without needing to work on the entire graph again and identify those vertices). For any node $v$ that was placed in $V_d$, we had that $deg(v)\in (d/2,d]$. First, we perform a scan on all edges of $v$ and read the number of edges marked for removal (from the previous iterations of the algorithm). If the number of edges of $v$ not marked for removal is below $d/3$, then we mark $v$ for downgrade from $v_{d}$ and we clean up its edges as follows: We first update its adjacency list by removing those edges marked for removal. The work for this is charged to those at least $d/2-d/3 = \Omega(d)$ edges that were marked for removal. In addition, we remove $v$ from $V_d$ and add it to the set $V_{d'}$ where $d'=2^{\lceil deg'(v)\rceil}$ and $deg'(v)$ denote the number of edges of $v$ that remain (i.e., were not removed). We do this using at most $O(|V_d| \log\log n)$ work, by sorting all nodes marked for downgrade from $V_d$, by their new rounded degree $d'$, using \Cref{lem:sorting}. We then move each continuous block of the sort to the appropriate set $V_{d'}$. At the end of this clean up, every node remaining in $V_d$ has at least $d/3$ edges that were not marked for removal (i.e., not incident on previously matched vertices).

\paragraph{Each iteration in a stage---core of the iteration} Next, we work on the set of edges incident on this cleaned up set $V_d$ and compute a matching among these edges. Notice that given $V_d$, we can read all these edges using $O(|V_d| d)$ work, and we can put those of them not marked for removal in an array of size $O(|V_d| d)$ by allocating $d$ locations for each node in $V_d$ and then doing one list ranking to remove empty locations, in $O(\log n)$ depth. Furthermore, using the transformation between graph representations mentioned in the preliminaries (\Cref{subsec:basics}), we can move from the list of edges to a compact representation of the subgraph with this set of edges and all nodes that have an edge in this list. This takes $O(|V_d| d\log\log n)$ work and $O(\log n)$ depth.

We compute in this subgraph a matching that is incident on at least $\Theta(|V_{d}| d)$ edges of the original graph, using $O(\Theta(|V_{d}| d) \poly(\log\log n))$ work. This is the core ingredient of the algorithm, for which we will give a randomized approach below, and then a derandomization via hitting set. For now, let us complete the outline assuming a black-box solution for this ingredient. Once we have computed the matching, we add the matching to the output (eventually maximal) matching, and we then mark all edges of the matching and and all edges incident on the matched nodes for removal from the graph (the actual removal is done in later iterations, as we described in the clean-up process above, so that we do not have to sort nodes by degree too frequently). 

Since each iteration removes $\Theta(|V_{d}| d)$ edges where $|V_d|$ is the set of nodes in the target category at the start of that iteration, after $O(\log n)$ iterations, no such node remains. Moreover, during these iterations, the number of remaining edges decays as a geometric series. Thus, the total work in the stage is at most $O(|V_d| d \poly(\log\log n))$. Since even the first iteration permanently removes $\Theta(|V_d| d)$ edges from the graph, we can charge this work to the removed edges and conclude that, over all stages, the work is bounded by $O(m\log\log n)$. What remains is how we compute a matching among edges incident on $V_{d}$ such that it is incident to $\Theta(|V_{d}| d)$ edges.

\paragraph{Randomized Ingredient} The randomized solution is simple: Mark each edge with probability $1/(4d)$ and put in the matching each marked edge that has no other marked edge incident on it. It is easy to see that each node in $V_{d}$ has a constant probability of having one incident edge in the matching. Hence, in expectation, $\Theta(|V_d| d)$ edges are incident on the computed matching.

\paragraph{Derandomization using Hitting Set} We first define the bipartite graph $H=(U\sqcup V, E)$ for the hitting set lemma. Let $E_d$ be the set of edges incident on $V_d$. We create $V$ by having one node for each edge $e\in E_d$. In addition, let $U$ have one node for each node in the original graph $G$ that is incident on an edge in $E_d$. 
We define the edges $E$ of the bipartite graph $H=(U\sqcup V, E)$ by connecting each $u\in U$ to all nodes in $V$ that represent the edges in $E_d$ that $u$ was incident on. 
We set the importance $imp_{u}$ of each node $u$ equal to the number of edges in $E_d$ that the node it represents in incident on.

For each $v\in V$,
we set $k_v = \log_{2} d - 4$. 
Notice that, for each node $u\in U$, we have $\sum_{v\in N_{H}(u)}2^{-k_v} \geq d \cdot (16/d) \leq 16$. In addition, for each node $w\in V_d$, we know that the node $u\in U$ that represents $w$ satisfies $\sum_{v\in N_{H}(u)}2^{-k_v} \geq d/3 \cdot (16/d) \geq 5$. 

The overall size of $H = |U|+|V|+|E(H)|$ is $\Theta(|V_d| d)$. \Cref{lem:hitting_set} thus runs in $\Theta(|V_d| d) \poly(\log\log n)$ work, and gives two subset $S\subseteq V$  and $U_{good} \subset U$ such that $\sum_{u\in U_{good}} deg_{E_d(u)} \geq 0.9\sum_{u\in} deg_{E_d(u)}$ and we have the following: for each node $u\in U_{good}$, we have $|N_{H}(u) \cap S| \leq \Theta(1)$, and for each node $u\in U_{good}$ that represents a node $w\in V_d$, we have $|N_{H}(u) \cap S| \geq 1$. These imply that $S$ gives a set of edges which induce a subgraph of constant degree and are incident on $\Theta(|V_d| d)$ edges. 

We can extract out of these a matching that is incident on $\Theta(|V_d| d)$ edges, as follows: Invoke \Cref{lem:coloring} on a graph where each edge is one node and two of these nodes are connected in their represented edges share an endpoint. This gives a $O(1)$ coloring of the edges using $\Theta(|V_d| d) \poly(\log\log n)$ work. Then, our desired matching is simply the color class, among these $O(1)$ colors, that maximizes the number of edges in $E_d$ incident on it. Therefore, the matching is incident on $\Theta(|V_d| d)$ edges of $E_d$, as desired.
\end{proof}

\section{Maximal Independent Set}
\subsection{Outer Shell}
In this section, we provide the proof of our maximal independent set algorithm, abstracted by \Cref{thm:MIS}, which we restate below for convenience.

\mis*
\begin{proof}
The core of \Cref{thm:MIS} is to compute, using $O((m+n) \poly(\log\log N))$ work and $\poly(\log N)$ depth, an independent set $S^* \subseteq V$ such that removing $S^* \cup N(S)$ from the graph removes $\Theta(|E|)$ edges, i.e., $\sum_{u\in S^*\cup N(S^*)}\geq \Theta(|E|)$. Let $I(S^*)$ be the set of vertices isolated in the graph $G[V\setminus (S^*\cup N(S))]$. We then add $S^*\cup I(S^*)$ and remove $S^*\cup N(S^*) \cup I(S^*)$ from the graph. The theorem follows readily by $O(\log n)$ iterations of this process, each time on the remaining graph. Notice that per iteration, we deal with a graph that has no isolated vertices. Hence, the work bound of the iteration is upper bounded by $O(m' \poly(\log\log N))$ where $m'$ is the number of edges remaining at that iteration. Due to the constant factor shrinkage of the number of edges, these work bounds form a geometric series, and their summation is dominated by the work in the first iteration, hence implying the work bound in the theorem statement. 

Let us now discuss the computation of $S^*$. We first invoke \Cref{lem:independentishSet} to compute a set $S\subseteq V$ such that $\sum_{u\in S\cup N_{G}(S)} deg(u) \geq 0.1 |E|$, and $\Delta(G[S])= O(1)$. Then we invoke \Cref{lem:coloring} to compute an $O(1)$ coloring of $\Delta(G[S])= O(1)$, using at most $O((m+n)\log\log N)$ work and $\poly(\log N)$ depth.
Consider sets $S_1$, $S_2$, \dots, $S_{O(1)}$ where $S_j$ denotes vertices of $S$ with color $j$. We set $S^* = \arg\max_{S_j} \left(\sum_{u\in S_j\cup N(S_j)} \; deg(u)\right)$. Since $\sum_{u\in S\cup N_{G}(S)} deg(u) \geq 0.1 |E|$, we know that $\sum_{j} \left(\sum_{u\in S_j\cup N(S_j)} deg(u)\right) \geq 0.1|E|$. As there are only $O(1)$ colors, we conclude that $$\max_{S_j} \sum_{u\in S_j\cup N(S_j)} deg(u)\geq \Theta(|E|).$$   
\end{proof}

\begin{lemma}[\textbf{Independentish Set Lemma}]
\label{lem:independentishSet}
Consider any $n$-node $m$-edge graph $G=(V, E)$. There is a deterministic parallel algorithm with $O((m+n) \poly(\log\log n))$ work and $\poly(\log n)$ depth that computes a set $S^*\subseteq V$ and an orientation of edges of $G[S^*]$ such that 
\begin{itemize}
    \item $\sum_{u\in S^*\cup N_{G}(S^*)} deg(u) \geq 0.1 |E|$, and
    \item the maximum outdegree in $G[S^*]$ is $O(1)$.
\end{itemize}
\end{lemma}
\begin{proof} The high-level idea is to derandomize one round of Luby's classic MIS algorithm\cite{luby1985simple}, via our derandomization abstracted in \Cref{lem:coreMIS}. Let us briefly review Luby's algorithm. 

\paragraph{Luby's Randomized MIS} Mark each node $v$ with probability $1/(10 deg(v))$. Orient each edge $\{u,v\}$ from $u$ to $v$ such that $(deg(v), id(v))>(deg(u), id(u))$. Let $S^*$ be the set of marked nodes that have no marked out-neighbor. Add $S^*$ to the output independent set, remove $S^*\cup N(S^*)$ from the graph. Add also any vertex that is now isolated, to the output independent set. Then repeat (so long some vertices remain). 

\paragraph{Analysis of Luby's Randomized MIS} The main property is to show that the removal of $S^*\cup N(S^*)$ removes a constant fraction of the edges in expectation, e.g., $\E[\sum_{u\in S^*\cup N_{G}(S^*)} deg(u)] \geq |E|/1000$. Call a node $u$ good if its indegree is at least $deg(u)/3$, and let $W$ be the set of good nodes. One can see that $\sum_{u\in W} deg(u) \geq |E|/2$. Moreover, using a simple pairwise analysis, we can show that each good node is removed with at least a constant probability, i.e., $\forall u\in W$, we have $\Pr[u\in S^*\cup N(S^*)] \geq 1/(500)$. Hence, $\E[\sum_{u\in S^*\cup N_{G}(S^*)} deg(u)] \geq |E|/1000$.
 

\paragraph{Minor modifications to fit our derandomization}  To align better with our derandomization, we make small adjustments to the algorithm and perform the analysis slightly differently. We set the marking probability for each node $v$ as $p_v=\min\{32/2^{\lceil\log_2(deg(v))\rceil}, 1\}$. In particular, each node $v$ of $deg(v)\leq 32$ is marked deterministically. Also, the marking probability can be viewed as $p_v=2^{-k_v}$ for $k_v\in\{0, 1, 2, \log N\}$. 
We split the set $W$ of good nodes into two categories: $U$ for those of degree at least $33$, and $U'$ for the rest.
For each good node $u\in U$, let $IN^*(u)$ be a subset of in-neighbors of $u$ such that $\sum_{v\in IN^*(u)} 2^{-k_v} \in [5, 7]$. Also notice that for each $v\in V$, we have $\sum_{w\in OUT(v)} 2^{-k_w} \leq 32$. Instead of finding the independent set $S^*$, which we did above, we do as follows: let $S$ be the set of marked nodes, and let $S^*$ be the set of marked nodes that have at most $10000$ marked outneighbors. Hence, we have the deterministic guarantee that the maximum outdegree of $G[S^*]$ is at most $10000$. For each node $u\in U'$, we have $Pr[u\in S^*]=1$. In addition, for each node $u\in U$ (for which we know $deg(u)\geq 33$), we have that (1) there is a probability $0.9$ that we have $|IN^*(u)\cap S|\geq 1$, and (2) with probability at least $0.6$, we have at most $10000$ outgoing edges in $(IN^*(v)\cap S)\times S$. Therefore, with probability at least $0.5$, node $u$ has at least one neighbor in $S^*$. Hence, overall, we conclude that $\E[\sum_{u\in S^*\cup N(S^*)} deg(u)] \geq (|E|/4)$. 

\paragraph{Derandomization} Except for the randomness in the process and the expectation in the bound on $\E[\sum_{u\in S\cup N(S)} deg(u)]$, the above would provide the desired guarantees of the lemma for the set $S$. We next discuss how we derandomize the process, by casting it as an instance to which we can apply \Cref{lem:coreMIS}. \Cref{lem:coreMIS} is the core derandomization of our MIS algorithm, and proving it will be the subject of the rest of this section. Set $U$ is the set of all good nodes $u$ such that $deg(u)\geq 33$, and the set $V$ is the set of all nodes. The graph $H$ is defined by setting $N_{H}(u)=IN^*(u)$. For each $u\in U$, set $imp_u=deg(u)$. Graph $G'$ is simply the set of all edges (in which we can omit edge directions). For each node $v\in V$, we compute $weight(v)=\sum_{u\in U; v\in IN^*(u)} deg(u)$ and we set $w(e) = weight(v)$ for each outgoing edge $e$ of $v$ (that is now placed in $G'$ as an undirected edge). We then apply \Cref{lem:coreMIS}, which runs in $O((m+n) \poly(\log\log n))$ work and $\poly(\log n)$ depth, providing two subsets $S \subseteq V$ and $U_{good} \subseteq U$ such that 
\begin{itemize}
    \item $\sum_{u \in U_{good}} imp_u \geq 0.9\sum_{u \in U} imp_u$,
    \item for every $u \in U_{good}$, $|N_H(u) \cap S| \in [1, C]$,
    \item $\sum_{e'\in \binom{S}{2}\cap E'} w(e') \leq C \cdot \sum_{e'=\{v, v'\} \in E'} w(e') 2^{-(k_v+k_{v'})}$,
\end{itemize}
where $C$ is a constant in \Cref{lem:coreMIS}. The first two properties provide an effect similar to property (1) discussed before. Concretely, they imply that $\sum_{u\in S\cup N(S)} deg(u)\geq (0.9)\sum_{u\in U} deg(u)$. Let us discuss how we use the third property. We have 
\begin{align*}
\sum_{e'=\{v, v'\} \in E'} w(e') 2^{-(k_v+k_{v'})} &= \sum_{u\in U} deg(u) \sum_{v\in IN^*(u)} 2^{-k_v} \left(\sum_{v'\in OUT(v)} 2^{-k_{v'}}\right) \\ &\leq \sum_{u\in U} deg(u) \left(\sum_{v\in IN^(*)} 2^{-k_v}\right) \cdot (32) \leq \sum_{u\in U} deg(u) \cdot (7) \cdot (32) \\
&\leq \sum_{u\in U} 250 \cdot deg(u).   
\end{align*} 
Hence, for the set $S$ chosen, we have $\sum_{e'\in \binom{S}{2}\cap E'} w(e') \leq \sum_{u\in U} 250C \cdot deg(u)$. By breaking the left hand side, we have $\sum_{u\in U} deg(u) \cdot \left(\sum_{v\in IN^*(u)\cap S} |OUT(v)\cap S|\right) \leq \sum_{u\in U} 250 C \cdot deg(u)$. Let us call a node $v\in S$ \textit{high-degree} if $OUT(v)\cap S\geq 2500C$, and we let $S^*$ be the nodes of $S$ that are not high-degree. Also, we call each node $u\in U$ \textit{overloaded} if it has a high-degree neighbor in $IN^*(u)\cap S$. Let $O\subseteq U$ be the set of overloaded nodes. We know that $\sum_{u\in O} deg(u) \left(\sum_{v\in IN^*(u)\cap S} |OUT(v)\cap S|\right) \geq \sum_{u\in O} 2500C \cdot deg(u)$. Hence, $\sum_{u\in O} 2500C \cdot deg(u) \leq \sum_{u\in U} 250C \cdot deg(u)$. Therefore, $$\sum_{u\in O} deg(u) \leq 0.1 \sum_{u\in U} deg(u).$$ On the other hand, notice that each node of $S$ that is not high-degree appears in $S^*$ and thus each node of $U\cap (S\cup N(S))$ that is not overloaded appears in $U\cap (S^* \cup N(S^*))$. This implies $U \cap (S^*\cup N(S^*))\supseteq U \cap ((S\cup N(S)) \setminus O)$. Therefore, we have 
\begin{align*}
    \sum_{u\in (S^*\cup N(S^*)\cap U)} deg(u) &\geq &&\sum_{u\in (S\cup N(S))\cap U} deg(u) &&&- &&&&\sum_{u\in O} deg(u) \\
    &\geq &&0.9 \left(\sum_{u\in U} deg(u)\right) &&&- &&&&0.1 \left(\sum_{u\in U} deg(u)\right) \\
    &\geq &&0.7 \left(\sum_{u\in U} deg(u)\right)\\   
\end{align*} 
 In addition, we know that each node $u\in U'$, which means it has $deg(u)\leq 32$, is in $S$ deterministically and has outdegree at most $32$. Thus, $u\in S^*$. Hence, we conclude 
 
 \begin{align*} 
 \sum_{u\in (S^*\cup N(S^*))} deg(u) &\geq&& 
 \sum_{u\in (S^*\cup N(S^*)\cap U)} deg(u) + \sum_{u \in U'} deg(u) \\
 &\geq &&0.7\sum_{u \in U} deg(u) + \sum_{u \in U'} deg(u) \\ &\geq&& 0.7 \sum_{u\in W} deg(u) \\
 &\geq&& |E|/4.
 \end{align*}
\end{proof}

\begin{restatable}{lemma}{MIShittingSet}\textnormal{(\textbf{Core Hitting Set Lemma for MIS})}
\label{lem:coreMIS}
Let $H$ be an $n$-vertex $m$-edge bipartite graph with bipartition $V(H) = U \sqcup V$. Let $N\geq n$ be a given upper bound. For every $u \in U$, let $imp_u \in \mathbb{R}_{\geq 0}$ and for every $v \in V$, let $k_v \in \{0,1,\ldots, \lceil \log N \rceil\}$ such that for every $u \in U$, we have $\sum_{v \in N_H(u)}2^{-k_v} \in [5,10]$. Consider also an additional graph $G'=(V, E')$ on the same set $V$ of vertices, with $m'$ edges, where each edge $e'\in E'$ has a weight $w(e')\in \mathbb{R}_{\geq 0}$.

There exists a constant $C\geq 1$ and a deterministic parallel algorithm with work $(m+m'+n)\poly(\log \log N)$ and depth $\poly(\log N)$ that computes two subsets $S \subseteq V$ and $U_{good} \subseteq U$ such that 
\begin{itemize}
    \item $\sum_{u \in U_{good}} imp_u \geq 0.9\sum_{u \in U} imp_u$,
    \item for every $u \in U_{good}$, $|N_H(u) \cap S| \in [1, C]$,
    \item $\sum_{e'\in \binom{S}{2}\cap E'} w(e') \leq C \cdot \sum_{e'=\{v, v'\} \in E'} w(e') 2^{-(k_v+k_{v'})}$
\end{itemize}
\end{restatable}

The proof of this lemma is presented later in \Cref{subsec:MISwrap}, after we show the key ingredient lemmas in the following subsections.

\subsection{Low Probability Regime}
\begin{lemma}[Low Probability Regime]
\label{lem:mis_set_low_probability}
Let $H$ be an $n$-vertex $m$-edge bipartite graph with bipartition $V(H) = U \sqcup V$, where each node has a unique identifier in $\{1, 2, \dots 2n\}$, let $N\geq n$ be a given upper bound, and let $K := \lceil 100\log\log(N)\rceil$. For every $u \in U$, let $imp_u \in \mathbb{R}_{\geq 0}$ and for every $v \in V$, let $k_v \in \{K+1,K+2,\ldots, \lceil \log N \rceil\}$ such that for every $u \in U$, we have $\sum_{v \in N_H(u)}2^{-k_v} \leq 10$ and $|N_H(u)| \geq \lceil 10\log^{25}(N) \rceil$. Consider also an additional graph $G'=(V, E')$ on the same set $V$ of vertices, with $m'$ edges, where each edge $e'\in E'$ has a weight $w(e')\in \mathbb{R}_{\geq 0}$ and each vertex $v\in V$ has a weight $w(v)\in \mathbb{R}_{\geq 0}$.

There exists a deterministic parallel algorithm with work $(m+m'+n)\poly(\log \log N) + \poly(\log n)$ and depth $\poly(\log N)$ that computes two subsets $S \subseteq V$ and $U_{good} \subseteq U$ such that 
\begin{itemize} 
\item $\sum_{u \in U_{good}} imp_u \geq 0.99\sum_{u \in U} imp_u$, and 
\item for every $u \in U_{good}$, $|N_H(u) \cap S|\cdot 2^{-K} \in (\sum_{v \in N_H(u)}2^{-k_v}) \pm 0.01$,
\item $\sum_{e'\in \binom{S}{2}\cap E'} w(e') 2^{-2K} + \sum_{v\in S} w(v) 2^{-K} \leq 2 \cdot \left( \sum_{e'=\{v, v'\} \in E'} w(e') 2^{-(k_v+k_{v'})} + \sum_{v\in V} w(v) 2^{-k_v} \right)$
\end{itemize}
\end{lemma}

\begin{proof}


We define $H^{(0)} = H$; $U^{(0)} = U$; $V^{(0)} = V, V^{(0)}_{fixed} = \emptyset$ and $k_v^{(i)} = \max(K,k_v - i)$ for every $v \in V$

For $i = 0,1,\ldots, I$ for $I=\lceil\log(N)\rceil$.

In round $i$, we invoke \cref{lem:mis_low_prob_half} with input:
\begin{itemize}
    \item $H^{(i)}$ with bipartition $U^{(i)} \sqcup V^{(i)}$
    \item $\gamma^{(i)} = \max(0.0000001 \cdot 0.99^i, 0.0000001/\log(N))$ 
    \item $imp_u$ for every $u \in U^{(i)}$
    \item $k_v^{(i)}$ for every $v \in V^{(i)}$
    \item $G'^{(i)} = G'[V^{(i)}]$
    \item For every $\{v,v'\} \in E(G'^{(i)})$, $w^{(i)}(\{v,v'\}) = w(\{v,v'\}) \cdot  2^{-k_v^{(i)}}\cdot 2^{-k_{v'}^{(i)}}$
    \item For every $v \in V^{(i)}$, $w^{(i)}(v) = w(v) \cdot 2^{-k_v^{(i)}}+ \sum_{v' \in N_{G'}(v) \cap \left(V^{(i)}_{fixed} \setminus V^{(i)} \right)} w(\{v,v'\}) \cdot 2^{-(k^{(i)}_v + k^{(i)}_{v'})}$
\end{itemize}

As a result, we obtain a bipartite graph $H'^{(i)} \subseteq H^{(i)}$, a subset $U^{(i)}_{good} \subseteq U^{(i)}$,  and a subset $S^{(i)} \subseteq V^{(i)}$, such that for $H''^{(i)} := H'^{(i)}[U^{(i)}_{good} \sqcup S^{(i)}]$, it holds that

\begin{itemize}
\item[(I)] $|E(H''^{(i)})|+|S^{(i)}|\leq (2/3) \left(|E(H^{(i)})|+|V^{(i)}| \right) + \poly(\log N)$, 
\item[II] $|E' \cap \binom{S^{(i)}}{2}| \leq (2/3)|E' \cap \binom{V^{(i)}}{2}| + \poly(\log N)$
\item[(III)] $\sum_{u \in U^{(i)}_{good}} imp_u \geq (1-\gamma^{(i)})\sum_{u \in U^{(i)}} imp_u$,
\item[(IV)] for every node $u\in U^{(i)}$, we have $|N_{H'^{(i)}}(u)| \geq |N_{H^{(i)}}(u)| -\gamma^{(i)} 2^{K}$, and
\item[(V)] for every node $u \in U^{(i)}_{good}$, we have $\sum_{v \in N_{H''^{(i)}}(u)}2^{-k^{(i)}_v} \in ((1/2)(1 \pm \gamma^{(i)})\sum_{v \in N_{H^{(i)}}(u)}2^{-k^{(i)}_v} \pm \gamma^{(i)})$. 
\item[(VI)] $\sum_{e'\in \binom{S^{(i)}}{2}\cap E'} w^{(i)}(e') \cdot 4  + \sum_{v\in S^{(i)}} w^{(i)}(v) \cdot 2  \leq (1+\gamma^{(i)}) \cdot \left( \sum_{e'=\{v, v'\} \in \binom{V^{(i)}}{2} \cap E'} w^{(i)}(e') + \sum_{v\in V^{(i)}} w^{(i)}(v)   \right).$
\end{itemize}
Then, we set
$V^{(i+1)}_{fixed} = V^{(i)}_{fixed} \cup \{v \in S^{(i)} \colon k_v - (i+1) = K \}$

We define $H^{(i+1)} = H''^{(i)}[U^{(i)}_{good} \sqcup (S^{(i)} \setminus V^{(i+1)}_{fixed})]$, $U^{(i+1)} = U^{(i)}_{good}$ and $V^{(i+1)} = S^{(i)} \setminus V^{(i+1)}_{fixed}$. 

At the end of the for loop, we set the final output as $S = V^{(I+1)}_{fixed}$, and $U_{good} = U_{good}^{(I)}$.

We start with proving the third bullet point, namely that 

\[\sum_{e'\in \binom{S}{2}\cap E'} w(e') 2^{-2K} + \sum_{v\in S} w(v) 2^{-K} \leq 2 \cdot \left( \sum_{e'=\{v, v'\} \in E'} w(e') 2^{-(k_v+k_{v'})} + \sum_{v\in V} w(v) 2^{-k_v} \right).\]

To that end, we first show that

\begin{align*}&\sum_{\{v,v'\} \in \binom{V^{(i+1)} \cup V^{(i+1)}_{fixed}}{2}\cap E'}w(\{v,v'\}) 2^{-k^{(i+1)}_v} \cdot 2^{-k^{(i+1)}_{v'}} + \sum_{v \in V^{(i+1)} \cup V^{(i+1)}_{fixed}} w(v) \cdot 2^{-k_v^{(i+1)}} \\
&\leq  (1+\gamma^{(i)})\sum_{\{v,v'\} \in \binom{V^{(i)} \cup V^{(i)}_{fixed}}{2}\cap E'} w(\{v,v'\}) 2^{-k^{(i)}_v} \cdot 2^{-k^{(i)}_{v'}} + \sum_{v \in V^{(i)} \cup V^{(i)}_{fixed}} w(v) \cdot 2^{-k_v^{(i)}}
\end{align*}

Indeed, we have

\begin{align*}
&\sum_{\{v,v'\} \in \binom{V^{(i+1)} \cup V^{(i+1)}_{fixed}}{2}\cap E'}w(\{v,v'\})\cdot 2^{-k^{(i+1)}_v} \cdot 2^{-k^{(i+1)}_{v'}} + \sum_{v \in V^{(i+1)} \cup V^{(i+1)}_{fixed}} w(v) \cdot 2^{-k_v^{(i+1)}} \\ 
&= \sum_{\{v,v'\} \in \binom{V^{(i)}_{fixed}}{2}\cap E'}w(\{v,v'\}) \cdot 2^{-k^{(i)}_v} \cdot 2^{-k^{(i)}_{v'}} \\
&+ 2 \cdot \sum_{\{v,v'\} \in E' \colon |\{v,v'\} \cap S^{(i)}|= 1 \text{ and } |\{v,v'\} \cap V^{(i)}_{fixed}|= 1}w(\{v,v'\}) 2^{-k^{(i)}_v} \cdot 2^{-k^{(i)}_{v'}} \\
&+ 4 \cdot \sum_{\{v,v'\} \in \binom{S^{(i)}}{2}\cap E'}w(\{v,v'\}) \cdot 2^{-k^{(i)}_v} \cdot 2^{-k^{(i)}_{v'}} \\
&+ \sum_{v \in V^{(i)}_{fixed}} w(v) \cdot 2^{-k_v^{(i)}} 
+ 2 \cdot \sum_{v \in S^{(i)}} w(v) \cdot 2^{-k_v^{(i)}}
\\
&= \sum_{\{v,v'\} \in \binom{V^{(i)}_{fixed}}{2}\cap E'}w(\{v,v'\}) \cdot 2^{-k^{(i)}_v} \cdot 2^{-k^{(i)}_{v'}} 
+ \sum_{v \in V^{(i)}_{fixed}} w(v) \cdot 2^{-k_v^{(i)}} \\
&+ 4 \cdot \sum_{\{v,v'\} \in \binom{S^{(i)}}{2}\cap E'}w^{(i)}(\{v,v'\}) 
+ 2 \cdot \sum_{v \in  S^{(i)}} w^{(i)}(v) \\
&\leq \sum_{\{v,v'\} \in \binom{V^{(i)}_{fixed}}{2}\cap E'}w(\{v,v'\}) \cdot 2^{-k^{(i)}_v} \cdot 2^{-k^{(i)}_{v'}} 
+ \sum_{v \in V^{(i)}_{fixed}} w(v) \cdot 2^{-k_v^{(i)}} \\
&+ (1+\gamma^{(i)}) \left( \cdot \sum_{\{v,v'\} \in \binom{V^{(i)}}{2}\cap E'}w^{(i)}(\{v,v'\}) 
+ \sum_{v \in  V^{(i)}} w^{(i)}(v) \right)\\
&\leq (1+\gamma^{(i)})\sum_{\{v,v'\} \in \binom{V^{(i)} \cup V^{(i)}_{fixed}}{2}\cap E'} w(\{v,v'\}) 2^{-k^{(i)}_v} \cdot 2^{-k^{(i)}_{v'}} + \sum_{v \in V^{(i)} \cup V^{(i)}_{fixed}} w(v) \cdot 2^{-k_v^{(i)}}.
\end{align*}

Therefore,

\begin{align*}
 & &&\sum_{e'\in \binom{S}{2}\cap E'} w(e') 2^{-2K} + \sum_{v\in S} w(v) 2^{-K} \\
  &= &&\sum_{\{v,v'\}\in \binom{V^{(I+1)} \cup V^{(I+1)}_{fixed}}{2}\cap E'} w(\{v,v'\}) 2^{-k^{(I+1)}_v}2^{-k^{(I+1)}_{v'}} + \sum_{v\in V^{(I+1)} \cup V^{(I+1)}_{fixed}} w(v) 2^{-k^{(I+1)}_v} \\
 &\leq &&\prod_{i=0}^{I} \left(1 +\gamma^{(i)}\right) \sum_{\{v,v'\}\in \binom{V^{(0)} \cup V^{(0)}_{fixed}}{2}\cap E'} w(\{v,v'\}) 2^{-k^{(0)}_v}2^{-k^{(0)}_{v'}} + \sum_{v\in V^{(0)} \cup V^{(0)}_{fixed}} w(v) 2^{-k^{(0)}_v} \\
 &\leq &&2 \cdot \left( \sum_{e'=\{v, v'\} \in E'} w(e') 2^{-(k_v+k_{v'})} + \sum_{v\in V} w(v) 2^{-k_v} \right).
\end{align*}

\begin{claim}
 For every $u \in U_{good}$, $|N_H(u) \cap S| \cdot 2^{-K} \in (\sum_{v \in N_H(u)} 2^{-k_v} \pm 0.01)$.
\end{claim}
\begin{proof}

For each iteration $i$, let us define the set of \textit{discarded neighbors} of node $u$ in this iteration as $\Gamma^{(i)}(u)= (N_{H^{(i)}}(u) \cap S^{(i)}) \setminus (N_{H^{''(i)}}(u) \cap S^{(i)})$. 

Let us also define $\bar{V}^{(i)} = V^{(i)}\cup V^{(i)}_{fixed}$ for each $i$. From item (IV) above, for every node $u\in U^{(i)}_{good}$,  we have that
$$\sum_{v \in (N_{H}(u)\cap \bar{V}^{(i+1)}) \setminus (\bigcup_{j=0}^{i}\Gamma^{(j)}(u))} 2^{-\max\{K, k(v)-i-1\}} \in (1\pm \gamma^{(i)})\sum_{v \in (N_{H}(u)\cap \bar{V}^{(i)}) \setminus (\bigcup_{j=0}^{i-1}\Gamma^{j}(u))} 2^{-\max\{K, k(v)-i\}} \pm \gamma^{(i)}.$$
From this, applying it iteratively, we can conclude that 
\begin{align*} 
\sum_{v \in (N_{H}(u)\cap S) \setminus (\bigcup_{j=0}^{I}\Gamma^{(j)}(u))} 2^{-K} = 
\sum_{v \in (N_{H}(u)\cap \bar{V}^{(I+1)}) \setminus (\bigcup_{j=0}^{I}\Gamma^{(j)}(u))} 2^{-K}
\in 
\prod_{i=0}^{I}(1\pm \gamma^{(i)}) \cdot \sum_{v \in N_{H}(u)} 2^{-k(v)} \pm \sum_{i=0}^{I} \gamma^{(i)}
\end{align*}

In addition, from item (III) above, for every node $u\in U^{(i)}_{good}$, we have that $|\Gamma^{(i)}(u)| \leq \gamma^{(i) } 2^{K}.$ Hence, for every node $u\in U_{good}$, we can conclude that $|\bigcup_{i} \Gamma^{(i)}(u)| \leq \sum_{i=0}^{I} \gamma^{(i) } 2^{K} $.

From the above two, we see that 

\begin{align*} \sum_{v \in (N_{H}(u)\cap S)}  2^{-K} &\geq && \sum_{v \in (N_{H}(u)\cap S) \setminus (\bigcup_{j=0}^{I}\Gamma^{(j)}(u))} 2^{-K} \\
&\geq &&
\prod_{i=0}^{I}(1 -  2\gamma^{(i)}) \sum_{v \in N_{H}(u)} 2^{-k(v)} - (\sum_{i=0}^{I} 2\gamma^{(i)}) \end{align*}

Given the initial upper bound $\sum_{v\in N_{H}(u)} 2^{-k(v)}\leq 10$ and since $\prod_{i=0}^{I}(1 -  2\gamma^{(i)})\geq 1-10^{-4}$ and $\sum_{i=0}^{I} 2\gamma^{(i)}) \leq 10^{-4}$, we can conclude that $\sum_{v \in (N_{H}(u)\cap S)}  2^{-K} \geq \sum_{v \in N_{H}(u)} 2^{-k(v)}-0.01$.

Similarly, we also have that

\begin{align*} \sum_{v \in (N_{H}(u)\cap S)}  2^{-K} &\leq && \sum_{v \in (N_{H}(u)\cap S) \setminus (\bigcup_{j=0}^{I}\Gamma^{(j)}(u))} 2^{-K} + |\bigcup_{i} \Gamma^{(i)}(u)| \cdot 2^{-K} \\ 
& \leq &&
\prod_{i=0}^{I}(1 +  2\gamma^{(i)}) \sum_{v \in N_{H}(u)} 2^{-k(v)} + (\sum_{i=0}^{I} \gamma^{(i)}) + \sum_{i=0}^I 2\gamma^{(i) } \end{align*}

Again, given the initial upper bound $\sum_{v\in N_{H}(u)} 2^{-k(v)}\leq 10$ and since $\prod_{i=0}^{I}(1 +  2\gamma^{(i)})\leq 1+10^{-4}$ and $\sum_{i=0}^{I} 3\gamma^{(i)} \leq 10^{-4}$, we can conclude that $\sum_{v \in (N_{H}(u)\cap S)}  2^{-K} \leq \sum_{v \in N_{H}(u)} 2^{-k(v)}+0.01$.

Hence, having both sides of the bound, we conclude that $|N_{H}(u)\cap S|\cdot 2^{-K} \in \sum_{v \in N_{H}(u)} 2^{-k(v)}\pm 0.01$

\end{proof}
It remains to discuss the work and depth of the algorithm.
The overall work is  

\begin{align*}
& &&\sum_{i=0}^I(|E(H^{(i)})| + |U^{(i)}| + |V^{(i)}| + |E(G'^{(i)})|)\poly(\log \log N)/(\gamma^{(i)})^{20} \\
&\leq &&\sum_{i=0}^I|U|\poly(\log \log N)/(\gamma^{(i)})^{20} + 
\sum_{i=0}^I(|E(H^{(i)})| + |V^{(i)}| +  |E(G'^{(i)})|)\poly(\log \log N)/(\gamma^{(i)})^{20} \\
&\leq &&\sum_{i=0}^I\frac{|E(H)|}{\lceil 10\log^{25}(N) \rceil}\poly(\log \log N)/(\gamma^{(i)})^{20} \\
& &&+  \sum_{i=0}^I \left(\max(0.9^i(|E(H)| + |V(H)| +  |E(G'^{(i)})|), \poly(\log N) \right) \poly(\log \log N)/(\gamma^{(i)})^{20} \\
&= &&(n+m + m')\poly(\log \log N) + \poly(\log N)
\end{align*}

and the overall depth is $I \cdot \poly(\log N) = \poly(\log N)$.
\end{proof}

\begin{lemma}[Low Probability 1/2 Sampling]
\label{lem:mis_low_prob_half}
Let $H$ be an $n$-vertex $m$-edge bipartite graph with bipartition $V(H) = U \sqcup V$, where each node has a unique identifier in $\{1, 2, \dots 2n\}$. Let $N\geq n$ be a given upper bound. For every $u \in U$, let $imp_u \in \mathbb{R}_{\geq 0}$ and for every $v \in V$, let $k_v \in \{K+1,K+2,\ldots, \lceil \log N \rceil\}$, where $K=\lceil 100\log\log N\rceil$, such that for every $u \in U$, we have $\sum_{v \in N_H(u)}2^{-k_v} \leq 10$. Consider also an additional graph $G'=(V, E')$ on the same set $V$ of vertices, with $m'$ edges, where each edge $e'\in E'$ has a weight $w(e')\in \mathbb{R}_{\geq 0}$ and each vertex $v\in V$ has a weight $w(v)\in \mathbb{R}_{\geq 0}$. Also, consider a given $\gamma \in [\Theta(1/\log N), 0.01)$.

There exists a deterministic parallel algorithm with work $(m+m'+n)\poly(\log \log N)/\gamma^{20}$ and depth $\poly(\log N)$ that computes a bipartite graph $H'\subseteq H$ with $V(H')=U_{good}\sqcup V'$, with $U_{good} \subseteq U$ and $S\subseteq V$ such that

\begin{itemize}
\item $|E(H')|+|S|\leq (2/3) \left(|E(H)|+|V| \right) + \poly(\log N)$, 
\item $|E'\cap \binom{S}{2}| \leq (2/3) |E'| + \poly(\log N)$,
\item $\sum_{u \in U_{good}} imp_u \geq (1-\gamma)\sum_{u \in U} imp_u$,
\item for every node $u\in U_{good}$, we have $|N_{H'}(u)| \geq |N_{H}(u)| -\gamma2^{K}$, 
\item for every node $u \in U_{good}$, we have $\sum_{v \in N_{H''}(u)} 2^{-k_v} \in ((1/2)(1\pm \gamma) \sum_{v \in N_H(u)}2^{-k_v} \pm \gamma/2)$, and
\item $\sum_{e'\in \binom{S}{2}\cap E'} w(e') \cdot 4  + \sum_{v\in S} w(v) \cdot 2 \leq (1+\gamma) \cdot \left( \sum_{e'=\{v, v'\} \in E'} w(e')  + \sum_{v\in V} w(v)  \right)$.
\end{itemize}

\end{lemma}
\begin{proof}
We partition $V$ into parts $V^{(j)}$ for $j\in \{K+1, \ldots, \lceil \log N \rceil\}$ where we put each node $v\in V$ in $V^{(k_v)}$. Also, let $H^{(j)}$ be the induced subgraph by $U\sqcup V^{(j)}$.
We can compute this partition using the sorting lemma (\cref{lem:sorting}) with $O(n \log \log N)$ work and $\poly(\log N)$ depth.

Let $b := \lfloor \min\{(1/\gamma)^{6}, \gamma 2^{K-1}/\lceil\log N \rceil\}\rfloor$. We define a graph $H'^{(j)}$ by dropping a few edges from $H^{(j)}$: For each node $u \in U$ and each $j\in \{K+1, \ldots, \lceil \log N \rceil\}$, we set $N_{H'^{(j)}}(u)$ by dropping up to $b-1$ neighbors of $u$ in $H^{(j)}$, such that $|N_{H'^{(j)}}(u)|$ is divisible by $b$. We let $H'$ be the graph by the union of the edges of $H'^{(j)}$ for all $j$ and $V(H') = U \sqcup V$. 

Then, for each $j$ and $u \in U$, we partition $N_{H'^{(j)}}(u)$ arbitrarily into buckets $N_{H'^{(j)}}(u) = B^{(j)}_1(u) \sqcup B^{(j)}_2(u) \sqcup \ldots \sqcup B^{(j)}_{|N_{H'^{(j)}}(u)|/b}(u)$ of size $b$. Notice that $|N_{H^{(j)}}(u) \setminus N_{H'^{(j)}}(u)| < b$ and therefore $|N_{H'}(u)| - |N_H(u)| \geq - \lceil \log(N)\rceil \cdot b \geq - \gamma  \cdot 2^K/2 $. This is permitted in the third item of the lemma. Furthermore, the total probability contribution of these dropped neighbors in $\sum_{v \in N_{H}(u)} 2^{-k_v}$ is at most $\gamma/4$, and thus this makes $\pm \gamma/2$ additive error in the fourth item of the lemma, which is permitted. 

Our derandomization will use five potential functions for choosing the set $S\subseteq V$: 

\begin{align*}
 \Phi_1(S) &:= &&\frac{4}{\sum_{u\in U} \sum_{j=K+1}^{\lceil\log N\rceil} |N_{H'^{(j)}}(u)|} \cdot \sum_{u \in U} \sum_{j=K+1}^{\lceil\log N\rceil} \left( \sum_{i=1}^{|N_{H'^{(j)}}(u)|/b}\left( |S \cap B^{(j)}_i(u)| -  b/2\right)^2 \right) 
\end{align*}

\begin{align*}
 \Phi_2(S) &:= &&\big(\frac{1}{\sum_{u\in U} imp_u}\big) \cdot \\
 & &&\sum_{u \in U} \left(4 \cdot \frac{imp_{u}}{\sum_{j=K+1}^{\lceil\log N\rceil} |N_{H'^{(j)}}(u)| 2^{-j}} \cdot \sum_{j=K+1}^{\lceil\log N\rceil} 2^{-j} \left( \sum_{i=1}^{|N_{H'^{(j)}}(u)|/b}\left( |S \cap B^{(j)}_i(u)| -  b/2\right)^2 \right) \right) 
\end{align*}

Additionally, let $ B'_1\sqcup B'_2\ldots \sqcup  B'_{\lfloor|V|/{b}\rfloor} \subseteq V$ be a partition of (all except up to $b-1$) $V$ into buckets of size $b$ and

\begin{align*}
 \Phi_3(S) := \frac{4}{b\lfloor|V|/{b}\rfloor} \sum_{i=1}^{\lfloor|V|/{b}\rfloor} \left( |S \cap B^{'}_i| -  b/2\right)^2.
\end{align*}

Also, we invoke \cref{lem:edge_buckets} with input $G'$ to compute in $O((n+m')\log \log N)$ work and $\poly(\log N)$ depth
buckets $B''_1 \sqcup B''_2 \sqcup \ldots \sqcup B''_k \subseteq E'$ where each bucket contains $b$ edges and at most $O(b^5)$ edges in $E'$ are not contained in any bucket. Moreover, for each bucket $B \subseteq E'$, and each edge $e = \{v,v'\} \in B$, one endpoint $v^*(e) \in \{v,v'\}$ of this edge is marked as special such that the special endpoints $v^*(e)$ of different edges $e \in B$ in this bucket are distinct. We define

\begin{align*}
    \Phi_4(S) =   \frac{4}{b \cdot k}\sum_{i=1}^{k}\left(||S \cap \{v^*(e) \colon e \in B''_i\}| - b/2 \right)^2.
\end{align*}

Finally, we define

\begin{align*}
    \Phi_5(S) =  \frac{100}{\gamma} \cdot \frac{\sum_{e'\in \binom{S}{2}\cap E'} w(e') \cdot 4  + \sum_{v\in S} w(v) \cdot 2 }{\left( \sum_{e'=\{v, v'\} \in E'} w(e')  + \sum_{v\in V} w(v)   \right)}.
\end{align*}

Notice that for $S$ chosen randomly by including each $V$ node in it with probability $1/2$, we have $\E[\Phi_1(S)]=\E[\Phi_2(S)]=\E[\Phi_3(S)]=\E[\Phi_4(S)]=1$ and $\E[\Phi_5(S)] = \frac{100}{\gamma}$. Our derandomization is much more sensitive to deviations in $\Phi_{5}(S)$ from its expectation, and this is the reason that we have normalized the potential functions such that $\Phi_{5}(S)$ is a $\Theta(1/\gamma) = \Theta(b^{1/6})$ factor larger than the others. 

Using \cref{lem:local_rounding}, by decomposing the potential into vertex utilities and edge costs, and by setting $\eps = \Theta(1/(b-1))$ with a small leading constant, we can compute with work $b \cdot (m + m' + n)\poly(\log \log N)/\eps = (m+m'+n)\poly(\log \log N)/\gamma^{30}$ and depth $\poly(\log N)$ a subset $S \subseteq V$ such that $\Phi_1(S) + \Phi_2(S) + \Phi_3(S)  + \Phi_4(S) + \Phi_5(S) \leq 5 + \frac{100}{\gamma}$. From this, we can make several conclusions.

We first use the above and in particular $\Phi_1\leq 5 + \frac{100}{\gamma}$ to conclude that $\sum_{u\in U} |N_{H'}(u) \cap S| \leq (0.51) \sum_{u\in U} |N_{H}(u)|$. Let us call a bucket $B^{(j)}_i(u)$ bad if $||B^{(j)}_i(u)\cap S|-b/2|\geq b^{0.8}$. Notice that for each bad bucket, we have $\left( |S \cap B^{(j)}_i(u)| -  b/2\right)^2 \geq b^{1.6}$. On the other hand, we know that $\Phi_1(S)\leq 5 + \frac{100}{\gamma}$ where \begin{align*}
 \Phi_1(S) &:= &&\frac{4}{\sum_{u\in U} \sum_{j=K+1}^{\lceil\log N\rceil} |N_{H'^{(j)}}(u)|} \cdot \sum_{u \in U} \sum_{j=K+1}^{\lceil\log N\rceil} \left( \sum_{i=1}^{|N_{H'^{(j)}}(u)|/b}\left( |S \cap B^{(j)}_i(u)| -  b/2\right)^2 \right). 
 \end{align*}
Hence, overall, at most $\left(\sum_{u\in U} \sum_{j=K+1}^{\lceil\log N\rceil} |N_{H'^{(j)}}(u)|/b\right)/\Theta(\gamma b^{0.6}) = \sum_{u\in U} |N_{H'}(u)|/\Theta(\gamma \cdot b^{1.6})$ buckets are bad. Therefore, we can conclude that

\begin{align*}
    \sum_{u\in U} |N_{H'}(u) \cap S| \leq  ((1/2+1/b^{0.2}+\Theta(1/(\gamma \cdot b^{0.6})))\sum_{u\in U} |N_{H}(u)| \leq (0.51) \sum_{u\in U} |N_{H}(u)|.
\end{align*}
 
Next, we use $\Phi_2(S) \leq 5 + \frac{100}{\gamma}$ to conclude that for nearly all nodes in $U$, weighted by importance, the probability summation in their neighborhood goes down by a $1/2$ factor. Let $\mathcal{B}^{(j)}(u)$ be the collection of all $i$ such that bucket $B^{(j)}_{i}(u)$ is bad, with the definition of bad bucket being the same as before, i.e., $||B^{(j)}_i(u)\cap S|-b/2|\geq b^{0.8}$. Let us call a node $u$ bad if $\sum_{j=K+1}^{\lceil\log N\rceil} \sum_{i\in \mathcal{B}^{(j)}(u)} b 2^{-j}$ is more than $\left(\sum_{j=K+1}^{\lceil\log N\rceil} \sum_{i=1}^{|N_{H'^{(j)}}(u)|/b} b 2^{-j}\right)/b^{0.2}$. Let $U_{good}$ be the set of all $u\in U$ that are not bad.

For any node $u\in U_{good}$, we have 
\begin{align*} 
& && \sum_{v \in N_{H'}(u)\cap S}2^{-k_v} = \left(\sum_{j=K+1}^{\lceil\log N\rceil} \sum_{i=1}^{|N_{H'^{(j)}}(u)|/b} |B^{(j)}_i(u)\cap S|  2^{-j}\right) \\
&\geq&& (1-1/b^{0.2})(1/2-1/b^{0.2}) \left(\sum_{j=K+1}^{\lceil\log N\rceil} \sum_{i=1}^{|N_{H'^{(j)}}(u)|/b} b 2^{-j}\right) \\
&\geq&& (1/2-3/b^{0.2}) \left(\sum_{j=K+1}^{\lceil\log N\rceil} \sum_{i=1}^{|N_{H'^{(j)}}(u)|/b} b 2^{-j}\right)\\
& = && (1/2-3/b^{0.2}) \sum_{v \in N_{H'}(u)}2^{-k_v} \\
& \geq  &&  (1/2)(1-\gamma) \sum_{v \in N_{H'}(u)}2^{-k_v} \\
& \geq &&(1/2)(1-\gamma) \sum_{v \in N_{H}(u)}2^{-k_v} - \gamma/2
\end{align*}
and also that
\begin{align*} 
& && \sum_{v \in N_{H'}(u)\cap S}2^{-k_v} = \left(\sum_{j=K+1}^{\lceil\log N\rceil} \sum_{i=1}^{|N_{H'^{(j)}}(u)|/b} |B^{(j)}_i(u)\cap S|  2^{-j}\right) \\
&\leq&& ((1-1/b^{0.2})(1/2+1/b^{0.2})+(1/b^{0.2})) \left(\sum_{j=K+1}^{\lceil\log N\rceil} \sum_{i=1}^{|N_{H'^{(j)}}(u)|/b} b 2^{-j}\right) \\
&\leq&& (1/2+3/b^{0.2}) \left(\sum_{j=K+1}^{\lceil\log N\rceil} \sum_{i=1}^{|N_{H'^{(j)}}(u)|/b} b 2^{-j}\right)\\
& \leq && (1/2+3/b^{0.2}) \sum_{v \in N_{H}(u)}2^{-k_v} \\
& \leq && (1/2)(1+\gamma) \sum_{v \in N_{H}(u)}2^{-k_v}.
\end{align*}

Now, we argue that $\sum_{u\in U_{good}} imp_{u} \geq (1-\gamma)\sum_{u\in U} imp_{u}$. For each bad bucket $B^{(j)}_i(u)$, we have $\left( |S \cap B^{(j)}_i(u)| -  b/2\right)^2 \geq b^{1.6}$. And a node $u$ is called bad if we have $$\sum_{j=K+1}^{\lceil\log N\rceil} \sum_{i\in \mathcal{B}^{(j)}(u)} b 2^{-j} \geq \left(\sum_{j=K+1}^{\lceil\log N\rceil} \sum_{i=1}^{|N_{H'^{(j)}}(u)|/b} b 2^{-j}\right)/b^{0.2}.$$ Hence, for any bad node $u$, we have 
\begin{align*}
   \left(\frac{imp_{u}}{\sum_{j=K+1}^{\lceil\log N\rceil} |N_{H'^{(j)}}(u)| 2^{-j}} \cdot \sum_{j=K+1}^{\lceil\log N\rceil} 2^{-j} \left( \sum_{i=1}^{|N_{H'^{(j)}}(u)|/b}\left( |S \cap B^{(j)}_i(u)| -  b/2\right)^2 \right) \right) \geq imp_{u} \cdot \Theta(b^{0.4})
\end{align*}
Given that $\Phi_2(S)\leq 5 + \frac{100}{\gamma}$, we conclude $\sum_{u\in U_{good}} imp_{u} \geq (1-1/\Theta(\gamma \cdot b^{0.4}))\sum_{u\in U} imp_{u} \geq (1-\gamma)\sum_{u\in U} imp_{u}$.

Next, we use $\Phi_3(S)\leq 5 + \frac{100}{\gamma}$ to conclude that $|S|\leq (2/3) |V| + b \leq (2/3)|V| + \poly(\log N)$. Consider the buckets $B'_1\sqcup B'_2\ldots \sqcup  B'_{\lfloor|V|/{b}\rfloor} \subseteq V$ that we had. Let us call bucket $B'_{i}$ bad if $||B'_{i}\cap S|-b/2|\geq 0.1b$. Notice that for each bad bucket, we have $(|B'_{i}\cap S|-b/2)^2\geq 0.01b^2$. Since $\Phi_3(S)= \frac{4}{b\lfloor|V|/{b}\rfloor} \sum_{i=1}^{\lfloor|V|/{b}\rfloor} \left( |S \cap B^{'}_i| -  b/2\right)^2\leq 5 + \frac{100}{\gamma}$, we know that at most $\Theta(1/(\gamma \cdot b))<0.01$ fraction of these buckets can be bad. Hence, we conclude that $|S|\leq (1 + 0.1 + 0.01)|V|/2 + \poly(\log N)$.

Next, we use $\Phi_4(S) \leq 5 + \frac{100}{\gamma}$ to conclude that $|E'\cap \binom{S}{2}| \leq (2/3) |E'| + \poly(\log N).$

Consider the buckets $B''_1\sqcup B''_2\ldots \sqcup  B''_{k} \subseteq E'$. Let us call bucket $B''_{i}$ bad if $||S \cap \{v^*(e) \colon e \in B''_i\}| - b/2| \geq 0.1b$. Notice that for each bad bucket, we have $\left(|S \cap \{v^*(e) \colon e \in B''_i\}| - b/2\right)^2 \geq 0.01b^2$. Since $\Phi_4(S) =   \frac{4}{b \cdot k}\sum_{i=1}^{k}\left(||S \cap \{v^*(e) \colon e \in B''_i\}| - b/2 \right)^2 \leq 5 + \frac{100}{\gamma}$, we know that at most $\Theta(1/(\gamma \cdot b))<0.01$ fraction of these buckets can be bad. Hence, we conclude that $|E' \cap \binom{S}{2}|\leq (1 + 0.1 + 0.01)|E'|/2 + \poly(\log N)$.

Finally, we use $\Phi_5(S) \leq 5 + \frac{100}{\gamma}$ to conclude that
 \begin{align*}\sum_{e'\in \binom{S}{2}\cap E'} w(e') \cdot 4  + \sum_{v\in S} w(v) 2 \leq (1+\gamma) \cdot \left( \sum_{e'=\{v, v'\} \in E'} w(e')  + \sum_{v\in V} w(v) \right).\end{align*}

We have 
 \begin{align*}
    \Phi_5(S) =  \frac{100}{\gamma} \cdot \frac{\sum_{e'\in \binom{S}{2}\cap E'} w(e') \cdot 4  + \sum_{v\in S} w(v) 2}{\left( \sum_{e'=\{v, v'\} \in E'} w(e')  + \sum_{v\in V} w(v)   \right)} \leq 5 + \frac{100}{\gamma}
\end{align*}
and therefore

\begin{align*}
&\sum_{e'\in \binom{S}{2}\cap E'} w(e') \cdot 4 \cdot  + \sum_{v\in S} w(v) 2\\
&\leq
\left(\frac{\gamma}{100}\right) \cdot \left(5 + \frac{100}{\gamma} \right) \left( \sum_{e'=\{v, v'\} \in E'} w(e')  + \sum_{v\in V} w(v)  \right) \\
&\leq (1+\gamma) \cdot \left( \sum_{e'=\{v, v'\} \in E'} w(e')  + \sum_{v\in V} w(v) \right).
\end{align*}

\end{proof}

\begin{lemma}
    \label{lem:edge_buckets}
    Consider any $n$-node $m$-edge graph $G'=(V, E')$ where nodes have identifiers in $\{1,2, \dots, n\}$, let $N\geq n$ be a given upper bound, and consider a given parameter $b\leq \poly(\log N)$. There is a deterministic parallel algorithm with work $O((m+n)\log \log N)$ and depth $\poly(\log N)$ that bundles all except  $O(b^5)$ edges of $E'$ into buckets of size exactly $b$ with the following property: in each bucket $B\subseteq E'$, for each edge $e=\{v, v'\}\in B$, one endpoint $v^*(e)\in \{v, v'\}$ of this edge is marked as special such that the special endpoints $v^*(e)$ of different edges $e\in B$ in this bucket are distinct.
\end{lemma}
\begin{proof}
For every node $v$ of degree $d(v)\geq b$, bucket its edges into buckets of size $b$, with the exception of at most $b-1$ remaining edges. In each bucket, mark the other endpoint as the special end. Remove all these bucketized edges, for all such $v$. We remain with up to $n$ nodes, each of degree at most $b-1< \poly(\log n)$. Sort these nodes by degree, using \Cref{lem:sorting}, with $O(n\log\log N)$ work. In particular, all nodes of any degree $d\in\{1,2, \dots b-1\}$ are together in one pile. Let $n_d$ be the number of nodes of degree $d$. We partition $\{1,2, \dots b\lceil n_d/b\rceil\}$ into parts of size $b$ in the straightforward way $\{1, 2, ... b\} \sqcup \{b+1, b+2, \dots, 2b\} \sqcup \ldots \{b(\lceil n_d/b\rceil-1)+1, b(\lceil n_d/b\rceil-1)+1, \dots, b\lceil n_d/b\rceil\}$. For each part, for each $i\in\{1,2, \dots, d\}$, create one new edge bucket by putting in it the $i^{th}$ edge of each node of this part. These nodes are declared the special endpoints of these edges (notice that they are distinct inside each bucket). What remains is at most $n_d - b\lceil n_d/b\rceil <b$ nodes of degree $d$, and thus at most $db$ edges. Over all degrees $d\in\{1,2,\dots, b-1\}$, this is at most $O(b^3)$ edges that are left outside buckets. The depth bound is $\poly(\log N)$, and the  work bound is $O((m+n)\log \log N)$ because of the usage of the sorting subroutine (after which everything else is linear work).
\end{proof}

\subsection{High Probability Regime}
\begin{lemma}[High Probability Regime]
\label{lem:mis_high_probability}
Let $H$ be an $n$-vertex $m$-edge bipartite graph with bipartition $V(H) = U \sqcup V$, where each node has a unique identifier in $\{1, 2, \dots 2n\}$. Let $N\geq n$ be a given upper bound. For every $u \in U$, let $imp_u \in \mathbb{R}_{\geq 0}$ and for every $v \in V$, let $k_v \in \{1,2,\ldots, K\}$, where $K=\lceil 100\log\log N\rceil$, such that for every $u \in U$, we have $\sum_{v \in N_H(u)}2^{-k_v} \leq 40$. Consider also an additional graph $G'=(V, E')$ on the same set $V$ of vertices, with $m'$ edges, where each edge $e'\in E'$ has a weight $w(e')\in \mathbb{R}_{\geq 0}$.

There exists a constant $C\geq 1$ and a deterministic parallel algorithm with work $(m+m'+n)\poly(\log \log N)$ and depth $\poly(\log N)$ that computes two subsets $S \subseteq V$ and $U_{good} \subseteq U$ such that: 

\begin{itemize}
\item $\sum_{u \in U_{good}} imp_u \geq (0.99)\sum_{u \in U} imp_u$,
\item for every $u\in U_{good}$, we have $|N_H(u) \cap S|  \in [ (\sum_{v \in N_H(u)}2^{-k_v}) - 0.1, (\sum_{v \in N_H(u)}2^{-k_v}) + C]$, and

\item $\sum_{e'\in \binom{S}{2}\cap E'} w(e') \leq C \cdot \sum_{e'=\{v, v'\} \in E'} w(e') \cdot 2^{-(k_v+k_{v'})}$.
\end{itemize}
\end{lemma}

\begin{proof}
During the process, in each iteration $i$, we will have a set $V^{(i)}_{alive}$ of $V$-vertices that remain alive, with the interpretation that $V^{(0)}_{alive} = V$, and per $i$, we have $V^{(i+1)}_{alive} \subseteq V^{(i)}_{alive}$. For each $v\in V^{(i)}_{alive}$, we define $k^{(i)}_{v} = \min\{k_{v}, K-i\}$. We would like to determine $V^{(i)}_{alive}$ such that, for each $i$, and for nearly all $u\in U$, we have that $\sum_{v\in N_{H}(u)\cap V^{(i+1)}_{alive}} 2^{-k^{(i+1)}_{v}}$ remains close to $\sum_{v\in N_{H}(u)\cap V^{(i)}_{alive}} 2^{-k^{(i)}_{v}}$ up to a small additive error such that these additive errors add up to a small constant over all iterations.

\medskip

As initialization, we define $U^{(0)} = U$, $V^{(0)} = \{v \in V \colon k_v = K\}$, $H^{(0)} = H[U^{(0)} \sqcup V^{(0)}]$ and $G'^{(0)} = G'[V^{(0)}]$.

For $i = 0,1,\ldots, I$ for $I=K-20$, in round $i$ invoke 
\cref{lem:mis_high_probability_half} with input:
\begin{itemize}
    \item $H^{(i)}$ with bipartition $U^{(i)} \sqcup V^{(i)}$
    \item $\gamma^{(i)} = 1/(100(K-i)^{2})$
    \item $imp_u$ for every $u \in U^{(i)}$
    \item $G'^{(i)}$ 
    \item For every $\{v,v'\} \in E(G'^{(i)})$, $w^{(i)}(\{v,v\}) = w(\{v,v'\}) \cdot  2^{-k_v^{(i)}}\cdot 2^{-k_{v'}^{(i)}}$
    \item For every $v \in V^{(i)}$, $w^{(i)}(v) = \sum_{v' \in N_{G'}(v) \cap \left(V^{(i)}_{alive} \setminus V^{(i)} \right)} w(\{v,v'\}) \cdot 2^{-(k^{(i)}_v + k^{(i)}_{v'})}$
\end{itemize}
As a result, we obtain two sets $S^{(i)} \subseteq V^{(i)}$ and $U_{good}^{(i)} \subseteq U^{(i)}$ satisfying

\begin{itemize}
\item[(I)] $\sum_{u \in U^{(i)}_{good}} imp_u \geq (1-\gamma^{(i)})\sum_{u \in U^{(i)}} imp_u$, and
\item[(II)] for every $u \in U^{(i)}_{good}$, $|N_{H^{(i)}}(u) \cap S^{(i)}| \in (1 \pm \gamma^{(i)})\frac{|N_{H^{(i)}}(u)|}{2} \pm (1/\gamma^{(i)})^7$.
\item[(III)] $\sum_{e'\in \binom{S^{(i)}}{2}\cap E'} 4 w^{(i)}(e') + \sum_{v\in S^{(i)}} 2 w^{(i)}(v) \leq (1+\gamma^{(i)}) \left(
\sum_{e' \in \binom{V^{(i)}}{2} \cap E'} w^{(i)}(e') + \sum_{v\in 
V^{(i)}} w^{(i)}(v)\right)$.
\end{itemize}
Then, we set
$V^{(i+1)}_{alive} = V^{(i)}_{alive} \setminus \left(V^{(i)} \setminus S^{(i)} \right)$, and we define $H^{(i+1)} = H[U^{(i)}_{good} \sqcup \{v \in V^{(i+1)}_{alive} \colon k^{(i+1)}_v = K-(i+1)\}]$, $G'^{(i+1)} = G'[\{v \in V_{alive}^{(i+1)} \colon k_v^{(i+1)} = K - (i+1)\}]$, and go to iteration $i+1$ of the for loop.

At the end of the for loop, we set the final output as $S = V^{(I+1)}_{alive}$, and $U_{good} = U_{good}^{(I)}$.

We start with proving the third bullet point, namely that there exists an absolute constant $C$ satisfying

\[\sum_{e'\in \binom{S}{2}\cap E'} w(e') \leq C \cdot \sum_{e'=\{v, v'\} \in E'} w(e') \cdot 2^{-(k_v+k_{v'})}.\]

To that end, we first show that

\[\sum_{\{v,v'\} \in \binom{V^{(i+1)}_{alive}}{2}\cap E'}w(\{v,v'\}) 2^{-k^{(i+1)}_v} \cdot 2^{-k^{(i+1)}_{v'}}  \leq  (1+\gamma^{(i)})\sum_{\{v,v'\} \in \binom{V^{(i)}_{alive}}{2}\cap E'} w(\{v,v'\}) 2^{-k^{(i)}_v} \cdot 2^{-k^{(i)}_{v'}}\]

Indeed, we have

\begin{align*}
& &&\sum_{\{v,v'\} \in \binom{V^{(i+1)}_{alive}}{2}\cap E'}w(\{v,v'\}) 2^{-k^{(i+1)}_v} \cdot 2^{-k^{(i+1)}_{v'}} \\ 
&= &&\sum_{\{v,v'\} \in \binom{V^{(i+1)}_{alive} \setminus S^{(i)}}{2}\cap E'}w(\{v,v'\}) 2^{-k^{(i)}_v} \cdot 2^{-k^{(i)}_{v'}} \\
& &&+ 2 \cdot \sum_{\{v,v'\} \in \binom{V^{(i+1)}_{alive}}{2}\cap E' \colon |\{v,v'\} \cap S^{(i)}|= 1}w(\{v,v'\}) 2^{-k^{(i)}_v} \cdot 2^{-k^{(i)}_{v'}} \\
& &&+ 4 \cdot \sum_{\{v,v'\} \in \binom{S^{(i)}}{2}\cap E'}w(\{v,v'\}) 2^{-k^{(i)}_v} \cdot 2^{-k^{(i)}_{v'}} \\
&=&& \sum_{\{v,v'\} \in \binom{V^{(i+1)}_{alive} \setminus S^{(i)}}{2}\cap E'}w(\{v,v'\}) 2^{-k^{(i)}_v} \cdot 2^{-k^{(i)}_{v'}} 
+ \sum_{v \in S^{(i)}} 2w^{(i)}(v) 
+ \sum_{e' \in \binom{S^{(i)}}{2} \cap E'}4 w^{(i)}(e') \\
&\leq&& \sum_{\{v,v'\} \in \binom{V^{(i+1)}_{alive} \setminus S^{(i)}}{2}\cap E'}w(\{v,v'\}) 2^{-k^{(i)}_v} \cdot 2^{-k^{(i)}_{v'}} + (1+\gamma^{(i)}) \left(
\sum_{e' \in \binom{V^{(i)}}{2} \cap E'} w^{(i)}(e') + \sum_{v\in 
V^{(i)}} w^{(i)}(v)\right) \\
&\leq&& (1+\gamma^{(i)})\sum_{\{v,v'\} \in \binom{V^{(i)}_{alive}}{2}\cap E'} w(\{v,v'\}) 2^{-k^{(i)}_v} \cdot 2^{-k^{(i)}_{v'}}.
\end{align*}

Therefore,

\begin{align*}
\sum_{e'\in \binom{S}{2}\cap E'} w(e')  &= \sum_{e'\in \binom{V_{alive}^{(I+1)}}{2}\cap E'} w(e') \\
&\leq 10^{100}\sum_{\{v,v'\}\in \binom{V_{alive}^{(I+1)}}{2}\cap E'} w(e') \cdot 2^{-k_v^{(I+1)}} \cdot 2^{-k_{v'}^{(I+1)}} \\
&\leq 10^{100} \left(\prod_{i=0}^I (1 + \gamma^{(i)})\right) \cdot
\sum_{\{v,v'\}\in \binom{V_{alive}^{(0)}}{2}\cap E'} w(e') \cdot 2^{-k_v^{(0)}} \cdot 2^{-k_{v'}^{(0)}} \\
&\leq C \cdot \sum_{e'=\{v, v'\} \in E'} w(e') \cdot 2^{-(k_v+k_{v'})}
\end{align*}
for $C$ being a sufficiently large constant.

We claim that for each $i$, for each node $u \in U^{(i)}_{good}$, we have 
\begin{align}
\label{eqnDrift2}
    \sum_{v\in N_{H}(u)\cap V^{(i+1)}_{alive}} 2^{-k^{(i+1)}_{v}} \in \sum_{v\in N_{H}(u)\cap V^{(i)}_{alive}} 2^{-k^{(i)}_{v}} \pm 50 \gamma^{(i)}
\end{align} To prove this claim, we will show by a simultaneous induction that we also have the following at all times:
\begin{align}
\label{eqnUB2}
\sum_{v\in N_{H}(u)\cap V^{(i)}_{alive}} 2^{-k^{(i)}_{v}} \leq 45.
\end{align}

Notice that \Cref{eqnUB2} is a guarantee about $V^{(i)}_{alive}$ and \Cref{eqnDrift2} is a guarantee about $V^{(i+1)}_{alive}$.

For $i=0$, the guarantee of \Cref{eqnUB2} for $V^{(0)}_{alive}$ is trivially implied by the initial assumption of  $\sum_{v \in N_H(u)}2^{-k_v} \leq 40$. Suppose that the two hold up to iteration $i-1$. We now conclude them for iteration $i$.

From (II), we have that $|N_{H^{(i)}}(u) \cap S^{(i)}| \in (1 \pm \gamma^{(i)})\frac{|N_{H^{(i)}}(u)|}{2} \pm (1/\gamma^{(i)})^7$. Hence, 
\begin{align*}
    \sum_{v\in N_{H}(u)\cap V^{(i+1)}_{alive}} 2^{-k^{(i+1)}_{v}} \in (1\pm \gamma^{(i)})\sum_{v\in N_{H}(u)\cap V^{(i)}_{alive}} 2^{-k^{(i)}_{v}} \pm \frac{(1/\gamma^{(i)})^7}{2^{(K-(i+1))}}.
\end{align*} 
Now notice that $\gamma^{(i)} = 1/(100(K-i)^{2})$ and hence we know $\frac{(1/\gamma^{(i)})^7}{2^{(K-(i+1))}} \leq 5\gamma^{(i)}$. Furthermore, from the inductive assumption, we knew that $\sum_{v\in N_{H}(u)\cap V^{(i)}_{alive}} 2^{-k^{(i)}_{v}} \leq 45$. Using these two, and the above inequality, we conclude that 
\begin{align*}
    \sum_{v\in N_{H}(u)\cap V^{(i+1)}_{alive}} 2^{-k^{(i+1)}_{v}} \in \sum_{v\in N_{H}(u)\cap V^{(i)}_{alive}} 2^{-k^{(i)}_{v}} \pm 50 \gamma^{(i)}
\end{align*}
which gives \Cref{eqnDrift2} for iteration $i$. In addition, from \Cref{eqnDrift2} for all iterations $0$ to $i$, we know that

\begin{align*}
   \sum_{v\in N_{H}(u)\cap V^{(i+1)}_{alive}} 2^{-k^{(i+1)}_{v}} \leq \sum_{v\in N_{H}(u)\cap V} 2^{-k_{v}} \pm (50 \sum_{j=0}^{i} \gamma^{(j)}) \in \sum_{v\in N_{H}(u)\cap V} 2^{-k_{v}} + 5 \leq 45
\end{align*}
which implies \Cref{eqnUB2} for $i+1$.

Hence, we get

\begin{align*}
  |N_H(u) \cap S| = |N_H(u) \cap V^{(I+1)}_{alive}| \leq 2^{20}  \left(\sum_{v \in N_H(u) \cap V^{(I+1)}_{alive}} 2^{-k_v^{(I+1)}}  \right) \leq 2^{20} \cdot 45\leq 10^{25}.
  \end{align*}

  and

  \begin{align*}
  |N_H(u) \cap S| &= |N_H(u) \cap V^{(I+1)}_{alive}| \\
  &\geq \sum_{v \in N_H(u) \cap V^{(I+1)}_{alive}} 2^{-k_v^{(I+1)}} \\
  &\geq \sum_{v \in N_H(u)} 2^{-k_v} - \left( 50 \sum_{j=0}^I \gamma^{(j)}\right)\\
  &\geq  \left(\sum_{v \in N_H(u)} 2^{-k_v} \right) - 0.01.
  \end{align*}

Next, we show that $\sum_{u \in U_{good}} \geq 0.99\sum_{u \in U} imp_u$.

Indeed, we have

\begin{align*}
 \sum_{u \in U_{good}} imp_u = \sum_{u \in U^{(I)}_{good}} imp_u \geq \left( 1 - \sum_{j=0}^I \gamma^{(j)}\right) \sum_{u \in U} imp_u \geq 0.99 \sum_{u \in U} imp_u.
\end{align*}

The work can be upper bounded by
\[(m+n + m')\poly(\log \log N)/\left(\sum_{j=0}^I\left(\gamma^{(j)}\right)^{30}\right) = (m+n+m')\poly(\log \log N)\]
and the depth can be upper bounded by 
\[I \cdot \poly(\log N) = \poly(\log N).\]

\end{proof}

\begin{lemma}[High Probability 1/2 Sampling]
\label{lem:mis_high_probability_half}
Let $H$ be an $n$-vertex $m$-edge bipartite graph with bipartition $V(H) = U \sqcup V$, where each node has a unique identifier in $\{1, 2, \dots 2n\}$, let $N\geq n$ be a given upper bound, and let $\gamma \in [1/\log(N), 0.01)$. For every $u \in U$, let $imp_u \in \mathbb{R}_{\geq 0}$. Consider also an additional graph $G'=(V, E')$ on the same set $V$ of vertices, with $m'$ edges, where each edge $e'\in E'$ has a weight $w(e')\in \mathbb{R}_{\geq 0}$ and each vertex $v\in V$ has a weight $w(v)\in \mathbb{R}_{\geq 0}$. 


There exists a deterministic parallel algorithm with work $(m+m'+n)\poly(\log \log N)/\gamma^{30}$ and depth $\poly(\log N)$ that computes two subsets $S \subseteq V$ and $U_{good} \subseteq U$ such that:

\begin{itemize}
\item $\sum_{u \in U_{good}} imp_u \geq (1-\gamma)\sum_{u \in U} imp_u$ 
\item for every $u \in U_{good}$, $|N_H(u) \cap S| \in (1 \pm \gamma)\frac{|N_H(u)|}{2} \pm (1/\gamma)^{7}$,
\item $\sum_{e'\in \binom{S}{2}\cap E'} 4 w(e') + \sum_{v\in S} 2 w(v) \leq (1+\gamma) \left(
\sum_{e'=\{v, v'\} \in E'} w(e') + \sum_{v\in 
V} w(v)\right)$.
\end{itemize}
\end{lemma}

\begin{proof}
Let $b := \lceil (1/\gamma)^{6}\rceil$. For each node $u \in U$, we slightly adjust its neighborhood $N_H(u)$ by dropping up to $b-1$ of the neighbors, so that the remaining $|N_H(u)|$ is a multiple of $b$---notice that this causes at most an additive $\pm b \in \pm (1/\gamma)^{7}$ difference in the output size of $|N_H(u) \cap S|$, which is permitted in the second condition. For the remaining part of $N_{H}(u)$, we partition it arbitrarily into buckets $N_H(u) = B_1(u) \sqcup B_2(u) \sqcup \ldots \sqcup B_{|N_H(u)|/b}(u)$ of size $b$.
We want a derandomization such that, for nearly all buckets $B_i(u)$ over different $u$ (normalized and weighted by importance $imp(u)$, to be made precise), it holds that $|B_{i}(u) \cap S|$ is very close to $b/2$. In particular, we will use $(|B_{i}(u) \cap S| - b/2)^2$ as a measure of the drift from the $b/2$ target in this bucket, and perform the derandomization such that an aggregate of these drifts over all buckets and all $u$ is minimized. Concretely, our derandomization will try to maintain the potential 
\begin{align*}
 \Phi_1(S) &:= && \frac{4}{\sum_{u \in U} imp_u}\sum_{u \in U} \frac{imp_u}{|N_H(u)|} \cdot \left( \sum_{i=1}^{|N_H(u)|/b}\left( |S \cap B_i(u)| -  b/2\right)^2 \right) \\
 &=  &&\frac{4}{\sum_{u \in U} imp_u} \sum_{u \in U} \frac{imp_u}{|N_H(u)|} \cdot \sum_{i=1}^{|N_H(u)|/b}\left( 2\binom{|S \cap B_i(u)|}{2} -  (b-1) |S \cap B_i(u)| + b^2/4 \right) 
\end{align*}
as if set $S$ is chosen by including each node $v\in V$ in it with probability $1/2$. Notice that if $S$ is chosen randomly by including each $v\in V$ in $S$ with probability $1/2$, we would get 
\begin{align*}
\E[\Phi_1(S)] &= \frac{4}{\sum_{u \in U} imp_u} \sum_{u \in U} \frac{imp_u}{|N_H(u)|} \cdot \sum_{i=1}^{|N_H(u)|/b}\left( 2 \E[\binom{|S \cap B_i(u)|}{2}] -  (b-1) \E[|S \cap B_i(u)|] + b^2/4 \right)  \\
&=  \frac{4}{\sum_{u \in U} imp_u}\sum_{u \in U} \frac{imp_u}{|N_H(u)|} \sum_{i=1}^{|N_H(u)|/b}\left((b)(b - 1)/4  - (b-1) b/2 + b^2/4\right) \\
&= \frac{4}{\sum_{u \in U} imp_u} \sum_{u \in U} \frac{imp_u}{|N_H(u)|} \sum_{i=1}^{|N_H(u)|/b}\frac{b}{4} \\
&= \frac{4}{\sum_{u \in U} imp_u}\sum_{u \in U} imp_u \cdot \frac{1}{4} \\
&= 1.
\end{align*}

\begin{align*}
 \Phi_2(S) = \frac{10}{\gamma} \left(\frac{1}{\sum_{e'=\{v, v'\} \in E'} w(e') + \sum_{v\in 
V} w(v)} \right)
\left(\sum_{e'\in \binom{S}{2}\cap E'} 4 w(e') + \sum_{v\in S} 2 w(v) \right),
\end{align*}

Notice that if $S$ is chosen randomly by including each $v \in V$ in $S$ with probability $1/2$, we would get $\E[\Phi_2(S)] = \frac{10}{\gamma}$.

We next translate this potential to the utility/cost language of \Cref{lem:local_rounding} using an appropriate auxiliary mutli-graph $\bar{H}=(V, \bar{E})$. 
In particular, for each node $v \in V$, let us define
\[util_1(v) := \frac{4}{\sum_{u \in U} imp_u}(b-1)\sum_{u \in N_H(v)} \frac{imp_u}{|N_H(u)|}.\]

Notice that we have
\begin{align*}
\sum_{v \in V} util_1(v)
=\frac{4}{\sum_{u \in U} imp_u}\sum_{v \in V} (b-1)\sum_{u \in N_H(v)} \frac{imp_u}{|N_H(u)|}
= 4(b-1).
\end{align*}
Furthermore, for every bucket $B_i(u)$ and two distinct vertices $v, v' \in B_i(u)$, let us add a new edge between $v$ and $v'$ in the edge multi-set $\bar{E}$ of the auxiliary graph $\bar{H}$ (notice that there might be several parallel such edges, one for each $u$). We set the cost of this newly added edge as
\[cost_1(\{v,v'\}) := \frac{8}{\sum_{u \in U} imp_u} \frac{imp_u}{|N_H(u)|} \]
Notice that we have

\begin{align*}
\sum_{e\in \bar{E}} cost_1(e) &= \sum_{u \in U} \sum_{i = 1}^{|N_H(u)|/b}\sum_{v,v' \in \binom{B_i(u)}{2}}  \frac{8}{\sum_{u \in U} imp_u}\cdot \frac{imp_u}{|N_H(u)|} \\
&= \sum_{u\in U} \frac{|N_H(u)|}{b} \cdot \binom{b}{2} \cdot \frac{8}{\sum_{u \in U} imp_u} \cdot \frac{imp_u}{|N_H(u)|} \\
&= 4(b-1).
\end{align*}

We can write $\Phi_1(S) = \big(\frac{4}{\sum_{u \in U} imp_u}\sum_{u\in U}\frac{imp_u}{|N_H(u)|} (\sum_{i=1}^{|N_{H}(u)|/b} b^2/4)\big) - \big(\sum_{v \in S} util_1(v)\big) + \big(\sum_{e \in \bar{E}\cap \binom{S}{2}} cost_1(e)\big).$

For each node $v \in V$, let us define

\[util_2(v) := - \frac{10}{\gamma} \left(\frac{1}{\sum_{e'=\{v, v'\} \in E'} w(e') + \sum_{v\in 
V} w(v)} \right) 2 w(v)\]

and for every edge $e' \in E'$, let us define

\[cost_2(e') := \frac{10}{\gamma} \left(\frac{1}{\sum_{e'=\{v, v'\} \in E'} w(e') + \sum_{v\in 
V} w(v)} \right) 4 w(e')\]

Notice that we have

\[\sum_{e' \in E'}cost_2(e') \leq \frac{10}{\gamma}.\]

We can write $\Phi_2(S) = - \sum_{v \in S} util_2(v) + \sum_{e' \in E' \cap \binom{S}{2}} cost_2(e')$.

Using \cref{lem:local_rounding} with $\eps = 1/(4(b-1))$, we compute a subset $S \subseteq V$ satisfying 

\begin{align*}
    \Phi_1(S) + \Phi_2(S) &\leq&& \big(\frac{4}{\sum_{u\in U} imp_u}\sum_{u\in U}\frac{imp_u}{|N_H(u)|} (\sum_{i=1}^{|N_{H}(u)|/b} b^2/4)\big) \\
    & &&- \big(\frac{1}{2}\sum_{v \in V} util_1(v)\big) + \frac{1}{4}\big(\sum_{e \in \bar{E}} cost_1(e)\big)  + \eps \sum_{e \in \bar{E}} cost_1(e) \\
    & &&+ \left(-\frac{1}{2}\sum_{v \in V} util_2(v) + \frac{1}{4}\sum_{e'\in E'} cost_2(e') + \eps\sum_{e'\in E'} cost_2(e')\right) \\
    &=&& 1 + \eps \sum_{e \in \bar{E}} cost_1(e) + \frac{10}{\gamma} + \frac{1}{4(b-1)} \cdot \frac{10}{\gamma} \\
    &=&& 1 + \frac{1}{4(b-1)} \cdot 4(b-1) +  \frac{10}{\gamma} + \frac{1}{4(b-1)} \cdot \frac{10}{\gamma}\ \\    
    &\leq&& 3 + \frac{10}{\gamma}.
\end{align*}

We first verify that  $\sum_{e'\in \binom{S}{2}\cap E'} 4 w(e') + \sum_{v\in S} 2 w(v) \leq (1+\gamma) \left(
\sum_{e'=\{v, v'\} \in E'} w(e') + \sum_{v\in 
V} w(v)\right)$.

We have

\begin{align*} \Phi_2(S) = \frac{10}{\gamma} \left(\frac{1}{\sum_{e'=\{v, v'\} \in E'} w(e') + \sum_{v\in 
V} w(v)} \right)
\left(\sum_{e'\in \binom{S}{2}\cap E'} 4 w(e') + \sum_{v\in S} 2 w(v) \right) \leq 3 + \frac{10}{\gamma}\end{align*}

and therefore

\begin{align*}
\sum_{e'\in \binom{S}{2}\cap E'} 4 w(e') + \sum_{v\in S} 2 w(v) 
&\leq \frac{\gamma}{10}\left(3 + \frac{10}{\gamma} \right)\left(
\sum_{e'=\{v, v'\} \in E'} w(e') + \sum_{v\in 
V} w(v)\right) \\
&\leq (1+\gamma) \left(
\sum_{e'=\{v, v'\} \in E'} w(e') + \sum_{v\in 
V} w(v)\right).
\end{align*}

Let us call a bucket $B_i(u)$ bad if $||B_i(u)\cap S|-b/2|\geq b^{0.8}$. Let us call a node $u$ bad if more than $1/b^{0.2}$ fraction of its buckets are bad, and let $U_{good}$ be the set of nodes that are not bad in this sense. For any node $u\in U_{good}$, we can conclude that $|N_{H}(u)\cap S| \geq (1-1/b^{0.2})(1/2-1/b^{0.2}) |N_{H}(u)|\geq (1/2-3/b^{0.2}) |N_{H}(u)| \geq (1/2)(1-\gamma)|N_{H}(u)|$, and also that $|N_{H}(u)\cap S| \leq ((1-1/b^{0.2})(1/2+1/b^{0.2})+(1/b^{0.2}))|N_{H}(u)| \leq (1/2+3/b^{0.2}) |N_{H}(u)| \leq (1/2)(1+\gamma)|N_{H}(u)|$.  

Recall that
\begin{align*}
\Phi_1(S) &:= &&\frac{4}{\sum_{u \in U} imp_u}\sum_{u \in U} \frac{imp_u}{|N_H(u)|} \cdot \left( \sum_{i=1}^{|N_H(u)|/b}\left( |S \cap B_i(u)| -  b/2\right)^2 \right)
\end{align*}
and that for our chosen set $S$, we have $\Phi_1(S)\leq 3 + \frac{10}{\gamma}$. Notice that for each bad bucket, we have $\left( |S \cap B_i(u)| -  b/2\right)^2 \geq b^{1.6}$. Thus, for each bad node, its contribution to the potential is at least $\frac{4}{\sum_{u \in U} imp_u}imp_{u} b^{0.4}$. Hence, we know that $\sum_{u\in U\setminus U_{good}} imp_{u} \leq (1/(b^{0.4}))\left(3 + \frac{10}{\gamma}\right) \sum_{u\in U} imp_{u} \leq (1/(b^{0.2}))\sum_{u\in U} imp_{u}$ and therefore $imp_{u} \geq (1-\gamma) \sum_{u\in U} imp_{u}$. 

It remains to discuss the work and depth of the algorithm. Note that the auxiliary multi-graph that we give as input to \cref{lem:local_rounding} has at most $n$ vertices and $m \cdot b + m'$ edges, and we can construct the multi-graph in $O(mb + m')$ work and $\poly(\log N)$ depth. With the same work and depth, we can also calculate the node utilities and edge costs.
The invocation \cref{lem:local_rounding} takes $(mb + m' + n)\poly(\log \log N)$ work and $\poly(\log N/\eps)$ depth.
As $b := \lceil (1/\gamma)^6\rceil$ and $\eps = 1/(4(b-1))$, we can conclude that the overall work is 
 at most $(m+n + m')\poly(\log \log N)/\gamma^{30}$ and the depth is $\poly(\log N)$.

\end{proof}
\subsection{Wrap Up}
\label{subsec:MISwrap}
We are now ready to present the proof of \Cref{lem:coreMIS}. For convenience, we first restate the lemma.
\MIShittingSet*

\begin{proof}[Proof of \Cref{lem:coreMIS}] 
Let $K :=\lceil 100\log \log N\rceil$. 
Let $H^{(low)} = H[U^{(low)} \sqcup V^{(low)}]$ and $G'^{(low)} = G'[V^{(low)}]$ with

\[V^{(low)} = \{v \in V \colon k_v \in \{K+ 1, K+2,\ldots, \lceil \log(N)\rceil\}\}\] 

and

\[U^{(low)} = \{u \in U \colon |N_H(u) \cap V^{(low)}| \geq \lceil 10 \log^{25} (N) \rceil\}.\]

We also define for every $\{v,v'\} \in E(G'^{(low)})$
\[w^{(low)}(\{v,v'\})= w(\{v,v'\})\]

and for ever $v \in V^{(low)}$, we define

\[w^{(low)}(v) = \sum_{v' \in N_{G'}(v) \setminus V^{(low)}} w(\{v,v'\}) \cdot 2^{-k_{v'}}.\]

We invoke \cref{lem:mis_set_low_probability} with input $H^{(low)}$, $G'^{(low)}$, and edge and vertex weights $w^{(low)}$ and as a result we obtain two subsets $S^{(low)} \subseteq V^{(low)}$ and $U^{(low)}_{good} \subseteq U^{(low)}$ such that 

\begin{itemize}
    \item $\sum_{U^{(low)}_{good}} imp_u \geq 0.99\sum_{u \in U^{(low)}} imp_u$
    \item for every $u \in U^{(low)}_{good}$, $|N_H(u) \cap S^{(low)}| \cdot 2^{-K} \in (\sum_{v \in N_H(u) \cap V^{(low)}} 2^{-k_v}) \pm 0.01$
    \item \begin{align*}&\sum_{e'\in \binom{S^{(low)}}{2}\cap E'} w^{(low)}(e') 2^{-2K} + \sum_{v\in S^{(low)}} w^{(low)}(v) 2^{-K} \\
    &\leq 2 \cdot \left( \sum_{e'=\{v, v'\} \in \binom{V^{(low)}}{2} \cap E'} w^{(low)}(e') 2^{-(k_v+k_{v'})} + \sum_{v\in V^{(low)}} w^{(low)}(v) 2^{-k_v} \right) \end{align*}
\end{itemize}

Let $H^{(high)} = H[U^{(high)} \sqcup V^{(high)}]$ and  $G'^{(high)} = G'[V^{(high)}]$ with

\[V^{(high)} = (V \setminus V^{(low}) \cup S^{(low)} \] 

and

\[U^{(high)} = U^{(low)}_{good} \cup (U \setminus U^{(low)}).\]

We also define for every $v \in V^{(high)}$,

\[k^{(high)}_v = \min(k_v, K).\]

Before invoking \cref{lem:mis_high_probability}, we have to verify that for every $u \in U^{(high)}$, it holds that

\[\sum_{v \in N_H(u) \cap V^{(high)}} 2^{- \min(k_v,K)} \leq 40.\]

First, consider the case that $u \in U \setminus U^{(low)}$, and therefore $|N_H(u) \cap V^{(low)}| < \lceil 10\log^{25}(N)\rceil$. We get

\begin{align*}
    \sum_{v \in N_H(u) \cap V^{(high)}} 2^{- \min(k_v,K)} \leq \sum_{v \in N_H(u)} 2^{- k_v} + |N_H(u) \cap V^{(low)}| \cdot 2^{-K} \leq 10 + \frac{ \lceil 10\log^{25}(N)\rceil}{2^{K}} \leq 40.
\end{align*}

Next, consider the case that $u \in U^{(low)}_{good}$. We get

\begin{align*}
 \sum_{v \in N_H(u) \cap V^{(high)}} 2^{- \min(k_v,K)} 
 &= \sum_{v \in N_H(u) \setminus V^{(low)}} 2^{-k_v} + |N_H(u) \cap S^{(low)}| 2^{-K} \\
 &\leq \sum_{v \in N_H(u) \setminus V^{(low)}} 2^{-k_v} + \sum_{v \in N_H(u) \cap V^{(low)}} 2^{-k_v} + 0.01 \\
 &\leq 10 + 0.01 \leq 40. 
\end{align*}

Thus, we can now invoke \cref{lem:mis_high_probability} and as a result we obtain two subsets $S \subseteq V^{(high)}$ and $U_{good} \subseteq U^{(high)}$ such that 

\begin{itemize}
\item $\sum_{u \in U_{good}} imp_u \geq 0.99\sum_{u \in U^{(high)}} imp_u$ \\
\item  for each $u \in U_{good}$,  we have 
\begin{align*} &|N_H(u) \cap S| \in \\
 &[ (\sum_{v \in N_H(u) \cap V^{(high)}}2^{-\min(k_v, K)}) - 0.1, (\sum_{v \in N_H(u) \cap V^{(high)}}2^{-\min(k_v, K)}) + 10^{25}]
 \end{align*}
\item $\sum_{e'\in \binom{S}{2}\cap E'} w(e') \leq C_{L\ref{lem:mis_high_probability}} \cdot \sum_{e'=\{v, v'\} \in  \binom{V^{(high)}}{2}\cap E'} w(e') \cdot 2^{-(k^{(high)}_v+k^{(high)}_{v'})}$
\end{itemize}
for some constant $C_{L\ref{lem:mis_high_probability}} > 0$.

We start to prove the third bullet point. We have

\begin{align*}
& &&\sum_{e'=\{v, v'\} \in  \binom{V^{(high)}}{2}\cap E'} w(e') \cdot 2^{-(k^{(high)}_v+k^{(high)}_{v'})} \\
&=&&\sum_{e'=\{v, v'\} \in  \binom{V \setminus V^{(low)}}{2}\cap E'} w(e') \cdot 2^{-(k^{(high)}_v+k^{(high)}_{v'})} \\
& &&+ \sum_{e'=\{v, v'\} \in  \binom{S^{(low)}}{2}\cap E'} w(e') \cdot 2^{-(k^{(high)}_v+k^{(high)}_{v'})} \\
& &&+ \sum_{v \in S^{(low)}}\sum_{v' \in N_{G'}(v) \setminus V^{(low)}} w(e') \cdot 2^{-(k^{(high)}_v+k^{(high)}_{v'})}\\
&=&&\sum_{e'=\{v, v'\} \in  \binom{V \setminus V^{(low)}}{2}\cap E'} w(e') \cdot 2^{-(k_v+k_{v'})} \\
& &&+ \sum_{e'=\{v, v'\} \in  \binom{S^{(low)}}{2}\cap E'} w(e') \cdot 2^{-2K} \\
& &&+ \sum_{v \in S^{(low)}}\sum_{v' \in N_{G'}(v) \setminus V^{(low)}} w(e') \cdot 2^{-k_{v'}} \cdot 2^{-K} \\
&=&&\sum_{e'=\{v, v'\} \in  \binom{V \setminus V^{(low)}}{2}\cap E'} w(e') \cdot 2^{-(k_v+k_{v'})} \\
&&&+\sum_{e'\in \binom{S^{(low)}}{2}\cap E'} w^{(low)}(e') 2^{-2K} + \sum_{v\in S^{(low)}} w^{(low)}(v) 2^{-K} \\
&\leq&& \sum_{e'=\{v, v'\} \in  \binom{V \setminus V^{(low)}}{2}\cap E'} w(e') \cdot 2^{-(k_v+k_{v'})} \\
&&&+ 2 \cdot \left( \sum_{e'=\{v, v'\} \in \binom{V^{(low)}}{2} \cap E'} w^{(low)}(e') 2^{-(k_v+k_{v'})} + \sum_{v\in V^{(low)}} w^{(low)}(v) 2^{-k_v} \right) \\
&\leq && 2 \cdot \sum_{e'=\{v, v'\} \in   E'} w(e') \cdot 2^{-(k_v+k_{v'})}
\end{align*}

Therefore,

\begin{align*}
\sum_{e'\in \binom{S}{2}\cap E'} w(e') &\leq C_{L\ref{lem:mis_high_probability}} \cdot \sum_{e'=\{v, v'\} \in  \binom{V^{(high)}}{2}\cap E'} w(e') \cdot 2^{-(k^{(high)}_v+k^{(high)}_{v'})} \\
&\leq (C_{L\ref{lem:mis_high_probability}} \cdot 2) \cdot \sum_{e'=\{v, v'\} \in   E'} w(e') \cdot 2^{-(k_v+k_{v'})}.
\end{align*}

Note that

 \begin{align*}
 \sum_{u \in U_{good}} imp_u 
 &\geq 0.99 \sum_{u \in U^{(high)}} imp_u  \\
 &= 0.99 \left(\sum_{u \in U^{(low)}_{good}} imp_u + \sum_{u \in U \setminus U^{(low)}} imp_u \right) \\
 &\geq 0.99 \left(0.99 \cdot\sum_{u \in U^{(low)}} imp_u + \sum_{u \in U \setminus U^{(low)}} imp_u \ \right) \\
 &\geq 0.9 \sum_{u \in U} imp_u.
 \end{align*}

Thus, it remains to verify that for every $u \in U_{good}$, $|N_H(u) \cap S| \in [1,10^{30}]$. First, let's consider a $u \in U_{good} \cap U^{(low)}_{good}$. We have

\begin{align*}
   |N_H(u) \cap S| 
   &\leq \sum_{v \in N_H(u) \cap V^{(high)}}2^{-\min(k_v, K)} + 10^{25} \\
   &= \sum_{v \in N_H(u) \setminus V^{(low)}} 2^{-k_v} + |N_H(u) \cap S^{(low)}| 2^{-K}  + 10^{25}\\
 &\leq \sum_{v \in N_H(u) \setminus V^{(low)}} 2^{-k_v} + \sum_{v \in N_H(u) \cap V^{(low)}} 2^{-k_v} + 0.01  + 10^{25} \\
 &\leq 10^{30}.
\end{align*}

and

\begin{align*}
   |N_H(u) \cap S| 
   &\geq \sum_{v \in N_H(u) \cap V^{(high)}}2^{-\min(k_v, K)} -0.1 \\
   &= \sum_{v \in N_H(u) \setminus V^{(low)}} 2^{-k_v} + |N_H(u) \cap S^{(low)}| 2^{-K} -0.1\\
 &\geq \sum_{v \in N_H(u) \setminus V^{(low)}} 2^{-k_v} + \sum_{v \in N_H(u) \cap V^{(low)}} 2^{-k_v} - 0.01  - 0.1 \\
 &\geq 5 - 0.01 - 0.1 \\
 &\geq 1.
\end{align*}

Next, consider a $u \in U_{good} \setminus U^{(low)}$.
We have

\begin{align*}
 |N_H(u) \cap S| 
   &\leq \sum_{v \in N_H(u) \cap V^{(high)}}2^{-\min(k_v, K)} + 10^{25}\\ 
   &\leq \sum_{v \in N_H(u)} 2^{- k_v} + |N_H(u) \cap V^{(low)}| \cdot 2^{-K}  + 10^{25}\\
   &\leq 10 + \frac{ \lceil 10\log^{25}(N)\rceil}{2^{K}} +10^{25}\leq 10^{30}
\end{align*}

and

\begin{align*}
 |N_H(u) \cap S| 
   &\geq \sum_{v \in N_H(u) \cap V^{(high)}}2^{-\min(k_v, K)} -0.1 \\
   &\geq \sum_{v \in N_H(u) \setminus V^{(low)}} 2^{-k_v} - 0.1 \\
   &\geq \sum_{v \in N_H(u)} 2^{-k_v} - 2^{-K}|N_H(u) \cap V^{(low)}| - 0.1 \\
   &\geq 5 - 2^{-K} \cdot \lceil 10\log^{25}(N)\rceil -0.1 \\
   &\geq 1.
\end{align*}

Both the invocation of \cref{lem:mis_set_low_probability} and \cref{lem:mis_high_probability} take $(m+n+m') \poly(\log \log N) + \poly(\log N)$ work and $\poly(\log N)$ depth, so the overall work is $(m+n+m') \poly(\log \log N) + \poly(\log N)$  and the overall depth is $\poly(\log N)$.
\end{proof}

\bibliographystyle{alpha}
\bibliography{ref}
\appendix
\section{A lemma on iterations of multiplicative and additive losses}
During our gradual rounding processes, per iteration, we incur multiplicative and additive losses. Iterating these can result in an untidy loss bound. The next lemma, which has a simple proof, allows us to state the aggregate loss effect in a simpler way (with one factor coming from all the multiplicative terms and an additive term from additive losses, both relaxed by a $2$ factor). 

\begin{lemma}
\label{lem:iterativeLoss}
    Consider $\gamma^{(i)}\geq 0$ for $i\in [1, L]$ such that $\sum_{i=1}^{L} \gamma^{(i)} \leq 1/2$. Consider an arbitrary value $Z\geq 0$, and define a recursive function as follows: set $f(Z, 0)=Z$, and for all $i\geq 1$, set $f(Z, i) = (1+\gamma^{(i)})f(Z, i-1) + \gamma^{(i)}$. Then, for all $i\in [0, L]$, we have
    $$f(Z, i) \leq g(Z, i):= 
    Z\prod_{j=1}^{i}(1+2\gamma^{(j)}) + \sum_{j=1}^{i} 2\gamma^{(j)}.$$ 

    Similarly, define another recursive function as follows: set $f'(Z, 0)=Z$, and for all $i\geq 1$, set $f'(Z, i) = (1-\gamma^{(i)})f'(Z, i-1) - \gamma^{(i)}$. Then, for all $i\in [0, L]$, we have
    $$f'(Z, i) \geq g'(Z, i):= 
    Z\prod_{j=1}^{i}(1-2\gamma^{(j)}) - \sum_{j=1}^{i} 2\gamma^{(j)}.$$
\end{lemma}
\begin{proof} We first prove, by induction on $i$, that $f(Z, i) \leq g(Z, i)$. For the base case of $i=0$, we have $f(Z, 0) = g(Z, 0) = Z$. For the inductive step of $i\geq 1$, suppose $f(Z, i-1) \leq g(Z, i-1)$. We prove that $f(Z, i) \leq g(Z, i)$. In particular, we have 
\begin{align*}
f(Z, i) 
&= &&(1+\gamma^{(i)})f(Z, i) + \gamma^{(i)} \\ 
&\leq &&(1+\gamma^{(i)})g(Z, i) + \gamma^{(i)} \\
&=&&(1+\gamma^{(i)}) \left(Z \prod_{j=1}^{i-1}(1+2\gamma^{(j)}) + \sum_{j=1}^{i-1} 2\gamma^{(j)}\right) + \gamma^{(i)}\\
&=&& Z(1+\gamma^{(i)})\prod_{j=1}^{i-1}(1+2\gamma^{(j)}) + \sum_{j=1}^{i-1} 2\gamma^{(j)} + \gamma^{(i)} \left(1+\sum_{j=1}^{i-1} 2\gamma^{(j)}\right) \\
&\leq && Z(1+\gamma^{(i)})\prod_{j=1}^{i-1}(1+2\gamma^{(j)}) + \sum_{j=1}^{i-1} 2\gamma^{(j)} + 2\gamma^{(i)} \\
& \leq  && Z\prod_{j=1}^{i}(1+2\gamma^{(j)}) + \sum_{j=1}^{i} 2\gamma^{(j)} \\
&=&& g(Z, i).
\end{align*} 

We next prove, again by induction on $i$, that $f'(Z, i)\geq g'(Z, i)$. For the base case of $i=0$, we have $f'(Z, 0) = g'(Z, 0) = Z$. For the inductive step of $i\geq 1$, suppose $f'(Z, i-1) \geq g'(Z, i-1)$. We prove that $f'(Z, i) \geq g'(Z, i)$. In particular, we have

\begin{align*}
f'(Z, i) 
&= &&(1-\gamma^{(i)})f'(Z, i) - \gamma^{(i)} \\ 
&\geq &&(1-\gamma^{(i)})g'(Z, i) - \gamma^{(i)} \\
&=&&(1-\gamma^{(i)}) \left(Z \prod_{j=1}^{i-1}(1-2\gamma^{(j)}) - \sum_{j=1}^{i-1} 2\gamma^{(j)}\right) - \gamma^{(i)}\\
&=&& Z(1-\gamma^{(i)})\prod_{j=1}^{i-1}(1-2\gamma^{(j)}) - \sum_{j=1}^{i-1} 2\gamma^{(j)} - \gamma^{(i)} \left(1-\sum_{j=1}^{i-1} 2\gamma^{(j)}\right) \\
&\geq && Z(1-\gamma^{(i)})\prod_{j=1}^{i-1}(1-2\gamma^{(j)}) - \sum_{j=1}^{i-1} 2\gamma^{(j)} - 2\gamma^{(i)} \\
& \geq  && Z\prod_{j=1}^{i}(1-2\gamma^{(j)}) - \sum_{j=1}^{i} 2\gamma^{(j)} \\
&=&& g'(Z, i).
\end{align*}
\end{proof}

\end{document}